\documentclass[envcountsame,envcountsect,runningheads]{llncs}

\pdfoutput=1

\usepackage{amssymb}
\usepackage{amsmath}
\usepackage{algorithmic}
\usepackage{enumerate}
\usepackage{graphicx}
\usepackage{floatflt}
\usepackage{wrapfig}
\usepackage{nicefrac}

\newcounter{prooffactcounter}
\newtheorem{prooffact}[prooffactcounter]{Fact}

\newcommand{\mathtext}[1]{\ensuremath{\mathrm{\text{#1}}}}
\newcommand{\nmodels}{\ensuremath{\not\models}}
\newcommand{\prblemname}[1]{\ensuremath{\mathsf{#1}}}
\newcommand{\card}[1]{\left| #1 \right|}
\newcommand{\set}[1]{\ensuremath\left\{#1\right\}}
\newcommand{\var}{\mathtext{VAR}}
\newcommand{\subf}[1]{\ensuremath{\mathsf{sf}\left(#1\right)}}
\newcommand{\md}[1]{\mathsf{md}\left(#1\right)}
\newcommand{\sat}[1]{#1$-$\textrm{\sf{SAT}}}

\newcommand{\decisionproblem}[3]{
\medskip
\vspace*{1mm}
\begin{tabular}{ll}
\textit{Problem:} & #1 \\
\textit{Input:} & #2 \\
\textit{Question:} & #3
\end{tabular}
\smallskip
\vspace*{1mm}
}

\newcommand{\complexityclassname}[1]{\ensuremath{\mathrm{#1}}}

\newcommand{\PSPACE}{\complexityclassname{PSPACE}}
\newcommand{\NPSPACE}{\complexityclassname{NPSPACE}}
\newcommand{\NP}{\complexityclassname{NP}}

\newcommand{\hvarphi}{\hat\varphi}
\newcommand{\hpsi}{\hat\psi}

\newcommand{\logicname}[1]{\ensuremath{\mathsf{#1}}}
\newcommand{\KL}{\logicname{KL}}
\newcommand{\K}[1]{\ensuremath{\logicname{K}(#1)}}
\newcommand{\genvarphi}[5]{\ensuremath{\hvarphi}^{#1\rightarrow #2}_{w\ge #3,x\ge #4, y\ge #5}}
\newcommand{\path}[3]{\ensuremath{#1\overset{#2}{\longrightarrow}#3}}
\newcommand{\gpath}[3]{\ensuremath{G\models\path{#1}{#2}{#3}}}
\newcommand{\mpath}[3]{\ensuremath{M\models\path{#1}{#2}{#3}}}
\newcommand{\mxdepth}[2]{\ensuremath{\mathtext{\textup{maxdepth}}^{#1}\left(#2\right)}} 
\newcommand{\mxheight}[2]{\ensuremath{\mathtext{\textup{maxheight}}^{#1}\left(#2\right)}} 
\newcommand{\preq}[1]{\ensuremath{\mathsf{\textup{prereq}}\left(#1\right)}}
\newcommand{\conc}[1]{\ensuremath{\mathsf{\textup{conc}}\left(#1\right)}}
\newcommand{\edges}[1]{\ensuremath{\mathsf{edges}\left(#1\right)}}
\newcommand{\vertices}[1]{\ensuremath{\mathsf{vertices}\left(#1\right)}}
\newcommand{\tl}{\ensuremath{\textit{types-list}}}
\newcommand{\tlhom}{\ensuremath{\textit{types-list}}\textrm{-}T^{\mathrm{hom}}_{\hvarphi}}
\newcommand{\algname}{\ensuremath{\text{\sc Horn-Classification}}}
\newcommand{\symm}{\textup{\texttt{symm}}}
\newcommand{\refl}{\textup{\texttt{refl}}}
\newcommand{\trans}{\textup{\texttt{trans}}}
\newcommand{\annot}[1]{\ensuremath{\mathsf{annot}\left(#1\right)}}
\newcommand{\verifyhorn}[1]{\ensuremath{\text{\sc Verify-Horn}\left(#1\right)}}
\newcommand{\verifycons}[1]{\ensuremath{\text{\sc Verify-Consistency}\left(#1\right)}}
\newcommand{\algnametwo}{\ensuremath{\text{\sc Satisfiability}}}

\title{On the Complexity of Elementary Modal Logics\thanks{Supported in part by NSF grants CCR-0311021 and IIS-0713061, a Friedrich Wilhelm Bessel
Research Award, and the DAAD postdoc program.}}

\author{Edith Hemaspaandra \and Henning Schnoor}

\institute{
  Department of Computer Science,
  Rochester Institute of Technology,
  Rochester, NY 14623, U.S.A.
  \email{eh@cs.rit.edu, hs@cs.rit.edu}}

\begin{document}

\maketitle

\begin{abstract}
Modal logics are widely used in computer science. The complexity of modal satisfiability problems has been investigated since the 1970s, usually proving results on a case-by-case basis. We prove a very general classification for a wide class of relevant logics: Many important subclasses of modal logics can be obtained by restricting the allowed models with first-order Horn formulas. We show that the satisfiability problem for each of these logics is either \NP-complete or \PSPACE-hard, and exhibit a simple classification criterion. Further, we prove matching \PSPACE\ upper bounds for many of the \PSPACE-hard logics.
\end{abstract}

\section{Introduction}

Modal logics have proven to be a valuable tool in mathematics and computer science. The traditional uni-modal logic enriches the propositional language with the operator $\Diamond,$ where $\Diamond\varphi$ is interpreted as \emph{$\varphi$ possibly holds}. The usual semantics interpret modal formulas over graphs, where $\Diamond\varphi$ means ``there is a successor world where $\varphi$ is true.'' In addition to their mathematical interest, modal logics are widely used in practical applications: In artificial intelligence, modal logic is used to model the knowledge and beliefs of an agent, see e.g.~\cite{bazh05}. Modal logics also can be applied in cryptographic and other protocols \cite{frhuje02,codofl03,hamotu88,lare86}. For many specific applications, there exist tailor-made variants of modal logics \cite{bega04}.

Due to the vast number of applications, complexity issues for modal logics are very relevant, and have been examined since Ladner's seminal work~\cite{lad77}. Depending on the application, modal logics with different properties are studied. For example, one might want the formula $\varphi\implies\Diamond\varphi$ to be an axiom---if something is true, then it should be considered possible. Or $\Diamond\Diamond\varphi\implies\Diamond\varphi$---if it is possible that $\varphi$ is possible, then $\varphi$ itself should be possible. Classical results~\cite{sahl73} show that there is a close correspondence between modal logics defined by axioms and logics obtained by restricting the class of considered graphs. Requiring the axioms mentioned above corresponds to restricting the classes of graphs to those which are reflexive or transitive, respectively. Determining the complexity of a given modal logic, defined either by the class of considered graphs or via a modal axiom system, has been an active line of research since Ladner's results. In particular, the complexity classes \NP\ and \PSPACE\ have been at the center of attention.

Most complexity results have been on a case-by-case basis, proving results for individual logics both for standard modal logics and variations like temporal or hybrid logics~\cite{hamo92,ngu05,sicl85}. Examples of more general results include Halpern and R{\^e}go's proof that logics including the \emph{negative introspection} axiom, which corresponds to the Euclidean graph property, have an \NP-complete satisfiability problem~\cite{hare07}. In~\cite{schpa06}, Schr\"oder and Pattinson show a way to prove \PSPACE\ upper bounds for modal logics defined by modal axioms of modal depth 1. In~\cite{lad77}, Ladner proved \PSPACE-hardness for all logics for which reflexive and transitive graphs are admissible models. In~\cite{spaan93}, Hemaspaandra showed that all normal logics extending S4.3 have an \NP-complete satisfiability problem, and work on the \emph{Guarded Fragment} has shown that some classes of modal logics can be seen as a decidable fragment of first-order logic~\cite{anbene98}.

While these results give hardness or upper bounds for classes of logics, they do not provide a full case distinction identifying \emph{all} ``easy'' or ``hard'' cases in the considered class. We achieve such a result: For a large class of modal logics containing many important representatives, we identify \emph{all} cases which have an \NP-complete satisfiability problem, and show that the satisfiability problem for \emph{all} other non-trivial logics in that class is \PSPACE-hard. Hence these problems avoid the infinitely many complexity classes between \NP\ and \PSPACE, many of which have natural complete problems arising from logical questions.  To our knowledge, such a general result has not been achieved before.

To describe the considered class of modal logics, note that many relevant properties of modal models can be expressed by first-order formulas: A graph is transitive if its edge-relation $R$ satisfies the clause $\forall xyz \left(x R y\wedge y R z\implies x R z\right)$ and symmetric if it satisfies $\forall xy \left(x R y\implies y R x\right)$. Many other graph properties can be defined using similar formulas, where the presence of a certain pattern of edges in the graph forces the existence of another. Analogously to propositional logic, we call conjunctions of such clauses \emph{universal Horn formulas}. Many relevant logics can be defined in this way: All examples form~\cite{lad77} fall into this category, as well as logics over Euclidean graphs. 

We study the following problem: Given a universal Horn formula $\hpsi,$ what is the complexity of the modal satisfiability problem over the class of graphs defined by $\hpsi$? 

The main results of this paper are the following: First, we identify all cases which give a satisfiability problem solvable in \NP\ (which then for every nontrivial logic is \NP-complete), and show that all other cases are \PSPACE-hard. Second, we prove a generalization of a ``tree-like model property,'' and use it to obtain \PSPACE\ upper bounds for a large class of logics. As a corollary, we prove that Ladner's classic hardness result is ``optimal'' in the class of logics defined by universal Horn formulas. A further corollary is that in the universal Horn class, all logics whose satisfiability problem is not \PSPACE-hard already have the ``polynomial-size model property,'' which is only one of several known ways to prove \NP\ upper bounds for modal logics. 

Various work was done on restricting the syntax of the modal formulas by restricting the propositional operators~\cite{bhss05b}, the nesting degree and number of variables~\cite{hal95} or considering modal formulas in Horn form~\cite{cheli94}. While these results are about restricting the \emph{syntax} of the modal formulas, the current work studies different \emph{semantics} of modal logics, where the semantics are specified by Horn formulas.

The organization of the paper is as follows: In Section~\ref{sect:preliminaries}, we introduce terminology and generalize classic complexity results. Section~\ref{sect:modal models} then establishes techniques to restrict the size of models for modal formulas, which are important tools for the \NP-membership later. Section~\ref{sect:universal horn clauses} contains the main results of the paper about universal Horn formulas. After introducing them in Section~\ref{sect:horn definitions} and proving their relationship to homomorphisms in Section~\ref{sect:horn homomorphism}, we prove \NP-results for special cases of Horn formulas in Sections~\ref{section:varphi k to l section} and~\ref{section:np results}. Using these results, Section~\ref{sect:main dichotomy result} then proves our main dichotomy result, which is Corollary~\ref{corollary:horn cunjunction without algorithm}. The remainder of the paper establishes \PSPACE\ upper bounds for many of the \PSPACE-hard logics. An important tool for these proofs is introduced in Section~\ref{sect:trees for horn}, where we show a tree-like model property for all \PSPACE-hard logics defined by universal Horn formulas. Section~\ref{sect:pspace for non-transitive} contains our \PSPACE-algorithm, which generalizes many previously known algorithms for modal logics. Finally, Section~\ref{sect:applications} obtains a series of corollaries, proving the above-mentioned optimality result for Ladner's hardness condition, and exhibiting a number of cases to which our \PSPACE-algorithm can be applied.  The paper closes with a summary and open questions in Section~\ref{sect:conclusion}.

\section{Preliminaries}\label{sect:preliminaries}

\subsection{Basic concepts and notation}

Modal logic is an extension of propositional logic. A modal formula is a propositional formula using variables, the usual logical symbols $\wedge,\vee,\neg,$ and a unary operator $\Diamond$. (A dual operator $\Box$ is often considered as well, this can be regarded as abbreviation for $\neg\Diamond\neg$.) A model for a modal formula is a set of connected ``worlds'' with individual propositional assignments. To be precise, a \emph{frame} is a directed graph $G=(W,R),$ where the vertices in $W$ are called ``worlds,'' and an edge $(u,v)\in R$ is interpreted as $v$ is ``considered possible'' from $u$. A \emph{model} $M=(G,X,\pi)$ consists of a frame $G=(W,R),$ a set $X$ of propositional variables and a function $\pi$ assigning each variable $x\in X$ a subset of $W,$ the set of worlds in which $x$ is true. We say the model $M$ is \emph{based} on the frame $(W,R)$. If $\mathcal F$ is a class of frames, then a model is an $\mathcal F$-model if it is based on a frame in $\mathcal F$. With $\card M$ we denote the number of worlds in the model $M$.

For a world $w\in W,$ we define when a formula $\phi$ is \emph{satisfied} at $w$ in $M$ (written $M,w\models\phi$.) If $\phi$ is a variable $x,$ then $M,w\models\phi$ if and only if $w\in\pi(x)$. As usual, $M,w\models\phi_1\wedge\phi_2$ if and only if
$M,w\models\phi_1$ and $M,w\models\phi_2$, and $M,w\models\neg\phi$ iff $M,w\nmodels \phi$. For the modal operator, $M,w\models\Diamond\phi$ if and only if there is a world $w'\in W$ such that $(w,w')\in R$ and $M,w'\models\phi$. 

\begin{table}
\begin{center}
\begin{tabular}{lclll}
\hline
\textbf{logic name} & \hspace*{5mm} &  \textbf{graph property} &  \textbf{formula definition} \\
\hline
\logicname{K} & & All graphs & $\K{\hvarphi_{\mathrm{taut}}}$ \\
\logicname{T} & & reflexive & $\K{\hvarphi_{\mathrm{refl}}}$ \\
\logicname{B} & & symmetric & $\K{\hvarphi_{\mathrm{symm}}}$ \\
\logicname{K4} & & transitive graphs & $\K{\hvarphi_{\mathrm{trans}}}$ \\
\logicname{S4} & & transitive and reflexive & $\K{\hvarphi_{\mathrm{trans}}\wedge\hvarphi_{\mathrm{refl}}}$ \\
\logicname{S5} & & equivalence relations & $\K{\hvarphi_{\mathrm{trans}}\wedge\hvarphi_{\mathrm{refl}}\wedge\hvarphi_{\mathrm{symm}}}$ \\
\hline
\end{tabular}
\end{center}
\caption{Common modal logics}
\label{table:common modal logics}
\end{table}

We now describe a way to define classes $\mathcal F$ of frames by propositional formulas. The \emph{frame language} is the first-order language containing (in addition to the propositional operators $\wedge,\vee,$ and $\neg$) the binary relation $R$. The relation $R$ is interpreted as the edge relation in a graph. Semantics are defined in the obvious way, for example, a graph satisfies the formula $\hvarphi_{\mathrm{trans}}:=\forall x,y,z (x Ry)\wedge(y Rz)\implies (x R z)$ if and only if it is transitive. In order to separate modal formulas from first-order formulas, we use $\hat .$ to denote the latter, i.e., $\hat\varphi$ is a first-order formula, while $\phi$ is a modal formula.

A modal logic usually is defined as the set of the formulas provable in it. Since a formula is satisfiable iff its negation is not provable, we can define a logic by the set of formulas satisfiable in it. For a first-order formula $\hvarphi$ over the frame language, we define the logic $\K\hvarphi$ as the logic in which a modal formula $\phi$ is satisfiable if and only if there is a model $M$ and a world $w\in M$ such that the frame which $M$ is based on satisfies the first-order formula $\hvarphi$ (we simply write $M\models\hvarphi$ for this), and $M,w\models\phi$. Such a logic is called \emph{elementary}. In the case that $\hvarphi$ is a universal formula (i.e., every variable in $\hvarphi$ is universally quantified at the beginning of the formula), we call these logics \emph{universal elementary}. In this way, many of the classic examples of modal logics can be expressed: In addition to the formula $\hvarphi_{\mathrm{trans}}$ defined above, let $\hvarphi_{\mathrm{refl}}:=\forall w (w R w),$ and let $\hvarphi_{\mathrm{symm}}:=\forall x,y (x R y)\implies (y R x)$. Finally, let $\hvarphi_{\mathrm{taut}}$ be some tautology over the frame language, for example let $\hvarphi_{\mathrm{taut}}:=\forall x (x R x)\implies (x R x)$. Table~\ref{table:common modal logics} introduces some common modal logics and how they can be expressed in our framework. For a formula $\hvarphi$ over the frame language, we consider the following problem:

\decisionproblem{$\sat{\K\hvarphi}$}{A modal formula $\phi$}{Is $\phi$ satisfiable in a model based on a frame satisfying $\hvarphi$?}

As an example, the problem $\K{\hvarphi_{\mathrm{trans}}}$\prblemname{-SAT} is the problem to decide if a given modal formula can be satisfied in a transitive frame, and therefore is the same as the satisfiability problem for the logic $\logicname{K4}$. It is important to note that in the problem $\sat{\K\hvarphi},$ regard the formula $\hvarphi$ is fixed. It is also interesting to study the \emph{uniform} version of the problem, where we are given a first-order formula $\hvarphi$ over the frame language and a modal formula $\psi,$ and the goal is to determine whether there exists a graph satisfying both. This problem obviously is \PSPACE-hard (this easily follows from Ladner's Theorem~\ref{theorem:ladner}, in fact, the problem is undecidable). In this paper, we study the complexity behavior of \emph{fixed} modal logics.

When interested in complexity results for modal logic, the property of having ``small models'' is often crucial, as these lead to a satisfiability problem in \NP, as long as the class of frames considered is reasonably well-behaved.

\begin{definition}
A modal logic $\KL$ has the \emph{polynomial-size model property}, if there is a polynomial $p,$ such that for every $\KL$-satisfiable formula $\phi,$ there is a \KL-model $M$ and a world $w\in M$ such that $M,w\models\phi,$ and $\card{M}\leq p(\card{\phi})$.
\end{definition}

The following standard observation is the basis of our \NP-containment proofs:

\begin{proposition}\label{proposition:elementary plus polysize model property is in np}
Let $\hvarphi$ be a first-order formula over the frame language, such that $\K\hvarphi$ has the polynomial-size model property. Then $\sat{\K\hvarphi}\in\NP$.
\end{proposition}

\begin{proof}
This easily follows from the literature, since for a given graph and a fixed first-order sentence $\hvarphi,$ it can be checked in polynomial time if the graph satisfies $\hvarphi$. Also, it can be verified in polynomial time if a model satisfies a modal formula. Hence, the obvious guess-and-verify approach works for \NP-containment.
\end{proof}

Since modal logic is an extension of propositional logic, the satisfiability problem for every non-trivial modal logic is \NP-hard. Therefore, proving the polynomial-size model property yields an optimal upper complexity bound for the satisfiability problem for modal logics.

\subsection{Ladner's Theorem and Applications}\label{subsection:ladner}

In the seminal paper~\cite{lad77}, Ladner showed \PSPACE-containment and \PSPACE-hardness for a variety of modal logics. In particular, he proved that the satisfiability problem for any logic between \logicname{K}\ and \logicname{S4}\ is \PSPACE-hard. In order to state Ladner's result, we introduce the concept of extensions of a logic, and how it relates to modal logics defined by first-order formulas.

For a modal logic $\KL,$ an \emph{extension} of $\KL$ is a modal logic $\logicname{KL'}$ such that every formula which is valid (a tautology) in $\KL$ is also valid in $\logicname{KL'},$ or equivalently such that every formula that is $\logicname{KL'}$-satisfiable is also $\KL$-satisfiable. As an example, every logic that we consider is an extension of $\logicname{K},$ and $\logicname{S4}$ is an extension of $\logicname{K4}$. In the case of elementary logics, this is related to an implication of the corresponding first-order-formulas. In the following, when we say that a formula $\hvarphi$ over the frame language implies a formula $\hpsi$ over the frame language, then we mean that every graph which is a model of $\hvarphi$ also satisfies $\hpsi$.

\begin{proposition}\label{proposition:extensions of elementary logics and first-order implication}
 Let $\hvarphi$ and $\hpsi$ be first-order formulas over the frame language, and let $\K\hvarphi$ and $\K\hpsi$ be the corresponding elementary modal logics. If $\hvarphi$ implies $\hpsi,$ then $\K\hvarphi$ is an extension of $\K\hpsi$.
\end{proposition}

\begin{proof}
 Let $\phi$ be a modal formula which is valid in $\K\hpsi$. Then $\neg\phi$ is not $\K\hpsi$-satisfiable. Now assume that $\neg\phi$ is $\K\hvarphi$-satisfiable. Then there exists a model $M$ and a world $w\in M$ such that $M,w\models\neg\phi$ and $M$ is a $\K\hvarphi$-model, i.e., $M\models\hvarphi$. Since $\hvarphi$ implies $\hpsi,$ we know that $M\models\hpsi,$ and therefore $M$ is a $\K\hpsi$-model. Therefore, $\neg\phi$ is $\K\hpsi$-satisfiable, a contradiction. Therefore we know that $\neg\phi$ is not $\K\hvarphi$-satisfiable, and hence $\phi$ is valid in $\K\hvarphi$. Therefore, it follows that $\K\hvarphi$ is an extension of $\K\hpsi$.
\end{proof}

Note that the converse of Proposition~\ref{proposition:extensions of elementary logics and first-order implication} does not hold. For example, consider the formulas $\hvarphi_1=\exists x \overline{(x R x)},$ and $\hvarphi_2=\forall x (x R x)\vee\overline{(x R x)}$. Then $\hvarphi_2$ is a tautology and $\hvarphi_1$ is not, in particular, we know that $\hvarphi_2$ does not imply $\hvarphi_1$. But the logics $\K{\hvarphi_1}$ and $\K{\hvarphi_2}$ are easily seen to be identical (and both identical to $\logicname{K}$). In particular, they are extensions of each other.

Ladner's main result can be stated as follows:

\begin{theorem}[\cite{lad77}]\label{theorem:ladner}
 \begin{enumerate}
  \item The satisfiability problems for the logics $\logicname{K},\logicname{K4},$ and $\logicname{S4}$ are \PSPACE-complete, and $\sat{\logicname{S5}}$ is \NP-complete.
  \item Let $\KL$ be a modal logic such that \logicname{S4} is an extension of \KL. Then $\sat{\KL}$ is \PSPACE-hard.
 \end{enumerate}
\end{theorem}

Ladner's proof for Theorem~\ref{theorem:ladner} shows some additional results, which we will give as a series of corollaries. 

The construction from Ladner's proof can be modified to prove the following result as well:

\begin{corollary}\label{corollary:ufos satisfied in symmetric tree are pspace hard}
 Let $\hvarphi$ be a formula over the frame language which is satisfied in every symmetric tree or in every reflexive and symmetric tree. Then $\sat{\K{\hvarphi}}$ is \PSPACE-hard.
\end{corollary}

Even though the proof is just a minor variation of Ladner's proof, we give the entire construction for completeness. We mention where the adjustments for the symmetric case are.

\begin{proof}
 The result follows from a slight modification of Ladner's proof for the hardness result in Theorem~\ref{theorem:ladner}. We follow the presentation of~\cite{blrive01}, where Ladner's theorem can be found as Theorem~6.50.

The proof shows a reduction from the evaluation problem for quantified Boolean formulas, \prblemname{QBF}. The main strategy is to create from a quantified formula $\chi$ a modal formula $\phi$ such that in any satisfying model for $\phi,$ a complete ``quantifier tree'' for the quantifier block of $\chi$ can be found, i.e., a tree where for each existentially quantified variable, one value is chosen, and for universally quantified variables, both alternatives true and false are evaluated.

Let $\chi=Q_1 p_1\dots Q_m p_m\theta(p_1,\dots,p_m)$ be a quantified Boolean formula, where $Q_i\in\set{\forall,\exists},$ and $\theta$ is a propositional formula. From this we construct a modal formula $\phi,$ in which variables $p_1,\dots,p_m$ and $q_1,\dots,q_m$ appear. The $p_i$ correspond directly to the variables of the propositional formula, and the $q_i$ mark the level of the node in the modal model: As mentioned, a model $M$ for the formula $\phi$ is essentially a quantifier tree, and hence each node has a unique level in $M$. The construction will ensure that (up to the depth of the model that we care about) $q_i$ is true in a state if and only if the state is on level $i$ in the model. The variables $p_i$ express the values of the variables in the quantifier tree, where the value of $p_i$ is only regarded as ``defined'' from level $i$ onwards (obviously, these variables also have values in other levels, but the values in these levels are not of interest to us).

 For constructing the formula which forces the quantifier tree, we first introduce two macros. The macro $B_i$ requires the model to ``split'' at level $i:$

$$B_i:=q_i\rightarrow\Diamond(q_{i+1}\wedge p_{i+1})\wedge\Diamond(q_{i+1}\wedge \overline{p_{i+1}}).$$

The effect of the formula $B_i$ is that if it is required to be true in level $i,$ then each node in level $i$ must have two different successors in the next level, setting the variable $p_{i+1}$ to true in one of them, and to false in the other. Therefore we can use this macro to force the ``branching'' of the quantifier tree in the levels corresponding to universal variables: if $p_{i+1}$ is a universally quantified variable, then the macro $B_i$ ensures that both possible truth values for $p_{i+1}$ are evaluated.

When forcing the quantifier tree, we also need to ensure that truth values for the variables $p_i$ are properly propagated to ``lower levels'' in the tree. For this, we use the following macro (and this is the only point in which our construction differs from the proof given in~\cite{blrive01}):

$$S_i:= ((p_i\rightarrow\Box((q_{i+1}\vee\dots\vee q_m)\rightarrow p_i))\wedge (\overline{p_i}\rightarrow\Box((q_{i+1}\vee\dots\vee q_m)\rightarrow \overline{p_i}))).$$

This formula forces the truth value of $p_i$ to be propagated ``down'' the tree from those levels on where we actually regard the value of $p_i$ as set, i.e., from the level $i$ on. Note that a node $v$ can only have a successor in which $(q_{i+1}\vee\dots\vee q_m)$ holds if its own level is at least $i,$ and hence we restrict the values of $p_i$ in exactly those parts of the tree where it is regarded as defined.

We now give the construction of the formula $\varphi,$ which as mentioned is identical to the one used in Ladner's proof, with the exception that our macro $S_i$ is different. The formula $\varphi$ is the conjunction of the following formulas (here, $\Box^m\psi$ is an abbreviation for a $\underbrace{\Box\Box\dots\Box}_m\psi,$ and $\Box^{(m)}\psi$ is a shorthand for $\psi\wedge\Box^1\psi\wedge\Box^2\psi\wedge\dots\wedge\Box^m\psi$):

\medskip

\begin{tabular}{llllllllll}
 $(i)$ & $q_0$ \\
 $(ii)$ & $\Box^{(m)}(q_i\rightarrow\bigwedge_{i\neq j}\overline{q_j})$ \hspace*{5mm}  $(0\leq i\leq m)$ \\
 $(iiia)$ & $\Box^{(m)}(q_i\rightarrow\Diamond q_{i+1})$ \hspace*{5mm} $(0\leq i<m)$ \\
 $(iiib)$ & $\bigwedge_{\set{i\ \vert\ Q_i=\forall}}\Box^i B_i$ \\
$(iv)$ & $\Box^{(m)}S_i$ \hspace*{5mm} $(0,\leq i<m)$ \\
$(v)$ & $\Box^m(q_m\rightarrow\theta )$
\end{tabular}

\medskip

The construction works as follows: With formula $(i),$ we give the start of the model and define the world $w$ in which $\varphi$ is satisfied to have the level $0$. Formula $(ii)$ requires each node which is reachable in at most $m$ steps from $w$ to have a well-defined level (or none). Formulas $(iiia)$ and $(iiib)$ require each level $i$ to be followed by level $i+1,$ and in the case that the $i$th quantifier is $\forall,$ formula $(iiib)$ requires the corresponding branching of the quantifier tree. Formula $(iv)$ then forces the truth values of the $p_i$ to be ``sent down'' the tree, as described earlier. Formula $(v)$ finally requires that the propositional formula $\theta$ is true for all possible truth assignments to $p_1,\dots,p_m$ generated by the quantifier tree. 
\end{proof}

Ladner's construction for the \PSPACE\ upper bound for the logic $\logicname{K}$ also reveals the following. The main idea behind this corollary is that any model for some modal formula can be transformed into a strict tree by ``unrolling.'' A much more general version of this result will be proven later as Theorem~\ref{theorem:tree-like models for horn logics}.

\begin{corollary}\label{corollary:every k-satisfiable formula satisfiable in strict tree}
 A modal formula is \logicname{K}-satisfiable if and only if it can be satisfied in a strict tree.
\end{corollary}

Corollary~\ref{corollary:every k-satisfiable formula satisfiable in strict tree} shows that we only need to apply our formulas to trees---as long as our first order formula does not ``say anything'' about trees, the generated logic is the same as \logicname{K}, although these logics do not necessarily have the same set of models.

\begin{corollary}\label{corollary:formulas satisfiable in trees give pspace complete logics}
 Let $\hvarphi$ be a first-order formula over the frame language such that $\hvarphi$ is satisfied in every strict tree. Then the satisfiability problem for $\K\hvarphi$ is \PSPACE-complete.
\end{corollary}

\begin{proof}
This follows immediately from Corollary~\ref{corollary:every k-satisfiable formula satisfiable in strict tree} and Theorem~\ref{theorem:ladner}, since every modal formula is $\logicname{K}$-satisfiable if and only it is $\K\hvarphi$-satisfiable. 
\end{proof}

For our complete classification in Theorem~\ref{theorem:horn conjunction classification}, we need a hardness result which is a slight variation of Ladner's hardness result for all logics between $\logicname{K}$ and $\logicname{S4}$. However, the proof is merely a closer inspection of Ladner's construction. 

In addition to the already mentioned graph properties, we define a generalization of transitivity. For a natural number $k,$ we say that a graph $G$ is $k$-transitive, if every pair of vertices $(u,v)$ in $G$ such that there is a $k$-step path from $u$ to $v$ in $G$ is connected with an edge. It is easy to see that a graph is transitive if and only if it is $2$-transitive, and a transitive graph is also $k$-transitive for every $k\in\mathbb{N}$. For a set $S\subseteq\mathbb{N},$ we say that a graph is $S$-transitive if it is $k$-transitive for every $k\in S$.

\begin{theorem}\label{theorem:ladner hardness cases for k-transitivity}
 Let $\hpsi$ be a first-order formula over the frame language such that one of the following cases applies:
\begin{itemize}
 \item $\hpsi$ is satisfied in every strict tree,
 \item $\hpsi$ is satisfied in every reflexive tree,
 \item there is a set $S\subseteq\mathbb{N}$ such that $\hpsi$ is satisfied in every $S$-transitive tree,
 \item $\hpsi$ is satisfied in every symmetric tree,
 \item $\hpsi$ is satisfied in every tree which is both reflexive and symmetric,
 \item there is a set $S\subseteq\mathbb{N}$ such that $\hpsi$ is satisfied in every tree which is both reflexive and $S$-transitive.
\end{itemize}

Then $\sat{\K\hpsi}$ is $\PSPACE$-hard.
\end{theorem}

Note that for the cases not including symmetry, the result ``almost'' follows directly from Ladner's Theorem~\ref{theorem:ladner}. It does not follow directly, since we only require that our first-order formulas are satisfied in $S$-transitive \emph{trees}, but nothing about arbitrary $S$-transitive graphs. Nevertheless, the result for this case follows directly from Ladner's proof.

\begin{proof}
First consider the cases in which symmetry does not occur. Now the result follows directly from Ladner's proof: Ladner proves \PSPACE-hardness with a reduction from \prblemname{QBF}. From a quantified Boolean formula $\chi,$ he constructs a modal formula $\phi$ such that the following holds:

\begin{enumerate}
 \item If $\chi$ is true, then there is a modal model $M$ and a world $w\in M$ such that $M$ is a strict tree, and $M,w\models\phi$. Moreover, adding any number of reflexive or transitive edges in $M$ preserves the fact that $M,w\models\phi$.
 \item If $\chi$ is false, then $\phi$ is not $\logicname{K}$-satisfiable.
\end{enumerate}

It is immediate that this reduction also proves the desired hardness result for $\K\hpsi:$ First assume that $\chi$ is true. Then, by the above, the model $M,w$ satisfies $\phi,$ and is a tree. If we close this model under the reflexive and/or $S$-transitive closures, then by the above this still is a model for $\phi$. Since $\hpsi$ is satisfied in every reflexive and/or $S$-transitive tree, this is a $\K\hpsi$-model for $\phi,$ and hence $\phi$ is $\K\hpsi$-satisfiable. On the other hand, if $\chi$ is false, then $\phi$ is not $\logicname{K}$-satisfiable, and therefore not $\K\hpsi$-satisfiable.

Now the cases involving symmetry follow from Corollary~\ref{corollary:ufos satisfied in symmetric tree are pspace hard}.
\end{proof}

As this section indicated, trees are an important subclass of modal models. And in fact, a good intuition to read this paper is to always think of the graphs we deal with as ``near-trees,'' i.e., trees with additional edges.

\section{About Modal Models}\label{sect:modal models}

\NP-results in this paper are shown with the explicit construction of small models for a given formula. We therefore introduce some notation on graphs. For a graph $G,$ the set of vertices of $G$ is denoted with $\vertices G,$ and $\edges G$ is the set of its edges. As usual, a \emph{homomorphism} from a graph $G_1$ to a graph $G_2$ is a function preserving the edge relation. A \emph{strict tree} is a tree in the usual sense, i.e., a directed, acyclic, connected graph which has a root $w$ from which all other vertices can be reached. We now define notation to describe paths in graphs. Note that we often identify modal models and their frames, when the propositional assignments are clear from the context or not important for our arguments.

\begin{definition}
Let $G$ be a graph, $w,v\in G$ vertices, and $i\in\mathbb{N}$. We write $G\models\path{w}{i}{v}$ if in $G,$ there is a path of length $i$ from $w$ to $v$. Additionally, $\gpath w0w$ for all $w\in G$. We also say that $w$ is a $i$-step predecessor of $v,$ and $v$ is a $i$-step successor of $w$ in $G,$ if $\gpath{w}{i}{v}$. The \emph{maximal depth} of $w\in G$ is defined as $\mxdepth{G}{w}:=\max\set{i\ \vert\ \exists w'\in G,\gpath {w'}{i}{w}}$. Similarly, the \emph{maximal height} of $w\in G$ is $\mxheight{G}{w}:=\max\set{i\ \vert\ \exists w'\in G,\gpath {w}{i}{w'}}$.
\end{definition}

Note that the maximal depth and maximal height of nodes can be (countably) infinite, even in finite graphs. The next definition is a restriction on graphs which is very natural for modal logics: for deciding whether $M,w\models\phi$ holds for some modal $M,$ a world $w\in M,$ and a modal formula $\phi,$ it is obvious that only the worlds which are reachable from $w$ are important.

\begin{definition}
Let $G$ be a graph, and let $w\in G$. The graph $G_w$ is obtained from $G$ by restricting $G$ to the worlds which can be reached from $w$.
\end{definition}

\subsection{Invariants}

Proving the polynomial-size model property for some logic is usually done starting with an arbitrary model for a given modal formula and building a smaller model out of it, which still satisfies the modal formula. However, we also need to ensure that the new model still satisfies the conditions of the logic under consideration. Therefore, we need results that allow us to perform modifications on our models and leave the modal and the first-order properties invariant. The first result in this way concerns the first-order aspect of the frames for universal formulas: 

\begin{theorem}\label{theorem:ufo invariants}
Let $G$ and $G'$ be graphs, and let $n$ be a natural number. The following conditions are equivalent:
\begin{enumerate}
\item Every universal first-order formula (with at most $n$ variables) satisfied by $G'$ is also satisfied by $G,$
\item Every existential first-order formula (with at most $n$ variables) satisfied by $G$ is also satisfied by $G',$
\item For each finite $V\subseteq G$ (such that $\card{V}\leq n$),  there exists a functions $f\colon V\rightarrow G'$ such that for $u,v\in V,$ it holds that $u R v$ iff $f(u) R' f(v)$.
\end{enumerate}
\end{theorem}

\begin{proof}
\begin{description}
\item[$1\leftrightarrow 2$]{Let $\exists x_1\dots\exists x_n\hvarphi(x_1,\dots,x_n)$ be an existential first-order formula over the frame language which holds in $G,$ and assume that it does not hold in $G'$. In this case, the negation of the formula holds in $G',$ i.e., $G'\models\overline{\exists x_1\dots\exists x_n\hvarphi(x_1,\dots,x_n)}$. This is equivalent to $G'\models\forall x_1\dots\forall x_n\overline{\hvarphi(x_1,\dots,x_n)},$ which is a universal first-order formula over the frame language. Hence by the prerequisites, we know that $G\models\forall x_1\dots\forall x_n\overline{\hvarphi(x_1,\dots,x_n)},$ which is a contradiction.
}
\item[$2\rightarrow 1$]{Analogously to the above.}
\item[$2\rightarrow 3$]{We construct a function $f\colon V\rightarrow G'$ with the desired properties. Let $\card{V}=n,$ and let $V=\set{x_1,\dots,x_n}$. We construct an existential first-order formula

$$\displaystyle \hpsi=\exists x_1,\dots,\exists x_n\bigwedge_{(x_i,x_j)\in R} x_i R x_j \wedge \bigwedge_{(x_i,x_j)\notin R}\overline{x_i R x_j}.$$

Obviously, $\hpsi$ has $n$ variables, and $G$ obviously is a model for $\hpsi,$ it follows that $G'\models\hpsi$. Therefore, there are $x_1',\dots,x_n'$ such that $(x_i,x_j)\in R$ if and only if $(x_i',x_j')\in R'$. Define $f(x_i):=x_i'$. This function obviously meets the criteria: Let $(x_i,x_j)\in R$. Then $x_i R x_j$ is a clause in $\hpsi$. Therefore for the values $x_i'$ and $x_j'$ chosen by the existential quantifiers, $(x_i',x_j')\in R$ must hold. Since $f(x_i)=x_i'$ and $f(x_j)=x_j',$ the claim follows. For the condition $\overline{R(x_i,x_j)},$ the proof is the same.
}
\item[$3\rightarrow 2$]{Let such a function $f$ exist for every $V\subseteq G$ with $\card{V}\leq n$. Let $\hpsi:=\exists x_1,\dots,\exists x_n\hvarphi$ be a first-order existential formula, let $n$ be the number of variables in $\hpsi$, and let $G\models\hpsi$. We show that $G'\models\hpsi$. Since $G\models\hpsi,$ there are $x_1,\dots,x_n\in G$ such that $\hvarphi(x_1,\dots,x_n)$ holds in $G$. This implies that $\hvarphi(f(x_1),\dots,f(x_n))$ holds in $G'$. Therefore, $G'\models\hpsi$.}
\end{description}
\end{proof}

Note that while the function $f$ required to exist in the conditions of the above theorem shares some properties with an isomorphism, it is not required to be injective. The following is an important special case, which immediately follows from this observation and Theorem~\ref{theorem:ufo invariants}:

\begin{proposition}\label{prop:trivial subgraph property}
Let $G$ be a graph, $\hvarphi$ a universal first-order formula over the frame language such that $G\models\hvarphi$. Then for every subgraph $G'$ of $G,$ it holds that $G'\models\hvarphi$.
\end{proposition}

For a modal model, a restriction of the model is a restriction of the graph, where the propositional assignment for the remaining worlds is unchanged. We now consider restrictions which are ``compatible'' with the modal properties of the formulas in question. The following lemma describes a standard way to reduce the number ``relevant'' of successors to worlds in models. This is an application of the more general idea of bounded morphisms, which we will encounter in Section~\ref{sect:trees for horn}. For a modal formula $\phi,$ $\subf{\phi}$ denotes the set of its subformulas. With $\md\phi,$ we denote the modal depth of a formula $\phi,$ i.e., the maximal nesting degree of the modal operator $\Diamond$ in $\phi$.

\begin{lemma}\label{lemma:modal restriction invariance}
Let $\phi$ be a modal formula, and let $M,w\models\phi$. Let $M'$ be a restriction of $M$ such that the following holds:

\begin{enumerate}
\item $w\in M',$
\item for all $u\in M',$ and all $\psi\in\subf{\phi}$ such that $M,u\models\Diamond\psi,$ there is some $v\in M'$ such that $(u,v)$ is an edge in $M,$ and $M,v\models\psi$.
\end{enumerate}

Then $M',w\models\phi$.
\end{lemma}

Lemma~\ref{lemma:modal restriction invariance} immediately follows from the following Lemma. The version stated in Lemma~\ref{lemma:modal restriction invariance} is the one we almost exclusively use, hence we stated this simpler version explicitly. We now prove a slightly more general result, which also takes into account that for a modal formula, worlds which are not reachable on a path with at most the length of the modal depth of the formula, are irrelevant.

\begin{lemma}\label{lemma:modal restriction invariance with levels}
Let $\phi$ be a modal formula, and let $M,w\models\phi$. Let $M'$ be a restriction of $M$ such that the following holds:

\begin{enumerate}
\item $w\in M',$
\item for all $u\in M',$ and all $\psi\in\subf\phi,$ such that $M,u\models\Diamond\psi$ and there exists an $i\in\mathbb{N}$ such that $\mpath wiu$ and $1+i+\md\psi\leq\md\phi,$ there is some $v\in M'$ such that $(u,v)$ is an edge in $M,$ and $M,v\models\psi$.
\end{enumerate}

Then $M',w\models\phi$.
\end{lemma}

\begin{proof}
We show the following claim: Let $\chi\in\subf\phi,i\in\mathbb{N},$ and $u\in M'$ such that $\mpath wiu,$ and $i+\md\chi\leq\md\phi,$ then $M,u\models\chi$ if and only if $M',u\models\chi$. For $\chi=\phi,$ $u=w,$ and $i=0,$ this implies the Lemma, since $w$ is an element of $M'$ by definition.

We show the claim by induction on $\chi$. If $\chi$ is a variable, then this holds trivially, since $M'$ is a restriction of $M$ and therefore, propositional assignments are not changed. The induction step for propositional operators is trivial. Therefore, assume that $\chi=\Diamond\psi$ for some $\psi\in\subf\phi,$ such that the claim holds for $\psi$. Now let $u,i$ meet the prerequisites of the claim, i.e., let $i+\md\chi\leq\md\phi,$ and let $\mpath wiu$.

First assume that $M,u\models\chi$. Since $\chi=\Diamond\psi,$ it follows that $\md\chi=\md\psi+1,$ and hence $i+1+\md\psi\leq\md\phi$. Since $M,u\models\Diamond\psi,$ the prerequisites of the Lemma therefore imply that there is a world $v\in M'$ such that $(u,v)$ is an edge in $M,$ and $M,v\models\psi$. Since $\mpath wiu,$ it follows that $\mpath w{i+1}v$. By the induction hypothesis, we know that $M',v\models\psi$. Since $M'$ is a restriction of $M,$ $(u,v)$ is an edge in $M'$ as well, and therefore we conclude that $M',u\models\Diamond\psi,$ i.e., $M',u\models\chi$.

For the other direction, assume that $M',u\models\chi$. Therefore, there is a node $v\in M'$ such that $M',v\models\psi,$ and $(u,v)$ is an edge in $M'$. Since $\mpath wiu,$ we know that $\mpath w{i+1}v$ holds as well, and since $\md\psi=\md\chi-1,$ from the induction hypothesis we conclude that $M,v\models\psi$. Since $(u,v)$ is an edge in $M$ as well, it therefore follows that $M,u\models\Diamond\psi,$ i.e., $M,u\models\chi,$ concluding the proof.

\end{proof}

The following easy proposition shows how Lemma~\ref{lemma:modal restriction invariance with levels} can be applied:

\begin{proposition}\label{prop:rooted models exist}
Let $\KL$ be a universal elementary logic, let $\phi$ be a modal formula, and let $M,w\models\phi,$ where $M$ is a \KL-model. Then there is a \KL-model $M'$ which is a restriction of $M,$ which is rooted at $w,$ and where every world can be reached from $w$ in at most $\md{\phi}$ steps, and $M',w\models\phi$.
\end{proposition}

\begin{proof}
The model $M'$ is obtained by simply removing all worlds from $M$ which cannot be reached from $w$ in at most $\md{\phi}$ steps. Due to Proposition~\ref{prop:trivial subgraph property}, $M'$ is still a \KL-model. We now show that $M',w\models\phi$ holds, by proving that $M'$ satisfies the conditions of Lemma~\ref{lemma:modal restriction invariance with levels}.

By definition, since $w$ can be reached from $w$ in $0$ steps, we know that $w\in M'$. Hence let $u\in M',$ and let $\psi\in\subf{\phi},$ and let $M,u\models\Diamond\psi,$ such that $M\models\path wiu$ for some $i$ such that $1+i+\md\psi\leq\md\phi$. Since $M,u\models\Diamond\psi,$ we know that there is a world $v\in M$ such that $M,v\models\psi,$ and $(u,v)$ is an edge in $M$. It follows that $M\models\path w{i+1}u$. Since $1+i+\md\psi\leq\md\phi,$ and $\md\psi\leq\md\phi,$ we know that $i+1\leq\md\phi,$ and hence we know that $v$ is an element of $M'$.

Therefore, $M'$ satisfies the conditions of Lemma~\ref{lemma:modal restriction invariance with levels}, and therefore the lemma implies that $M',w\models\phi,$ as claimed.
\end{proof}

\subsection{Restrictions}

As mentioned, our \NP-containment results are obtained by proving the polysize model property and applying Proposition~\ref{proposition:elementary plus polysize model property is in np}. The polynomial models are obtained by restricting arbitrary models to polynomial size. We will now show some restrictions which we can make in any model, showing that we can assume certain parts of the model to be only polynomial in size.

\begin{lemma}\label{lemma:few nodes with few predecessors}
 Let $c\in\mathbb{N}$. Then for any modal formula $\phi$ and any $M,w\models\phi,$ there is a submodel $M'$ of $M$ such that $M',w\models\phi$ and the following holds:
$$\card{\set{v\in M'\ \vert\ \mxdepth{M'}{v}<c}}\leq (c+1)\cdot \card{\phi}^c.$$
\end{lemma}

\begin{proof}
For each world $u$ in $M,$ let $F_u:=\set{\psi\in\subf\phi\ \vert\ M,u\models\Diamond\psi}$. For each $u,$ let $W_u$ be a subset of the $1$-step successors of $u$ in $M$ such that for every $\psi\in F_u,$ there is a world $v\in W_u$ such that $(u,v)$ is an edge in $M,$ and $M,v\models\psi,$ and $\card{W_u}\leq\card{F_u}$. Now define $M_0:=\set{w},$ and for each $i\in\mathbb{N},$ let $M_{i+1}:=\bigcup_{v\in M_i}W_v$. Finally, define $M'$ to be the restriction of $M$ to $\bigcup_{i\in\mathbb{N}}M_i$.

To show that $M',w\models\phi,$ we prove that $M'$ satisfies the conditions of Lemma~\ref{lemma:modal restriction invariance}. Obviously, $M'$ is a restriction of $M$ and $w\in M_0\subseteq M$. Therefore, let $u$ be a world from $M',$ and let $\psi$ be a subformula of $\phi$ such that $M,u\models\Diamond\psi$. Since $u\in M',$ there is some $i$ such that $u\in M_i$. Since $M,u\models\Diamond\psi,$ we know that $\psi\in F_u,$ and hence there is a world $v\in W_u$ such that $M,v\models\psi,$ and $(u,v)$ is an edge in $M$. It follows that $v\in W_u\subseteq M_{i+1}\subseteq M',$ and hence $M'$ satisfies the conditions of Lemma~\ref{lemma:modal restriction invariance} as claimed.

Now let $$A:={\set{v\in M'\ \vert\ \mxdepth{M'}{v}<c}}.$$

It remains to show the cardinality bound for $A$. It is obvious that $A\subseteq\cup_{i=0}^c M_i,$ since for every $i,$ every vertex in $M_{i+1}$ has a predecessor in $M_i,$ and hence inductively, every vertex in $M_i$ has an $i$-step predecessor in $M'$.  

Obviously, $\card{M_0}=1,$ and $\card{M_{i+1}}\leq \card{M_i}\cdot\card{\subf{\phi}}$. Therefore, $\card{M_i}\leq\card{\subf{\phi}}^i$ for all $i\in\mathbb{N}$. Now, due to the above, $\card{A}\leq\card{\cup_{i=0}^c M_i}\leq (c+1)\cdot\card{\subf{\phi}^c}$. Since $\card{\subf{\phi}}\leq\card{\phi},$ the claim follows.
\end{proof}

By construction, the model $M'$ given in the proof of the above lemma is countable. Hence we obtain the following corollary:

\begin{corollary}\label{corrollary:universal elementary logics have countable model property}
Let $\KL$ be a universal elementary logic. Then every \KL-satisfiable formula $\phi$ is satisfiable in a countable \KL-model. Moreover, every \KL-model satisfying $\phi$ at a node $w$ contains a countable \KL-submodel satisfying $\phi$ at the node $w$.
\end{corollary}

The following lemma shows that it is sufficient to restrict the number of those vertices in the model which have a minimal height in the graph. In combination with Lemma~\ref{lemma:few nodes with few predecessors}, this shows that we only need to be concerned about vertices which have both a certain number of predecessors, and a certain number of successors. We already saw in Proposition~\ref{prop:rooted models exist} that we are only interested in rooted graphs. In such graphs, Lemmas~\ref{lemma:few nodes with few predecessors} and~\ref{lemma:few nodes with few successors} can be seen as limiting the number of vertices near the ``top'' or the ``bottom'' of the model. This is useful, because in graphs satisfying some universal formula, often special cases can occur in these regions of the graph. These lemmas show that we do not need to look too closely at these exceptions.

\begin{lemma}\label{lemma:few nodes with few successors}
Let $c\in\mathbb{N}$ be a constant. Then for any modal formula $\phi$ and any $M,w\models\phi,$ there is a submodel $M'$ of $M$ such that $M',w\models\phi,$ and for which the following holds:

$$\card{M'}\leq f\left(\card{\set{v\in M'\ \vert\ \mxheight{M'}{v}\ge c}}\right),$$

where $f(n)=(n+1)\cdot(1+\card{\phi})^c$.
\end{lemma}

\begin{proof}
Let $A:=\set{v\in M\ \vert\ \mxheight{M}{v}\ge c},$ i.e., the set of nodes in $M$ which have a $c$-step successor. We now define a sequence of submodels of $M:$ Let $M_0:=\emptyset,M_1:=A\cup\set{w},$ and for $i\ge 1,$ let $M_{i+1}$ be defined as follows:

\begin{itemize}
 \item $M_i\subseteq M_{i+1},$
 \item For every $u\in M_{i}\setminus M_{i-1},$ and each $\psi\in\subf{\phi}$ such that $M,u\models\Diamond\psi,$ add one world $v$ from $M$ to $M_{i+1}$ such that $(u,v)$ is an edge in $M$ and $M,v\models\psi$.
\end{itemize}

Now let $M'$ be the restriction of $M$ to the worlds in $M_{c+1}$.  We show that $M',w\models\phi,$ by showing that it satisfies the conditions of Lemma~\ref{lemma:modal restriction invariance}. By definition, $w\in M_1\subseteq M'$. Hence let $u\in M',$ $\psi\in\subf{\phi},$ and let $M,u\models\Diamond\psi$. Since $u\in M',$ there exists a minimal $i$ such that $u\in M_i$. First assume that $i=c+1$. By construction, for every relevant $j,$ every node in $M_{j+1}\setminus M_j$ has a $1$-step predecessor in $M_j\setminus M_{j-1},$ and hence, inductively, the node $u\in M_{c+1}\setminus M_c$ is a $c-1$-step successor of a node $x$ in $M_2\setminus M_1$. Since $M,u\models\Diamond\psi,$ we know that $u$ has a successor in $M$. This implies that $x$ has a $c$-step successor in $M,$ and hence $x\in A,$ which is a contradiction, since $x\in M_2\setminus M_1,$ and $A\subseteq M_1$.

Therefore, we know that $i\leq c$. By construction, there is a world $v$ in $M_{i+1}\subseteq M'$ such that $(u,v)$ is an edge in $M,$ and $M.v\models\psi$. Therefore, the model $M'$ satisfies the conditions of Lemma~\ref{lemma:modal restriction invariance}, and therefore we conclude that $M',w\models\phi$.

By definition, it holds that $\card{M_1}\leq\card{A}+1,$ and for $i\ge 1,$ $\card{M_{i+1}}\leq\card{M_i}(1+\card{\subf{\phi}})$. Since $\card{M'}=\card{M_{c+1}},$ this implies that $\card{M'}\leq (\card{A}+1)\cdot(1+\card{\subf{\phi}})^c\leq(\card{A}+1)\cdot(1+\card{\phi})^c,$ as claimed.
\end{proof}

The main purpose of Lemmas~\ref{lemma:few nodes with few predecessors} and~\ref{lemma:few nodes with few successors} is the following: If for a modal logic $\KL,$ there is a constant $c,$ such that every \KL-satisfiable formula has a model in which we can restrict the number of nodes which have both a $c$-step predecessor and a $c$-step successor, then Lemmas~\ref{lemma:few nodes with few predecessors} and~\ref{lemma:few nodes with few successors} can be used to show the polynomial model property for \KL. This idea plays a crucial role in the proof of our main \NP-containment result, Theorem~\ref{theorem:above k to k+1 leads to np}, and is formalized in the following corollary:

\begin{corollary}\label{corollary:suffices to restrict middle vertices for np}
 Let $\KL$ be a universal elementary modal logic such that there exists a constant $c\in\mathbb{N}$ such that there is a polynomial $p$ such that for all \KL-satisfiable formulas $\phi,$ there is a \KL-model $M$ and a world $w\in M$ such that $M,w\models\phi,$ and

$$\card{\set{u\in M\ \vert\ \mxdepth{M}{w}\ge c\mathtext{ and }\mxheight{M}{w}\ge c}}\leq p(\card{\phi}).$$

Then $\KL$ has the polynomial-size model property, and $\sat{\KL}\in\NP$.
\end{corollary}

Note that the function $p$ is only required to be polynomial in its argument $\card{\phi},$ and not in the value $c,$ which is a constant depending only on the logic, and not on the formula.

\begin{proof}
 Let $\phi$ be a \KL-satisfiable formula, and let $M$ be a \KL-model and $w\in M$ meeting the prerequisites of the corollary. By Lemma~\ref{lemma:few nodes with few predecessors}, there is a submodel $M'$ of $M$ such that $M',w\models\phi,$ and 

$$\card{\set{v\in M'\ \vert\ \mxdepth{M'}{v}<c}}\leq(c+1)\cdot\card{\phi}^c.$$

Therefore, since the conditions required in the prerequisites of the corollary are invariant under further restrictions of the model, assume without loss of generality that $M$ already satisfies this condition. Now let $M=T\cup C\cup B,$ where

\begin{eqnarray*}
 T & := & \set{v\in M\ \vert\ \mxdepth{M}{v}<c}, \\
 C & := & \set{v\in M\ \vert\ \mxdepth{M}v\ge c\mathtext{ and }\mxheight{M}{v}\ge c}, \\
 B & := & \set{w\in M\ \vert\ \mxheight{M}{c}<c}.
\end{eqnarray*}

Note that $T$ and $B$ are not necessarily disjoint ($T$ contains the nodes with only small depth at the ``\textbf{t}op'' of the model, $C$ represents the ``\textbf{c}enter'' of the model, and $B$ is the ``\textbf{b}ottom.''). By the prerequisites of the corollary, we can assume that $\card{C}\leq p\left(\card\phi\right),$ and by the above we know that $\card{T}\leq(c+1)\cdot\card{\phi}^c$. Now let $A:=\set{v\in M\ \vert\ \mxheight{M}{v}\ge c}$. It follows that $A\subseteq T\cup C,$ and hence $\card{A}\leq\card{T}+\card{C}\leq(c+1)\cdot\card{\phi}^c+p\left(\card{\phi}\right).$

By Lemma~\ref{lemma:few nodes with few successors}, there is a submodel $M'$ of $M$ such that $M',w\models\phi$ and $\card{M'}\leq(\card{A}+1)\cdot(1+\card{\phi})^c,$ and hence by the above we have that $\card{M'}\leq
((c+1)\cdot\card{\phi}^c+p(\card{\phi})+1)\cdot(1+\card{\phi})^c$ (note that the cardinality of the size $A$ defined with respect to the submodel $M'$ is bounded by the cardinality of the original set $A$). Since $p$ is a polynomial and $c$ is a constant only depending on the logic \KL, this is a polynomial size bound in $\card{\phi},$ and therefore we have proven the polynomial-size model property. By Proposition~\ref{prop:trivial subgraph property}, $M'$ is a \KL-model. The complexity result now follows from Proposition~\ref{proposition:elementary plus polysize model property is in np}.
\end{proof}

\section{Universal Horn Formulas}\label{sect:universal horn clauses}

We now consider a syntactically restricted case of universal first order formulas, namely Horn formulas. Many well-known logics can be expressed in this way.

\subsection{Definitions}\label{sect:horn definitions}

Usually, a Horn clause is defined as a disjunction of literals of which at most one is positive. If a positive literal occurs, then the clause can be written as an implication, since $\overline{x_1}\vee\dots\vee\overline{x_n}\vee y$ is equivalent to $x_1\wedge\dots\wedge x_n\implies y$. If no positive literal occurs, then the clause $(\overline{x_1}\vee\dots\vee\overline{x_n})$ can be written as $x_1\wedge\dots\wedge\ x_n\implies\mathsf{false}$. Since in the context of the frame language, an atomic proposition is of the form $(x R y),$ the following is the natural version of Horn clauses for our purposes:

\begin{definition}
A \emph{universal Horn clause} is a formula of the form 
$(x_1 R x_2)\wedge\dots\wedge(x_{k-1} R x_k)\implies (x_i R x_j),$ or of the form 
$(x_1 R x_2)\wedge\dots\wedge(x_{k-1} R x_k)\implies \mathsf{false},$
where all (not necessarily distinct) variables are implicitly universally quantified.
\end{definition}

A \emph{universal Horn formula} is a conjunction of universal Horn clauses. With universal Horn formulas, many of the usually considered graph properties can be expressed, like transitivity, symmetry, euclidicity, etc. In the following definition, we show how universal Horn clauses can be represented as graphs.

\begin{definition}
Let $\hvarphi$ be a universal Horn clause.
\begin{itemize}
\item The \emph{prerequisite graph of $\hvarphi$}, denoted with \preq{\hvarphi}, consists of the variables appearing on the left-hand side of the implication $\varphi,$ where $(x_1,x_2)$ is an edge if the clause $(x_1 R x_2)$ appears.
\item If $\hvarphi$ is a universal Horn clause where the right-hand side of the implication in $\hvarphi$ is $R(x,y),$ then the \emph{conclusion edge of $\hvarphi$}, denoted with \conc{\hvarphi}, is the edge $(x,y)$. If the right-hand side of the implication is $\mathsf{false},$ then $\conc\hvarphi$ is the empty set.
\end{itemize}
\end{definition}

\subsection{Universal Horn Clauses and Homomorphisms}\label{sect:horn homomorphism}

The definition of the prerequisite graph and the conclusion edge of a universal horn formula establishes a one-to-one correspondence between universal Horn clauses and their representation as graphs. These definitions allow us to relate truth of a Horn clause to homomorphic images of the involved graphs:

\begin{proposition}\label{prop:homomorphism and horn clauses}
\begin{enumerate}
 \item Let $\hvarphi$ be a universal Horn clause with $\conc\hvarphi=(x,y)$. A graph $G$ satisfies $\hvarphi$ if and only the following holds: For every homomorphism $\alpha\colon\preq{\hvarphi}\cup\set{x,y}\rightarrow G$,  $(\alpha(x),\alpha(y))$ is an edge in $G.$
\item Let $\hvarphi$ be a universal Horn clause such that $\conc{\hvarphi}=\emptyset$. Then a graph $G$ satisfies $\hvarphi$ is and only if there is no homomorphism $\alpha\colon\preq\hvarphi\rightarrow G.$
\end{enumerate}
\end{proposition}

\begin{proof}
\begin{enumerate}
\item Let $\preq\hvarphi=(x_1 R x_2)\wedge\dots\wedge (x_{n-1} R x_n),$ where all variables are implicitly universally quantified. First assume that $G\models\hvarphi,$ and let $\alpha\colon\preq{\hvarphi}\cup\set{x,y}\rightarrow G$ be a homomorphism. Due to the definition of $\preq\hvarphi,$ there are edges $(x_1,x_2),\dots,(x_{n-1},x_n)$ in $\preq{\hvarphi}$. Since $\alpha$ is a homomorphism, this implies that $(\alpha(x_1),\alpha(x_2)),\dots,(\alpha(x_{n-1}),\alpha(x_n))$ are edges in $G$. Hence the nodes $\alpha(x_1),\dots,\alpha(x_n),\alpha(x),\alpha(y)$ satisfy the formula $R(\alpha(x_1),\alpha(x_2))\wedge\dots\wedge R(\alpha(x_{n-1}),\alpha(x_{n}))$. Therefore, the nodes $\set{\alpha(v)\ \vert\ v\in\preq\hvarphi\cup\set{x,y}}$ satisfy the prerequisites of the clause $\hvarphi$. Since $G\models\hvarphi,$ this implies that $(\alpha(x),\alpha(y))$ is an edge in $G.$

Now for the other direction, assume that $G$ fulfills the homomorphism property, and let $\var{\hvarphi}=\set{x_1,\dots,x_n}$. Let $a_1,\dots,a_n$ be nodes in $G$ satisfying the prerequisite clause of $\hvarphi,$ i.e., if $(x_{i_1},x_{i_2})$ is a clause in $\preq\hvarphi,$ then $(a_{i_1},a_{i_2})$ is an edge in $G$. Then obviously the function $\alpha$ mapping the variable $x_i$ to the node $a_i,$ is a homomorphism from $\preq\hvarphi$ to $G$. By the prerequisites, we know that $(\alpha(x),\alpha(y))$ is an edge in $G$. Hence, $G$ satisfies the formula $\hvarphi$. 
\item Analogous.
\end{enumerate}
\end{proof}

There is a natural correspondence between implications of these formulas and graph homomorphisms.

\begin{proposition}\label{prop:homomorphism and horn clause implication}
\begin{enumerate}
\item Let $\hvarphi_1$ and $\hvarphi_2$ be universal Horn clauses such that there exists a homomorphism $\alpha\colon\preq{\hvarphi_1}\rightarrow\preq{\hvarphi_2},$ which maps the conclusion edge of $\hvarphi_1$ to the conclusion edge of $\hvarphi_2$. Then $\hvarphi_1$ implies $\hvarphi_2.$
\item Let $\hvarphi_1$ and $\hvarphi_2$ be universal Horn clauses such that $\conc{\hvarphi_1}=\conc{\hvarphi_2}=\emptyset,$ and let $\alpha\colon\preq{\hvarphi_1}\rightarrow\preq{\hvarphi_2}$ be a homomorphism. Then $\hvarphi_1$ implies $\hvarphi_2.$
\end{enumerate}
\end{proposition}

\begin{proof}
\begin{enumerate}
 \item
 Let $\conc{\hvarphi_1}=(x,y),$ then by the prerequisites it follows that $\conc{\hvarphi_2}=(\alpha(x),\alpha(y))$. Now let $G$ be a graph such that $G\models\hvarphi_1,$ and let $\beta\colon(\preq{\hvarphi_2}\cup\set{\alpha(x),\alpha(y)})\rightarrow G$ be a homomorphism. By Proposition~\ref{prop:homomorphism and horn clauses}, it suffices to show that $(\beta(\alpha(x)),\beta(\alpha(y)))$ is an edge in $G$. Since $\alpha$ and $\beta$ are homomorphisms, $\beta\circ\alpha\colon\preq{\hvarphi_1}\rightarrow G$ is a homomorphism as well. Since $G\models\hvarphi_1,$ Proposition~\ref{prop:homomorphism and horn clauses} implies that $(\beta(\alpha(x)),\beta(\alpha(y)))$ is an edge in $G,$ as claimed.

\item Let $G$ be a graph such that $G\models{\hvarphi_1}$. Due to Proposition~\ref{prop:homomorphism and horn clauses}, to show that $G\models\hvarphi_2,$ it suffices that there is no homomorphism $\beta\colon\preq{\hvarphi_2}\rightarrow G$. Hence assume that such a homomorphism exists. Then $\beta\circ\alpha$ is a homomorphism from $\preq{\hvarphi_1}$ into $G,$ which is a contradiction to Proposition~\ref{prop:homomorphism and horn clauses}, since $G\models\hvarphi_1.$
\end{enumerate}
\end{proof}

\subsection{Important Special Cases}\label{section:varphi k to l section}

We now consider special cases of Horn clauses, which will be central for the logics having satisfiability problems in \NP. The following definition captures the case where the variables in the conclusion edge have a common predecessor in the prerequisite graph, but there is not necessarily a direct path between them. Using results about graphs satisfying generalizations of formulas of this type, we will be able to show all of the \NP-containment results that the proof of the later classification theorem, Theorem~\ref{theorem:horn conjunction classification}, depends on.

\begin{definition}
Let $\hvarphi^{k\rightarrow l}$ be the formula $$(w R x_1) \wedge (x_1 R x_2)\wedge\dots \wedge (x_{k-1} R x_k) \wedge (w R y_1) \wedge (y_1 R y_2)\wedge\dots\wedge (y_{l-1} R y_l)\implies (x_k R y_l),$$
where all variables are universally quantified (and in the case that $x_0$ or $y_0$ appear in the formula, we replace them with $w$).
\end{definition}

\begin{wrapfigure}[14]{l}{3.3cm}
\includegraphics{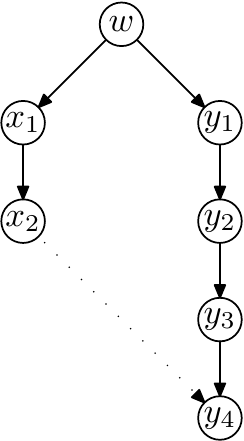}
\caption{Example clause $\hvarphi^{2\rightarrow 4}$}
\label{fig:varphi 2 to 4 example}
\end{wrapfigure}
In Figure~\ref{fig:varphi 2 to 4 example}, we present the graph representation of the formula $\hvarphi^{2\rightarrow 4}$. The graph property described by these formulas is easy to see:

\begin{proposition}\label{prop:varphi k to l property}
Let $G$ be a graph, and let $k,l\in\mathbb{N}$. Then $G\models\hvarphi^{k\rightarrow l}$ if and only if for any nodes $w,x_k,y_l\in G,$ if $G\models\path{w}{k}{x_k}$ and $G\models\path{w}{l}{y_l},$ then $(x_k,y_l)$ is an edge in $G.$
\end{proposition}

This definition generalizes several well-known examples---in particular, a graph is reflexive if and only if it satisfies $\hvarphi^{0\rightarrow 0},$ symmetry is expressed with $\hvarphi^{1\rightarrow 0},$ and transitivity with $\hvarphi^{0\rightarrow 2}$. Finally, a graph is Euclidean iff it satisfies $\hvarphi^{1\rightarrow 1}$. Therefore, this notation allows us to capture many interesting graph properties, and it is not surprising that generalizations of this idea are the main ingredients for our polynomial size model proofs. We start with looking at some properties and implications of formulas of the form $\hvarphi^{k\rightarrow l}.$

\begin{lemma}\label{lemma:varphi k to l implications}
Let $1\leq k,l\in\mathbb{N},$ and let $G$ be a graph such that $G\models\hvarphi^{k\rightarrow l}.$
\begin{enumerate}
\item $G\models\hvarphi^{l+k-1\rightarrow l+k}.$
\item If $l=k+1,$ then for any $i\ge k$, $G\models\hvarphi^{i\rightarrow i+1}.$
\item There is some $k'\ge 1$ such that $G\models\hvarphi^{i\rightarrow i+1}$ for all $i\ge k'.$
\end{enumerate}
\end{lemma}

\textit{Proof}
\begin{enumerate}
\item Let $w$ be some node in $G,$ such that $\gpath{w}{l+k-1}{x_{l+k-1}},$ and $\gpath{w}{l+k}{y_{l+k}},$ and let the (not necessarily distinct) intermediate vertices be denoted with $x_i,y_i$. Since $\hvarphi^{k\rightarrow l}$ holds in $G$, this implies that there is an edge $(y_k,x_l)$. By choice of nodes, 
$\gpath{x_l}{k-1}{x_{l+k-1}}$. Combining these, we obtain a path of length $k$ from $y_k$ to $x_{l+k-1}$. On the other hand, $\gpath{y_k}{l}{y_{k+l}}$. Since $\hvarphi^{k\rightarrow l}$ holds in $G$, it follows that there is an edge from $x_{l+k-1}$ to $y_{l+k},$ proving that $\hvarphi^{l+k-1\rightarrow l+k}$ holds in $G$.

\item \indent Clearly it suffices to show $G\models\hvarphi^{k+1\rightarrow k+2},$ the claim for arbitrary $i$ follows inductively. Let $w,x_j,y_j$ be chosen such that $w=x_0=y_0,$ and there are edges $(x_j,x_{j+1})$ and $(y_j,y_{j+1})$. We need to show that there is an edge $(x_{k+1},y_{k+2}).$

\indent Since $G\models\hvarphi^{k\rightarrow k+1},$ it follows that $(y_k,x_{k+1})$ is an edge in $G$. Since there obviously is a path of length $k-1$ from $y_1$ to $y_k$, it follows that there is a path of length $k$ from $y_1$ to $x_{k+1}$. Since there also is a path of length $k+1$ from $y_1$ to $y_{k+2},$ it follows that there is an edge $(x_{k+1},y_{k+2})$ in $G$, which concludes the proof.

\item \indent This follows immediately from the above: from points $1$ and $2,$ it follows that the claim holds for $k':=l+k-1.$
\end{enumerate}
\hfill$\Box$\bigskip

\begin{wrapfigure}[21]{l}{4.5cm}
 \includegraphics{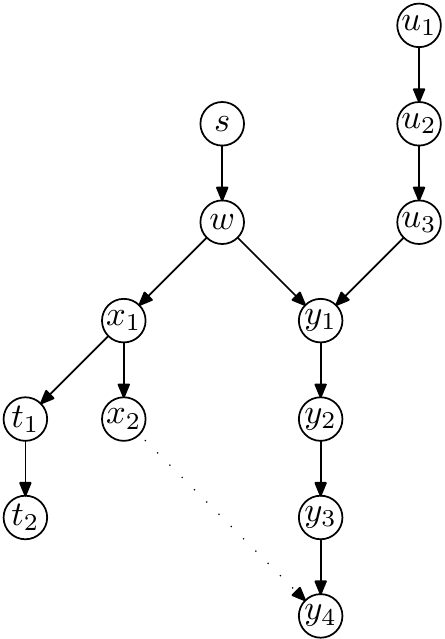}
\caption{More general formula}
\label{fig:transformation original formula}
\end{wrapfigure}
The formula $\hvarphi^{k\rightarrow l}$ is supposed to capture the case where the variables in the conclusion edge of a universal Horn clause have a common predecessor. But not all of these cases are covered with this formula. The above Figure~\ref{fig:varphi 2 to 4 example} is a graphical representation of what the implication $\hvarphi^{2\rightarrow 4}$ does. But what if this is only a subgraph of the prerequisite graph? In a more general case, the node $w$ and the nodes $x_k,y_l$ will have more predecessors and successors. Figure~\ref{fig:transformation original formula} gives an example of a more general formula. We will now see that this formula can be ``simplified.'' This simplification is not an equivalent transformation of the formula, but we construct a new formula which is implied by the original one. The one-sided implication suffices to show many of the results we need. The simpler formula is presented in Figure~\ref{fig:transformation resulting formula}.

It is easy to see that every graph which satisfies the formula displayed in Figure~\ref{fig:transformation original formula} also satisfies the formula from Figure~\ref{fig:transformation resulting formula}. This follows directly from Proposition~\ref{prop:homomorphism and horn clause implication}, since the prerequisite graph from Figure~\ref{fig:transformation original formula} can obviously be mapped homomorphically to the prerequisite graph from Figure~\ref{fig:transformation resulting formula} (the homomorphism $\alpha$ is defined as $\alpha(t_1):=x_2,\alpha(s):=u_2,\alpha(u_3):=w,$ and maps the other nodes to the ones with the same labels). Hence, if we can show \NP-containment for all universal elementary modal logics extending the one defined by the later formula, this puts the logic defined by the original formula into \NP\ as well.

\begin{wrapfigure}[19]{l}{2.4cm}
 \includegraphics{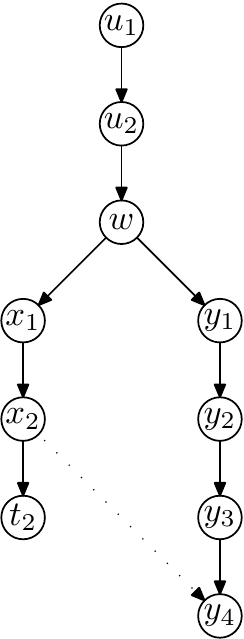}
\caption{Simplified formula}
\label{fig:transformation resulting formula}
\end{wrapfigure}
Any universal Horn clause which can be mapped onto a tree can be embedded in a graph with certain properties, namely the properties of the formula we now define. Due to Corollary~\ref{corollary:formulas satisfiable in trees give pspace complete logics}, it is natural that tree-like homomorphic images of our universal Horn formulas are of interest to us. These formulas capture the generalizations of $\hvarphi^{k\rightarrow l}$ mentioned above, where we demand that the nodes $w,x_k,y_l$ have a sufficient number of predecessors or successors. We again use the representation of Horn clauses as graphs.

\begin{wrapfigure}[20]{l}{3.3cm}
 \includegraphics[scale=0.5]{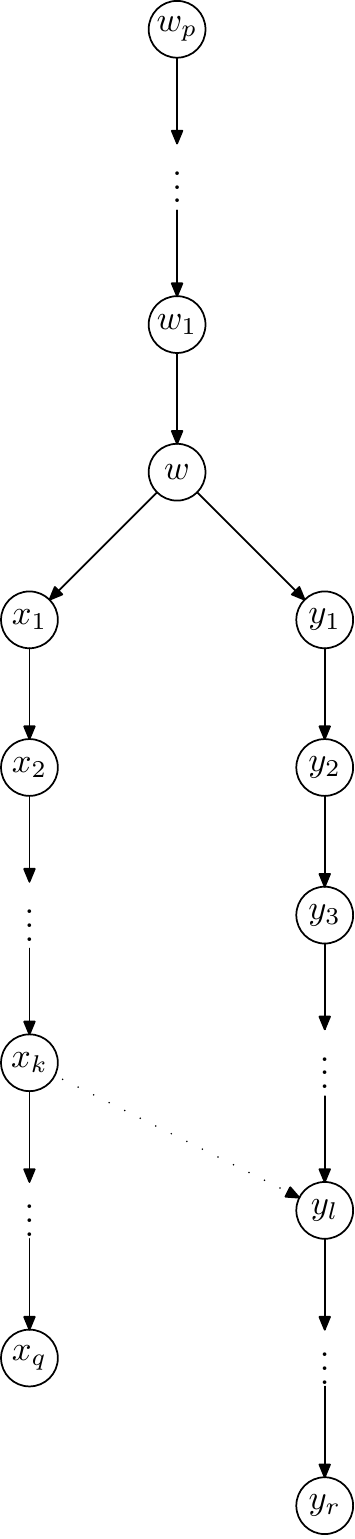}
\caption{The formula $\genvarphi{k}{l}{p}{q}{r}$}
\label{fig:general xyz stu formula}
\end{wrapfigure}

\begin{definition}
 For $k,l,p,q,r\in\mathbb{N},$ the formula $\genvarphi{k}{l}{p}{q}{r}$ is defined as the universal Horn clause displayed in Figure~\ref{fig:general xyz stu formula}.
\end{definition}

It should be noted that the notation $\genvarphi klpqr$ suggests that $w,y,x$ can be compared to natural numbers, but what is meant in that notation is simply that the number of predecessors (successors, resp.) of $w,x_k,y_l$ can be compared to $p,q,$ and $r,$ respectively. Hence, if we use ``natural names'' for the vertices, i.e. we have vertices $w=x_0,x_1,\dots,$ and $w=y_0,y_1,\dots,$ then this ensures that the vertices up to $x_q,$ $y_r,$ and $w_p$ exist. When proving that this formula holds in a graph, we will usually rely on the notation provided in Figure~\ref{fig:general xyz stu formula}, and assume that there are nodes $w_{p},\dots,w_0=w=x_0=y_0,x_1,\dots,x_q,x_1,\dots,y_r$ with edges as seen in Figure~\ref{fig:general xyz stu formula}, i.e., most of the time we do not mention the homomorphism explicitly.

The formula $\genvarphi klpqr$ can be seen to be only a slight generalization of the formulas $\hvarphi^{k\rightarrow l}$ we already considered, as exhibited by the following proposition:

\begin{proposition}\label{prop:genvarphi property}
Let $G$ be a graph, and let $k,l,p,q,r\in\mathbb{N}$. Then $G\models\genvarphi klpqr$ if and only if the following condition holds: For any nodes $w,x_k,y_l\in G,$ such that $w$ has a $p$-step predecessor, $x$ has a $q-k$-step successor and $y$ has an $r-l$-step successor, $G\models\path{w}{k}{x_k}$ and $G\models\path{w}{l}{y_l},$ it follows that $(x_k,y_l)$ is an edge in $G.$
\end{proposition}

The above proposition immediately implies the following:

\begin{proposition}\label{prop:cut to many pred many suc part gives non-gen varphi}
Let $G\models\genvarphi klpqr,$ and let $G'$ be the restriction of $G$ to the set $$C:=\set{w\in G\ \vert\ \mxdepth Gw\ge p,\mxheight Gw\ge\max(q-k,r-l)},$$
then $G'\models\varphi^{k\rightarrow l}.$
\end{proposition}

This Proposition is one of the reasons why Corollary~\ref{corollary:suffices to restrict middle vertices for np} is important: We can use it together with Proposition~\ref{prop:cut to many pred many suc part gives non-gen varphi} to restrict our attention to the vertices in the ``middle'' of the graph, and then talk about the formula $\hvarphi^{k\rightarrow l}$ instead of $\genvarphi klpqr$. This is the general approach for our \NP-containment proofs, although some technical difficulties remain, as we will see in the proof of Theorem~\ref{theorem:above k to k+1 leads to np}.

There obviously is a close relationship between formulas of the form $\hvarphi^{k\rightarrow l}$ and formulas of the form $\genvarphi klpqr$. In particular, this relationship allows the ``lifting'' of implications, as is shown in the following easy lemma.

\begin{lemma}\label{lemma:varphi k to l implication gives genvarphi implications}
 Let $p,q,r,k,l,k',l'\in\mathbb{N},$ and let $\hvarphi^{k\rightarrow l}$ imply $\hvarphi^{k'\rightarrow l'}$. Then $\genvarphi klpqr$ implies $\genvarphi{k'}{l'}{p'}{q'}{r'},$ where
 $$
 \begin{array}{rl}
    p' & := p, \\ 
    q' & := k'+\max(q-k,r-l), \\
    r' & := l'+\max(q-k,r-l).
 \end{array}
 $$
\end{lemma}

\begin{proof}
Let $G$ be a graph satisfying $\genvarphi klpqr$. Let $w,x_i,y_i$ be the nodes in the graph connected as the nodes in the prerequisite graph of $\genvarphi{k'}{l'}{p'}{q'}{r'}$. Let $G'$ be the graph $G$ restricted to the set of vertices which have a $p$-step predecessor, and $\max(q-k,r-l)$-step successor in $G$. Then, by choice of nodes, and Propositions~\ref{prop:genvarphi property} and~\ref{prop:varphi k to l property}, it follows that $G'\models\hvarphi^{k\rightarrow l}$. Hence, due to the prerequisites, we know that $G\models\hvarphi^{k'\rightarrow l'}$. In particular, since the nodes $w,x_i,y_i$ satisfy the prerequisite graph of $\genvarphi{k'}{l'}{p'}{q'}{r'},$ we know that $w,x_k,y_l\in G'$. Hence, it follows that $(x_k,y_l)$ is an edge in $G'$, and therefore it is an edge in $G,$ as claimed. 
\end{proof}

We need three more implications between formulas of this kind for our later \NP-results:

\begin{lemma}\label{lemma:genvarphi k to 0 implications}\label{lemma:(prop):varphi 0 to k gives varphi 0 to k^2}
 Let $p,q,r,k\in\mathbb{N},k\ge 2,$ and let $G$ be a graph.
 \begin{enumerate}
  \item If $G\models\hvarphi^{k\rightarrow 0},$ then $G\models\hvarphi^{0\rightarrow k^2}.$
  \item If $G\models\genvarphi k0pqr,$ then $G\models\genvarphi{k-1}{k^3}{p'}{q'}{r'},$ where
  $$ 
  \begin{array}{rl}
    p' & := p, \\
    q' & := k-1+\max(q-k,r), \\
    r' & := k^3+\max(q-k,r). \\
  \end{array}
   $$
  \item Let $k\in\mathbb{N}$. Then $\hvarphi^{0\rightarrow k}$ implies $\hvarphi^{0\rightarrow k^2}.$
 \end{enumerate}
\end{lemma}

\begin{proof}
 Let $w_p,\dots,w=x_0,\dots,w=y_0,\dots$ be nodes in the graph fulfilling the prerequisite graph of the respective formulas.
 \begin{enumerate}
  \item By the $\hvarphi^{k\rightarrow 0}$-property, it is clear that $(y_k,y_0),(y_{2k},y_k),\dots,(y_{k^2},y_{(k-1)\cdot k})$ are edges in $G$. Hence,  $\gpath{y_{k^2}}{k}{w}$. Again due to the property, it follows that $(w,y_{k^2})$ is an edge in $G,$ as required.
  \item Consider the subgraph $G':=G\cap\set{w,x_1,\dots,x_{k-1},y_1,\dots,y_{k^3}}$. By choice of $p',q',r',$ every node in $G'$ has a $p$-step predecessor, a $q-k$-step successor, and an $r$-step successor in $G$. Hence, $G'\models\hvarphi^{k\rightarrow 0}$. Due to part $1,$ this implies that $G'\models\hvarphi^{0\rightarrow k^2}$. Hence, in $G'$ there are edges $(w,y_{k^2}),(y_{k^2},y_{2k^2}),\dots,(y_{(k-1)\cdot k^2},y_{k^3})$. Therefore, due to the $\hvarphi^{k\rightarrow 0}$-property, it follows that $(y_{k^3},w)$ is an edge in $G'$. Hence, it follows that $\gpath{y_{k^3}}{k}{x_{k-1}}$. The $\hvarphi^{k\rightarrow 0}$-property implies that $(x_{k-1},y_{k^3})$ is an edge, as required.
  \item This follows from Lemma~\ref{lemma:test case 1}, since $k$-transitivity implies $k+k\cdot(k-1)=k^2$-transitivity.
 \end{enumerate}
\end{proof}

\subsection{\NP\ upper complexity bounds}\label{section:np results}

In this section, we prove \NP-containment results for logics defined by universal Horn clauses. We show that ``most'' of the logics of the form $\K{\genvarphi klpqr}$ give rise to a satisfiability problem in \NP.

It is comparably easy to show that a logic of the form $\K{\hvarphi^{k\rightarrow k+1}}$ leads to a satisfiability problem in \NP, by carefully ``copying'' vertices and adding the correct neighbors. However, while this process leads to a model which still satisfies $\hvarphi^{k\rightarrow k+1},$ it does not give the desired result that the problem can be solved in $\NP$ for all universal Horn logics which are extensions of $\K{\hvarphi^{k\rightarrow k+1}}$. In order to prove this, our model-manipulations must be consistent with the conditions of Theorem~\ref{theorem:ufo invariants}. In the proof of the following theorem, we construct a small model using only restriction. 

Note that the proof for this theorem is the only occasion where we actually prove the polynomial-size model property by explicitly constructing the model. All further \NP-results make use of Theorem~\ref{theorem:above k to k+1 leads to np}, by showing that the logic in question is an extension of one of the logics that this theorem deals with. For example, the logic $\logicname{K4B}$ satisfies the conditions of this Theorem. It is therefore a central theorem for our complexity classification. 

The proof relies on one later result, namely the second case of Theorem~\ref{theorem:tree-like models for horn logics}. However, the proof for this case of Theorem~\ref{theorem:tree-like models for horn logics} does not rely on any facts about the classification algorithm studied later. In the following, a graph $G$ with root $w$ is $w$-canonical if for every node $x\in G,$ if $\gpath wix$ and $\gpath wjx,$ then $i=j$. Note that this does not mean that every node $x$ is reachable on only one path from the root, but that all paths from the root $w$ to $x$ have the same length. We say that a graph $G$ is \emph{canonical} if it is $w$-canonical for a root $w$ of $G$. It is easy to see that a graph is canonical if and only if it can be homomorphically mapped onto a strict line.

\begin{theorem}\label{theorem:above k to k+1 leads to np}
Let $\hvarphi$ be a universal Horn formula implying $\genvarphi{k}{k+1}pqr$ for some $k,p,q,r\in\mathbb{N},k\ge 1$. Then $\K\hvarphi$ has the polynomial-size model property, and $\sat{\K\hvarphi}\in\NP.$
\end{theorem}

\begin{proof}
Due to Proposition~\ref{proposition:elementary plus polysize model property is in np}, it suffices to prove the polynomial-size model property for $\K\hvarphi.$ Hence, let $\phi$ be a $\K\hvarphi$-satisfiable modal formula, and, following Proposition~\ref{prop:rooted models exist}, let $M$ be a $\K\hvarphi$-model with root $w$ such that $M,w\models\phi$. Define $\mathtext{\textit{maxline}}$ to be the maximal length of a strict line on which $\hvarphi$ is satisfied, and $\mathtext{\textit{maxline}}$ to be zero in the case that $\hvarphi$ is satisfied on every strict line (note that due to Proposition~\ref{prop:trivial subgraph property}, if $\hvarphi$ is satisfied on the strict line of length $i,$ then it is also satisfied on the strict line of length $i-1$). This number is obviously a constant depending only on $\hvarphi$. If $\hvarphi$ is satisfied on every strict line, we can, due to Theorem~\ref{theorem:tree-like models for horn logics}, assume that $M$ is $w$-canonical. Note that we will work with restrictions of the model during the course of the proof---since any restriction of a canonical graph having the same root is still canonical, the submodels we consider later are canonical as well. Now, define

\begin{eqnarray*}
A & :=&\set{v\in M\ \vert\ \mxdepth{M}{v}\leq p+k+2}\cup\set{w}, \\
S & :=&\set{v\in M\ \vert\ \mxheight{M}{v}\ge\max(q-k,r-k-1)}, \\
C & :=&(M\setminus A)\cap S.
\end{eqnarray*}

The set $C$ contains all the nodes in $M$ which have ``enough'' predecessors and successors to ensure that the formula $\genvarphi{k}{k+1}pqr$ gives us all the necessary edges that we are interested in. To be more precise, we show the following fact:

\begin{prooffact}\label{prooffact:C models varphi i to i+1}
For all $i\ge 0,$ it holds that $C\models\hvarphi^{i\rightarrow i+1}.$
\end{prooffact}

\begin{proof}
It suffices to show the claim for $i=1$. The result for arbitrary $i$ then follows from Lemma~\ref{lemma:varphi k to l implications}, and the observation that $\hvarphi^{0\rightarrow 1}$ is true in any graph. Hence, we show $C\models\hvarphi^{1\rightarrow 2}:$ Let $w',x_1,y_1,y_2$ be vertices from $C,$ such that they satisfy the prerequisites of $\hvarphi^{1\rightarrow 2},$ i.e., let $(w',x_1),$ $(w',y_1),$ and $(y_1,y_2)$ be edges in $M$. Since $w'\in C,$ we know that $w'$ has a $k-1$-step predecessor $w'',$ which in turn has a $p$-step predecessor. Due to the edges mentioned above, it obviously holds that $\mpath{w''}{k}{x_1},$ and $\mpath{w''}{k+1}{y_2}$. Since both $x_1$ and $y_2$ are elements of $C,$ they have the required number of successors, and therefore, since $M\models\genvarphi{k}{k+1}pqr,$ Proposition~\ref{prop:genvarphi property} implies that $(x_1,y_2)$ is an edge in $M,$ as claimed.
\end{proof}

Due to Lemma~\ref{lemma:few nodes with few predecessors}, we can assume, without loss of generality, that $\card A$ is polynomially bounded (in the length of $\phi$). We can additionally assume, again due to Lemma~\ref{lemma:few nodes with few predecessors}, that the number of nodes $v$ with $\mxdepth{M}{v}\leq p+k+3$ is also restricted by a polynomial. Due to Proposition~\ref{prop:rooted models exist}, we can also assume that every world in $M$ can be reached from $w$ in $\md\phi$ steps. Therefore, we can choose a set $G\subseteq M$ with the following properties:
\begin{itemize}
\item $G$ is polynomially bounded,
\item $G\subseteq C,$
\item for every node $u\in A$ and every $v\in C$ such that $(u,v)$ is an edge in $M,$ and every subformula $\psi$ of $\phi$ such that $M,v\models\psi,$ there is a node $v'\in G$ such that $(u,v')$ is an edge in $M,$ and $M,v'\models\psi,$
\item every node in $C$ can be reached on a path from some node in $G$. Additionally, if $\varphi$ is satisfied on every strict line, every node in $C$ can be reached from some node in $G$ on a path of length at most $\md\phi.$
\end{itemize}

Such a set can be chosen with an application of the technique used to prove Lemma~\ref{lemma:few nodes with few predecessors}:

\begin{prooffact}
 A set $G$ can be chosen with the properties above.
\end{prooffact}

\begin{proof}
First consider the case that $\hvarphi$ is satisfied on every strict line. In this case, since $M$ is $w$-canonical, we know that $A$ is exactly the set of vertices on the first $p+k+2$ levels of $M$ (the $i$-th level is the set of vertices $v$ in $M$ such that $\mpath wiv$. Note that a canonical graph can only have one root). Since $M$ is $w$-canonical, we can simply choose $G$ to be the set of those nodes at level $p+k+3$ in the model $M$ which are also elements of $C$. Since in a canonical model $M,$ for each node $v$ $\mxdepth{M}{v}$ is exactly the level of $v$ in $M,$ this is a set of polynomial size, and every path from $w$ to a node in $C$ passes through $G$. Further, every successor $v\in C$ of a node $u\in A$ is trivially a member of $G$. Since every node in $M$ can be reached from the root $w$ in at most $\md\phi$ steps, the claim follows.

Now consider the case where $\hvarphi$ is not satisfied on every strict line. For each world $u\in M,$ let $F_u:=\set{\psi\in\subf{\phi}\ \vert\ M,u\models\Diamond\psi},$ and let $W_u$ be a subset of successors of $u$ in $M$ such that for every $\psi\in F_u,$ there is a world $v\in W_u$ such that $(u,v)$ is an edge in $M,$ $M,v\models\psi,$ and $\card{W_u}\leq\card{\subf\phi}$. If it is possible to choose $v\in C,$ then do so.

Define $M_0:=\set{w},$ and for each $i\in\mathbb N,$ let $M_{i+1}:=\cup_{u\in M_i}W_u\setminus\cup_{j=0}^{i} M_j$ Now define $M'$ to be the restriction of $M$ to the worlds in $\cup_{i\in\mathbb N}M_i$. We show that $M',w\models\phi$ by proving that $M'$ satisfies the conditions of Lemma~\ref{lemma:modal restriction invariance}. By construction, we know that $w\in M'$. Now let $u\in M',$ and let $\psi$ be a subformula of $\phi$ with $M,u\models\Diamond\phi$. Then $\phi\in F_u,$ and therefore there is a world $v\in W_u$ with $M,u\models\phi$. Since $u\in M',$ there is some $i$ with $u\in M_i$. It follows that $W_u\subseteq\cup_{j=0}^{i+1}M_j\subseteq M',$ and hence $v\in M'$ as claimed. Since $M'$ is a restriction of $M,$ we know that $M'$ also is a $\K\hvarphi$-model, and inherits all size restrictions on submodels that we have already established. Therefore, we can, without loss of generality, assume that $M'=M.$

Now define $G:=\cup_{i=0}^{p+k+3}M_i\cap C$. By the proof of Lemma~\ref{lemma:few nodes with few predecessors}, since $p+k+3$ is a constant, we conclude that $G$ is polynomial in $\card\phi$. By construction, $G\subseteq C$. We prove that every node $u\in C$ can be reached on a path from a node in $C$. For every $u\in M,$ there exists a sequence of nodes $w=u_0,u_1,\dots,u_n=u$ such that for all relevant $i,$ $u_{i+1}\in W_{u_i},$ $(u_i,u_{i+1})$ is an edge in $M,$ and $u_i\in M_i\setminus\cup_{j=0}^{i-1}M_i$. Let $t$ be minimal such that $u_t\in C$. Since $w\notin C,$ it follows that $t\ge 1,$ and we know that such a $t$ exists, since $u_n=u\in C$. It suffices to prove that $u_t\in G$. Since $u_t\in C,$ it remains to prove that $u_t\in\cup_{i=0}^{p+k+3}M_i$. Assume that this is not the case, from the above it then follows that $t\ge p+k+4$. By choice of nodes, we know that $u_{p+k+3}$ has a $p+k+3$-step predecessor, and since $u_{p+k+3}$ is a predecessor of $u,$ $u_{p+k+3}$ also has a $\max((q-k),(r-k-1))$-step successor. Since $u_{p+k+3}\in M_{p+k+3}\setminus M_0,$ we know that $u_{p+k+3}\neq w,$ therefore it follows that $u_{p+k+3}\in C,$ a contradiction to the minimality of $t.$

Now, let $u\in A,$ and let there be some $v\in C$ such that $(u,v)$ is an edge in $M,$ and $M,v\models\psi$ for some $\psi\in\subf\phi$. Since $u\in A,$ we know that $u=w$ or $u$ does not have a $p+k+3$-step predecessor, and hence $u\in M_i$ for some $i\leq p+k+2$. By construction of $M_{i+1},$ there is some world $v'\in\cup_{j=0}^{i+1} M_j$ which satisfies the requirements, since the successors are chosen to be from $C$ if possible. Since $v'\in \cup_{j=0}^{i+1} M_j\cap C\subseteq G,$ this proves that $G$ indeed satisfies the conditions.
\end{proof}

For every $g\in C,$ we define $M^C_g$ to be the set $M_g\cap C$. Obviously, $M^C_g$ is a graph with root $g,$ and is a restriction of $C$ (recall that $M_g$ is the restriction of $M$ to all vertices reachable from $g$ on a directed path). Hence, Proposition~\ref{prop:trivial subgraph property} and Fact~\ref{prooffact:C models varphi i to i+1} immediately imply the following:

\begin{prooffact}\label{prooffact:M_C_g models varphi i to i+1}
Let $g\in C$. Then $M^C_g$ is a graph with root $g$ such that $M^C_g\models\hvarphi^{i\rightarrow i+1}$ for all $i\ge 0.$
\end{prooffact}

For any $g\in C$ and $i\in\mathbb{N},$ we define $$L^g_i:=\set{v\in M^C_g\ \vert\ \mpath giv}.$$

We observe the following fact:

\begin{prooffact}\label{prooffact:x in L(g)(i) y in L(g)(i+1) gives edge (x,y)}
Let $g\in C,$ $i\in\mathbb{N},$ $x\in L^g_i,$ and $y\in L^g_{i+1}$. Then $(x,y)$ is an edge in $M.$
\end{prooffact}

\begin{proof}
This immediately follows from Fact~\ref{prooffact:M_C_g models varphi i to i+1} and the definition of $L^g_i:$ By definition, it follows that $C\models\path gix,$ and $C\models\path g{i+1}y$. From Fact~\ref{prooffact:C models varphi i to i+1}, we know that $C\models\hvarphi^{i\rightarrow i+1},$ and hence there is an edge $(x,y)$ in $C$ as claimed.
\end{proof}

Graphs fulfilling the formula $\hvarphi^{1\rightarrow 2}$ are ``layered:''  Whenever there are nodes $x$ and $y$ such that there is a path of length $i$ from the root to $i,$ and of length $i+1$ to the node $y,$ then there is an edge between $x$ and $y$. From the definition, it is obvious that for any subgraph $X$ of $C$ with root $g,$ that $X$ is canonical, if for every pair of natural numbers $i\neq j,$ it follows that $L^g_i\cap L^g_j\cap X=\emptyset$. We make a distinction between those elements in $G$ which lead to a canonical graph, and those which do not. In light of Lemma~\ref{lemma:modal restriction invariance with levels}, it is natural that we do not need to look at the entire graph, but can ignore nodes which do not have ``short'' paths from the root of the graph leading to it. For a natural number $b,$ we say that $X$ is \emph{$b$-canonical}, if for all $0\leq i<j\leq b$ it holds that $L^g_i\cap L^g_j\cap X=\emptyset.$

We define the following:

\begin{eqnarray*}
G_{\mathrm{can}}&:=&\set{g\in G\ \vert\ M^C_g\mathtext{ is }\md\phi\mathtext{-canonical}}, \\
G_{\mathrm{non-can}}&:=&\set{g\in G\ \vert\ M^C_g\mathtext{ is not }\md\phi\mathtext{-canonical}}.
\end{eqnarray*}

It is obvious that $G=G_{\mathrm{can}}+G_{\mathrm{non-can}}$. For each $g\in G_{\mathrm{non-can}},$ let $i(g),j(g)$ denote natural numbers, and let $n(g)$ denote some node such that $n(g)\in L^g_{i(g)}\cap L^g_{j(g)},$ with $0\leq i(g)<j(g)\leq\md\phi,$ and $i(g)$ is minimal with these properties. In particular, observe that $i(g)$ and $j(g)$ are polynomial in $\card\phi.$

We now show that for $g\in G_{\mathrm{non-can}},$ the graph $M^C_g\setminus M^C_{n(g)}$ still has root $g,$ unless it is empty. It particular, this means that it makes sense to ask if these graphs are $\md\phi$-canonical.

\begin{prooffact}\label{prooffact:reduces models still rooted}
Let $g\in G_{\mathrm{non-can}},$ such that $M^C_g\setminus M^C_{n(g)}$ is not empty. Then $M^C_g\setminus M^C_{n(g)}$ has root $g.$
\end{prooffact}

\begin{proof}
Assume that this is not the case, i.e., that there is some $v\in M^C_g\setminus M^C_{n(g)},$ and in this graph, there is no path from $g$ to $v$. Since $v$ is an element of $M^C_g,$ a path from $g$ to $v$ exists in the original model $M$. Let $g\rightarrow v_1\rightarrow\dots\rightarrow v_i\rightarrow v$ be this path. Since the path does not exist in the graph $M^C_g\setminus M^C_{n(g)},$ it follows that one of the $v_j$ must be an element of $M^C_{n(g)}$. Since there is a path from this $v_j$ to $v,$ it follows that $v$ is an element of $M^C_{n(g)}$ as well, which is a contradiction.
\end{proof}

The nodes $n(g)$ can be seen as ``minimal non-canonical points'' in $M^C_g$. The following makes this more precise:

\begin{prooffact}
Let $g\in G_{\mathrm{non-can}}$. Then $M^C_g\setminus M^C_{n(g)}$ is $\md\phi$-canonical.
\end{prooffact}

\begin{proof}
Due to Fact~\ref{prooffact:reduces models still rooted}, we know that $M^C_g\setminus M^C_{n(g)}$ has root $g,$ and by definition this graph is a subset of $C$. Assume that it is not $\md\phi$-canonical. Then there exist natural numbers $i_1<i_2\leq\md\phi,$ and some node $v\in \left(L^g_{i_1}\cap L^g_{i_2}\right)\setminus M^C_{n(g)}$. Due to the minimality of $i(g),$ it follows that $i(g)\leq i_1,$ and hence $i(g)<i_2$. In particular, it follows by induction on Fact~\ref{prooffact:x in L(g)(i) y in L(g)(i+1) gives edge (x,y)} that there is a path from $n(g)\in L^g_{i(g)}$ to $v\in L^g_{i_2},$ i.e., $v\in M^C_{n(g)}$. This is a contradiction.
\end{proof}

As our next connectivity result, we show the following:

\begin{prooffact}\label{prooffact:forward and back edges}
 Let $g\in G_{\mathrm{non-can}},$ and let $x\in L^{n(g)}_{i_1},$ $y\in L^{n(g)}_{i_2},$ where $i_1$ and $i_2$ are natural numbers such that $i_2\equiv i_1+1\mathtext{ mod }(j(g)-i(g))$. Then $(x,y)$ is an edge in $M.$
\end{prooffact}

\begin{proof}
 Since $n(g)\in L^g_{i(g)}\cap L^g_{j(g)}$ and $j(g)>i(g),$ by induction on Fact~\ref{prooffact:x in L(g)(i) y in L(g)(i+1) gives edge (x,y)} we know that $\mpath{n(g)}{j(g)-i(g)}{n(g)},$ and obviously this implies that for every multiple of $j(g)-i(g),$ a path of that length exists in $M$ from $n(g)$ to $n(g)$ itself. Now assume that $i_2=i_1+1+j(j(g)-i(g))$ for some integer $j$. We make a case distinction:

\begin{description}
 \item[Case 1: $j\ge 0$]{The above implies that $\mpath{n(g)}{i_1+j(j(g)-i(g))}{x},$ and by choice of $j,$ we also know that $\mpath{n(g)}{i_1+j(j(g)-i(g))+1}{y}$. Since all of the involved nodes are elements of $C,$ and from Fact~\ref{prooffact:C models varphi i to i+1} we know that $C\models\hvarphi^{l\rightarrow l+1}$ for all $l\ge 0,$ this implies that there is an edge from $x$ to $y$ in $M,$ as claimed.}
\item[Case 2: $j<0$]{By choice of $i_1,$ we know that $\mpath{n(g)}{i_1}{x}$. By choice of $j,$ the above and since $-j$ is positive, we also know that $\mpath{n(g)}{i_1-j(j(g)-i(g))+j(j(g)-i(g))+1}{y}$. Hence the existence of the edge $(x,y)$ again follows from Fact~\ref{prooffact:C models varphi i to i+1}.}
\end{description}
\end{proof}

We therefore have the following structure of the model $M:$ By definition, $M=A\cup C\cup M\setminus S,$ and by choice of $G,$ every node in $C$ can be reached on a path from a node in $G$. Therefore, $C=\cup_{g\in G}M^C_g,$ and hence $M$ can be written as the union of sub-graphs as follows:

$$\displaystyle M=A
\cup\bigcup_{g\in G_{\mathtext{can}}} \underbrace{M^C_g}_{\md\phi-\mathrm{canonical}}
\cup\bigcup_{g\in G_{\mathtext{non-can}}} \underbrace{(M^C_g\setminus M^C_{n(g)})}_{\md\phi-\mathrm{canonical}}
\cup\bigcup_{g\in G_{\mathtext{non-can}}}M^C_{n(g)}
\cup (M\setminus S).$$

Recall that the set $G$ is polynomial in $\card\phi$. Due to Corollary~\ref{corollary:suffices to restrict middle vertices for np}, it suffices to restrict the ``middle part'' of this equation, i.e., the components except $A$ and $M\setminus S,$ to polynomial size in order to obtain the desired polynomial model. In order to do this, we will now prove further connectivity results for the sub-models $M^C_{n(g)}$. It is important to note that due to Proposition~\ref{prop:trivial subgraph property}, all of these submodels inherit all of the properties of their respective super-models which can be expressed by a universal first order formula.

The idea behind the construction is the following: For each node in the original model, we add enough successors in our new model to ensure that the new model satisfies the conditions of Lemma~\ref{lemma:modal restriction invariance with levels}. For the $\md\phi$-canonical submodels, we know that we can stop adding nodes at depth $\md\phi,$ since other nodes cannot be reached on a shorter path from the root. For the non-$\md\phi$-canonical submodels, we cannot do this, but we also do not need to: Due to Fact~\ref{prooffact:forward and back edges}, we have ``circular'' edges, hence we know that nodes in a ``low'' level also are in a ``high'' level.

For the construction, let $G_{\mathrm{can}}=\set{g_1,\dots,g_n},$ and let $G_{\mathrm{non-can}}=\set{g_{n+1},\dots,g_{n+m}},$ For $1\leq i\leq m,$ define $g_{n+m+i}:=n(g_{n+i})$ (note that in this case, by definition $g_{n+i}$ is a member of $G_{\mathrm{non-can}}$). For $1\leq i\leq n+2m,$ define

$$N_i:=
\begin{cases}
M^C_{g_i}, &\mathtext{ if }1\leq i\leq n, \\
M^C_{g_i}\setminus M^C_{n(g_i)}, &\mathtext{ if }n+1\leq i\leq n+m, \\
M^C_{g_i}, &\mathtext{ if }n+m+1\leq i\leq n+2m.
\end{cases}$$

Note that by definition and Fact~\ref{prooffact:reduces models still rooted} it follows that $N_i$ is a graph which is either empty or has root $g_i$. It also follows from the above that for $1\leq i\leq n+m,$ the graph $N_i$ is $\md\phi$-canonical. Since the union over all $N_i$ is the same as the union over all $M^C_g$ (for all $g\in G$), it follows that

$$\displaystyle M=A\cup\left(\bigcup_{1\leq i\leq n+2m}N_i\right)\cup(M\setminus S).$$

Note that $n+2m$ is polynomial in $\card\phi,$ since $G$ is. We need one helpful fact about the non-$\md\phi$-canonical $N_i$-models:

\begin{prooffact}\label{prooffact:non canonical models are union over restricted levels}
 Let $i\in\set{n+m+1,\dots,n+2m},$ i.e., let $N_i$ be not $\md\phi$-canonical. Then $N_i=\cup_{j=0}^{j(g_{i-m})-i(g_{i-m})}L^{g_i}_j.$
\end{prooffact}

\begin{proof}
The inclusion $\supseteq$ is obvious, since $N_i=M^C_{g_i}$. For the inclusion $\subseteq,$ let $u\in N_i$. Since, by definition, $g_i=n(g_{i-m}),$ and hence $N_i=M^C_{n(g_{i-m})},$ there is a minimal $j\in\mathbb{N}$ such that $u\in L^{n(g_{i-m})}_j$. If $j\leq j(g_i)-i(g_i),$ then the claim holds. Therefore, assume that $j>j(g_{i-m})-i(g_{i-m}),$ and thus $j':=j-(j(g_{i-m})-i(g_{i-m}))>0$. Since $L^{n(g_{i-m})}_{j}\neq\emptyset,$ we also know that $L^{n(g_{i-m})}_{j'-1}\neq\emptyset,$ therefore let $v\in L^{n(g_{i-m})}_{j'-1}$. By definition, it holds that $j=(j'-1)+1+(j(g_{i-m})-i(g_{i-m})),$ and in particular, $j\equiv(j'-1)+1\mathtext{ mod }(j(g_{i-m})-i(g_{i-m}))$. By Fact~\ref{prooffact:forward and back edges}, we know that there is an edge from $v\in L^{n(g_{i-m})}_{j'-1}$ to $u\in L^{n(g_{i-m})}_j$. This implies, by definition, that $u\in L^{n(g_{i-m})}_{j'},$ and by minimality of $j,$ we know that $j\leq j',$ and hence $j(g_{i-m})-i(g_{i-m})\leq 0,$ i.e., $j(g_{i-m})\leq i(g_{i-m}),$ a contradiction.
\end{proof}

We now define a series of models $M_i$ for $0\leq i\leq n+2m,$ which, step by step, integrate ``enough'' vertices of the original model $M$ to ensure that the formula $\phi$ still holds, but restrict the size to a polynomial. As the induction start, we define $M_0:=A$. For $i\ge 1,$ the construction is as follows:

\begin{itemize}
\item Add every world from $M_{i-1}$ to $M_i,$
\item add $g_i$ to $M_i,$
\item If $1\leq i\leq m+n,$ i.e., if $N_i$ is $\md\phi$-canonical, then for each $0\leq j\leq\md\phi+\mathtext{\textit{maxline}}+1,$ and each formula $\psi\in\subf\phi:$ if there is a world $v\in L^{g_i}_j$ such that $M,v\models\psi,$ then add \emph{one} of these worlds into $M_i.$
\item If $n+m+1\leq i\leq n+2m,$ i.e., if $N_i$ is not $\md\phi$-canonical, then for each $0\leq j\leq j(g)-i(g)+1,$ and each formula $\psi\in\subf\phi,$ perform the following: 
\begin{itemize}
\item If there is a world $v\in L^{g_i}_j$ such that $M,v\models\psi,$ then add \emph{one} of these worlds $v$ into $M_i$. 
\item If there is a world $u$ which has been added into one of the $\md\phi$-canonical submodels in the step above, and there is a successor $v\in L^{g_i}_j$ of $u$ such that $M,v\models\phi,$ then add \emph{one} of these worlds $v$ into $M_i.$
\end{itemize}
\end{itemize}

By construction, since $j(g)$ and $i(g)$ are polynomial in $\card\phi$ for each $g\in G_{\mathrm{non-can}},$ and $\card{\subf\phi}\leq\card\phi,$ there are only polynomially many worlds in $M_{n+2m}$. We now define $M'$ to be the model $M_{n+2m}\cup(M\setminus S)$. Then the set of worlds in $M'$ which have $\max(q-k,r-k-1)$-step successor and a $p+k+3$-step predecessor is polynomially bounded, since this is a subset of $M_{n+2m}$. Hence, due to Corollary~\ref{corollary:suffices to restrict middle vertices for np}, it suffices to show that $M',w\models\phi$ in order to exhibit a model of $\phi$ which is polynomial in size. In order to prove this, we show that $M'$ satisfies the conditions of Lemma~\ref{lemma:modal restriction invariance with levels}. Since $w\in A$ by definition and $M_0=A,$ $w$ is an element of $M'$ by definition. Therefore, let $u$ be an element of $M',$ and let $\psi$ be a subformula of $\phi,$ such that $M,u\models\Diamond\psi,$ and let $a$ be a natural number such that $\mpath wau,$ and $1+a+\md\psi\leq\md\phi$. It suffices to show that there is a world $v\in M'$ such that $M,v\models\psi,$ and $(u,v)$ is an edge in $M.$

Since $M,u\models\Diamond\psi,$ there is a world $v'\in M,$ such that $(u,v')$ is an edge in $M,$ and $M,v'\models\psi$. 

There are several cases to consider. If $v'\in A,$ then by the construction of $M',$ $v'$ is an element of $M'$ as well. Hence we can choose $v$ to be $v'$. If $u\in A,$ and $v'\notin A,$ then it holds that either $v'\in C$ or $v'\notin S$. If $v'\notin S,$ then we know that $v'\in M',$ and we can choose $v=v'$. Hence assume that $v'\in C$. In this case, by choice of $G,$ there is a world $v''\in G,$ such that $M,v''\models\psi,$ and $(u,v'')$ is an edge in $M$. Since $G\subseteq M',$ we can choose $v=v''$. If $u\in M\setminus S,$ then $v'\in M\setminus S$ holds as well.

Therefore, it remains to consider the case $u,v'\in C$. First, assume that there is no $i\in\set{1,\dots,n+2m}$ such that $u,v'\in N_i$. Since $u\in C,$ it follows that $u\in N_i$ for some $i$. Since $v'\notin N_i,$ $N_i$ must be $\md\phi$-canonical, since all other $N_i$ contain all successors from $C$ to nodes in $N_i$. It follows that $u\in M^C_{g_i},$ and therefore $v\in M^C_{g_i}$ holds as well. Since $v'\notin N_i,$ there is a non-$\md\phi$-canonical submodel $N_j$ such that $v'\in N_j=M^C_{n(g_{j-m})},$ and hence, due to Fact~\ref{prooffact:non canonical models are union over restricted levels}, there is some $j'\in\set{0,\dots,j(n(g_{j-m}))-i(n(g_{j-m}))}$ such that $v'\in L^{n(g_{j-m})}_{j'}$. In this case, a node $v$ fulfilling the requirements has been added to the model $M'$ due to the last condition in the construction. Now assume that there is some $i\in\set{1,\dots,n+2m},$ such that $u,v'\in N_i$. If it is possible to choose this $i$ in such a way that $N_i$ is not $\md\phi$-canonical, then we do so. 

In particular, since $(u,v')$ is an edge in $M,$ there is some natural number $j$ such that $u\in L^{g_i}_j,$ and $v'\in L^{g_i}_{j+1}$. If it is only possible to choose an $i$ leading to a $\md\phi$-canonical model $N_i,$ then choose $i$ and $j$ in such a way that $j$ is minimal with this property. We make a case distinction.

\begin{description}
\item[Case 1: $i\in\set{n+m+1,\dots,n+2m},$ i.e., $N_i$ is not canonical.]{By Fact~\ref{prooffact:non canonical models are union over restricted levels}, due to the minimality of $j,$ we know that $j\leq j(g_{i-m})-i(g_{i-m})$. Now, since $u\in L^{g_{i}}_{j}$ and $v'\in L^{g_{i}}_{j+1},$ it follows that $j+1\leq j(g_{i-m})-i(g_{i-m})+1$. Hence, by construction there is a $v\in L^{g_{i}}_{j+1}\cap M',$ such that $M,v\models\psi,$ and due to Fact~\ref{prooffact:x in L(g)(i) y in L(g)(i+1) gives edge (x,y)}, there is an edge $(u,v)$ in $M.$}
\item[Case 2: $i\in\set{1,\dots,n+m},$ i.e., $N_i$ is $\md\phi$-canonical.] {If $j+1\leq\md\phi+\mathtext{\textit{maxline}}+1,$ then, due to the construction of $M',$ a world $v$ from $L^{g_i}_{j+1}$ satisfying $\psi$ was added in the construction, and due to Fact~\ref{prooffact:x in L(g)(i) y in L(g)(i+1) gives edge (x,y)}, $(u,v)$ is an edge in $M.$

Therefore, assume that $j+1>\md\phi+\mathtext{\textit{maxline}}+1,$ i.e., $j\ge\md\phi+\mathtext{\textit{maxline}}+1$. Since $N_i$ is $\md\phi$-canonical, we know that $u$ cannot be reached from $g_i$ with a path shorter than $\md\phi$ steps. Due to the choice of $i$ and $j,$ we also know that $u$ does not appear in a non-$\md\phi$-canonical submodel (otherwise, since the non-$\md\phi$-canonical submodels are successor-closed there would be a non-$\md\phi$-canonical submodel $N_i$ containing both $u$ and $v',$ and we would have chosen this), and that in each $\md\phi$-canonical submodel where $u$ appears, it has depth of at least $\md\phi+1$. 

If $\hvarphi$ is satisfied on every strict line, we know that every node from $C$ can be reached from an element in $G$ with at most $\md\phi$ steps, which is a contradiction.

Now consider the case that $\hvarphi$ is not satisfied on every strict line. Since in the $\md\phi$-canonical submodels, strict lines having the length of the depth of the submodel appear, and $\hvarphi$ is satisfied in $M,$ we know that these submodels cannot have depth of more than $\mathtext{\textit{maxline}}$. Therefore, we know that $j+1\leq\mathtext{\textit{maxline}},$ a contradiction.
}
\end{description}

Hence, we know that $M',w\models\phi,$ concluding the proof of Theorem~\ref{theorem:above k to k+1 leads to np}.
\end{proof}

We will now show that Theorem~\ref{theorem:above k to k+1 leads to np} implies \NP-results for a number of related logics. The following theorem shows that modal logics for classes of frames which fulfill a natural generalization of the Euclidean property have the polynomial-size model property, and hence can be solved in \NP. Note that the case $k=l=1$ of Corollary~\ref{corollary:main np result corollary} follows from the main result of~\cite{hare07}. Our results and theirs are incomparable: They achieve the \NP-result for \emph{all} normal modal logics extending what in our notation is $\K{\hvarphi^{1\rightarrow 1}},$ where our results only hold for logics defined by universal Horn clauses (although the proof of Theorem~\ref{theorem:above k to k+1 leads to np} in most cases gives the \NP-result for all extensions of the logic which are defined by universal formulas over the frame language---the only exception is the case where the graph formula is satisfied on every strict line, note that this case needed special treatment in the proof of Theorem~\ref{theorem:above k to k+1 leads to np}, and relies on the tree-like property for these logics proven later in Theorem~\ref{theorem:tree-like models for horn logics}). We achieve \NP-results for many logics which are not extensions of $\K{\hvarphi^{1\rightarrow 1}},$ but do not prove these results for all extensions of these logics.

With the preceding theorem and the results on the implication for formulas of the form $\genvarphi klpqr,$ we obtain the following corollary:

\begin{corollary}\label{corollary:main np result corollary}
Let $\hvarphi$ be a universal Horn formula implying $\genvarphi klpqr$ for some $k,l,p,q,r\in\mathbb N,$ such that one of the following conditions holds:

\begin{itemize}
 \item $1\leq k,l.$
 \item $l=0$ and $k\ge 2.$
\end{itemize}

 Then $\K\hvarphi$ has the polynomial-size model property, and $\sat{\K\hvarphi}\in\NP.$
\end{corollary}

\begin{proof}
 First assume that $1\leq k,l$. Lemma~\ref{lemma:varphi k to l implications} shows that $\hvarphi^{k\rightarrow l}$ implies $\hvarphi^{k'\rightarrow k'+1}$ for some $k'\ge 1$. From Lemma~\ref{lemma:varphi k to l implication gives genvarphi implications}, we conclude that $\genvarphi klpqr$ implies $\genvarphi{k'}{k'+1}{p'}{q'}{r'}$ for some natural numbers $p',q',$ and $r'$. Therefore, $\hvarphi$ also implies $\genvarphi{k'}{k'+1}{p'}{q'}{r'},$ and due to Theorem~\ref{theorem:above k to k+1 leads to np}, $\K\hvarphi$ has the polynomial-size model property and $\sat{\K\hvarphi}\in\NP.$

For the case $l=0$ and $k\ge 2,$ observe that due to Lemma~\ref{lemma:genvarphi k to 0 implications}, $\K\hvarphi$ is also an extension of $\K{\genvarphi{k-1}{k^3}{p'}{q'}{r'}}$ for some constants $p',q',r'$. Due to the prerequisites, it holds that $1\leq k-1\leq k^3-2$. Hence, the result follows from the case above.
\end{proof}

Corollary~\ref{corollary:main np result corollary} covers all cases obtained from formulas of the form $\genvarphi klpqr$ where we possibly could expect the polynomial size property to hold: The requirement demanding that $k\ge 2$ is crucial, since the formula $\hvarphi^{1\rightarrow 0}$ is satisfied in any symmetric graph, and hence satisfiability problems for the corresponding logics are \PSPACE-hard due to Corollary~\ref{corollary:ufos satisfied in symmetric tree are pspace hard}. Additionally, logics defined by formulas of the form $\hvarphi^{0\rightarrow k}$ are \PSPACE-hard as well: This formula is satisfied in every reflexive, transitive graph, and therefore the complexity result follows from Theorem~\ref{theorem:ladner hardness cases for k-transitivity}. We therefore have proven that for Horn formulas defined by these single clauses, as soon as they are not satisfied in various combinations of strict, transitive, reflexive, and symmetric trees, they already imply the polynomial-size model property.

Up to now, we only considered the effect of a single universal Horn clause. However, there are many cases where the logics defined by individual formulas have a \PSPACE-hard satisfiability problem, but the complexity of the problem for logic defined by their conjunction drops to \NP. A well-known example for such a case is the modal logic $\logicname{K4B},$ which is the logic over graphs which are both symmetric and transitive. On their own, both of these properties lead to \PSPACE-hard logics, but their conjunction gives a logic which is in \NP. In our notation, it can easily be seen that $\logicname{K4B}=\K{\hvarphi^{1\rightarrow 0}\wedge\hvarphi^{0\rightarrow 2}}$. For this concrete example and many generalizations, the following Theorem gives this result.

\begin{theorem}\label{theorem:transitive and symmetric gives np generalization}
 Let $k\ge 2\in\mathbb{N},$ and let $p_1,q_1,r_1,p_2,q_2,r_2\in\mathbb{N}$. Let $\hvarphi$ be a universal Horn formula such that $\hvarphi$ implies $\genvarphi{1}{0}{p_1}{q_1}{r_1}\wedge\genvarphi{0}{k}{p_2}{q_2}{r_2}$. Then $\K\hvarphi$ has the polynomial size model property, and $\sat{\K\hvarphi}\in\NP.$
\end{theorem}
 
\begin{proof}
We first prove that $\hvarphi^{1\rightarrow 0}\wedge\hvarphi^{0\rightarrow k}$ implies $\hvarphi^{1\rightarrow (k-1)k}$ if $k$ is odd, and $\hvarphi^{2\rightarrow (k-1)k^2}$ if $k$ is even.
 Let $G$ be a graph satisfying $\hvarphi^{1\rightarrow 0}\wedge\hvarphi^{0\rightarrow k},$ and let $w=x_0,x_1$ ($,x_2$ if $k$ is even) and $w=y_0,\dots,y_{(k-1)k}$ ($,\dots,y_{(k-1)k^2}$ if $k$ is even) be nodes in $G$ such that $(x_i,x_{i+1})$ and $(y_i,y_{i+1})$ are edges for all relevant $i$. First, let $k$ be odd. Since $G\models\hvarphi^{1\rightarrow 0},$ it follows that $G$ is symmetric. 
By applying the $\hvarphi^{0\rightarrow k}$-property $k-1$ times, we get a path of length $k-1$ from $w$ to $y_{(k-1)k}$. Hence, since $(x_1,w)$ is an edge, it follows that $\gpath{x_1}{k}{y_{(k-1)k}},$ and due to the $\hvarphi^{0\rightarrow k}$-property, this implies that $(x_1,y_{(k-1)k})$ is an edge as required.
 
 Now let $k$ be even. Since $k\equiv 0\mod 2,$ and $k\ge 2,$ the symmetry of $G$ ensures that there is a path of length $k$ from $w$ to $x_2$. Hence, there is an edge $(w,x_2),$ and due to the symmetry of $G,$ an edge $(x_2,w)$. By $(k-1)$ applications of the $\hvarphi^{0\rightarrow k^2}$-property, we know that $\gpath{w}{k-1}{y_{(k-1)k^2}}$. Hence, we conclude that $\gpath{x_2}{k}{y_{(k-1)k^2}},$ and another application of the $\hvarphi^{0\rightarrow k}$-property gives the edge $(x_2,y_{(k-1)k^2}),$ as required.
 
 For the \NP-result, observe that $\genvarphi{1}{0}{p_1}{q_1}{r_1}\wedge\genvarphi{0}{k}{p_2}{q_2}{r_2}$ implies $\genvarphi{1}{0}{p}{q}{r}\wedge\genvarphi{0}{k}{p}{q}{r},$ where $p=\max(p_1,p_2),$ $q=\max(q_1,q_2),$ and $r=\max(r_1,r_2)$. Due to the above and Lemma~\ref{lemma:varphi k to l implication gives genvarphi implications}, we know that this formula implies $\genvarphi {1}{(k-1)k}{p'}{q'}{r'}$ for some $p',q',r'$ if $k$ is odd, and it implies $\genvarphi{2}{(k-1)k^2}{p''}{q''}{r''}$ for some $p'',q'',r''$ if $k$ is even. Further, if $k$ is odd, then it follows that $k\ge 3$. Hence, $(k-1)k\ge 6$. If $k$ is even, then, since $k\ge 2,$ we know that $(k-1)k^2\ge 4,$ and in both cases the \NP\ result follows from Corollary~\ref{corollary:main np result corollary}.
 \end{proof}

Theorem~\ref{theorem:transitive and symmetric gives np generalization} concludes our results about logics defined by specific Horn formulas. We now have collected all tools required to prove the main result of the paper, the complexity classification of satisfiability problems for logics defined by universal Horn formulas.

\subsection{The Main Result: A Dichotomy for Horn Formulas}\label{sect:main dichotomy result}

In this section, we show a dichotomy theorem, which classifies the complexity of the satisfiability problem for logics of the form $\K\hpsi,$ where $\hpsi$ is a conjunction of universal Horn clauses, into solvable in \NP\ and \PSPACE-hard. The classification is given in the form of the algorithm \algname\ presented in Figure~\ref{figure:classification algorithm}. In order to explain the algorithm, we need some more definitions.

For a set $\tl\subseteq\set{\texttt{refl},\texttt{symm},\texttt{trans}^k\ \vert\ k\in\mathbb{N}},$ we say that a graph $G$ satisfies the conditions of $\tl$ if it has the corresponding properties, i.e., if $\texttt{refl}\in\tl,$ then $G$ is required to be reflexive, if $\texttt{symm}\in\tl,$ then $G$ is required to be symmetric, and if $\texttt{trans}^k\in\tl,$ then $G$ is required to be $k$-transitive. A $\tl$-tree is a graph which can be obtained from a strict tree $T$ by adding exactly those edges required to make it satisfy the conditions of $\tl$ (note that this is a natural closure operator). Similarly, a $\tl$-line is the $\tl$-closure of a strict line. Finally, for a universal Horn clause $\varphi,$ let $\tl-T^{\mathrm{hom}}_{\hvarphi}$ denote the pairs $(\alpha,T)$ such that $T$ is a $\tl$-tree, and $\alpha\colon\preq\hvarphi\rightarrow T$ is a homomorphism. Intuitively, due to Proposition~\ref{prop:homomorphism and horn clauses}, this is the set of $\tl$-trees about which the clause $\hvarphi$ makes a statement, along with the corresponding homomorphisms. We first define the cases of universal Horn clauses leading to \NP-containment of the satisfiability problem:

\begin{definition}
 Let $\hvarphi$ be a universal Horn clause, and $\tl\subseteq\set{\refl,\symm,\trans^k\ \vert\ k\in\mathbb N}$. We say that $(\hvarphi,\tl)$ satisfies the \emph{\NP-case}, if one of the following occurs:
\begin{enumerate}
           \item $\conc\hvarphi=(x,y)$ for $x\neq y\in\preq\hvarphi,$ and there is $(\alpha,T)\in\tlhom$ such that there is no directed path connecting $\alpha(x)$ and $\alpha(y)$ in $T,$
           \item $\tlhom\neq\emptyset,$ and $\conc\hvarphi=\emptyset$ or the vertices from $\conc\hvarphi$ are different and not connected with an undirected path in $\preq\hvarphi$ (this also applies if $x$ or $y$ do not appear in $\preq\hvarphi$)
           \item $\conc\hvarphi=(x,y)$ for $x,y\in\preq\hvarphi,$ and there is a homomorphism $\alpha\colon\preq\hvarphi\rightarrow L,$ where $L$ is the $\tl$-line $(x_1,\dots,x_n),$ and there are $i,j$ such that $1\leq i\leq j-2$ such that $\alpha(x)=x_j$ and $\alpha(y)=x_i.$
           \item $\conc\hvarphi=(x,y)$ for $x,y\in\preq\hvarphi,$ and there exist \tl-lines $L_1=(x_0,\dots,x_{n_1})$ and $L_2=(y_0,\dots,y_{n_2})$ and homomorphisms $\alpha_1\colon\preq\hvarphi\rightarrow L_1,$ $\alpha_2\colon\preq\hvarphi\rightarrow L_2,$ such that $\alpha_1(x)=x_{i_1},$ and $\alpha_1(y)=x_{i_1-1},$ $\alpha_2(x)=y_{i_2},$ and $\alpha_2(y)=y_{i_2+k},$ where $k\ge 2.$
          \end{enumerate}
\end{definition}

We now define the properties of Horn clauses which do not lead to \NP-containment of the satisfiability problems on their own. Recalling Section~\ref{subsection:ladner}, it is natural that clauses which are satisfied in every reflexive, transitive, or symmetric tree are among these. This is captured by the following definitions: If $\hvarphi$ is a universal Horn clause with $\conc\hvarphi=(x,y)$ for $x=y$ or $x,y\in\preq\hvarphi,$ we say that $(\hvarphi,\tl)$ satisfies the \emph{reflexive case}, if $x=y$ or for every $(\alpha,T)\in\tl-T^{\textrm{hom}}_{\varphi}$ such that $(\alpha(x),\alpha(y))$ is not an edge in $T,$ it holds that $\alpha(x)=\alpha(y)$. $(\hvarphi,\tl)$ satisfies the \emph{transitive case for $k\in\mathbb{N}$}, if it does not satisfy the reflexive case, and for every $(\alpha,T)\in\tl-T^{\textrm{hom}}_{\hvarphi}$ such that $(\alpha(x),\alpha(y))$ is not an edge in $T,$ there is a path from $\alpha(x)$ to $\alpha(y)$ in $T,$ and there is some $(\alpha,T)\in\tlhom$ such that there is no edge $(\alpha(x),\alpha(y)),$ and $\alpha(y)$ is exactly $k$ levels below $\alpha(x)$ in $T$. Finally, $(\hvarphi,\tl)$ satisfies the \emph{symmetric case} if it does not satisfy the reflexive of the transitive case, and for every $(\alpha,T)\in\tlhom$ such that $(\alpha(x),\alpha(y))$ is not an edge in $T,$ there is an edge $(\alpha(y),\alpha(x))$ in $T$.

We can now state the classification theorem---the proof will follow from the individual results in this section. Note that the algorithm as stated can not be implemented directly, since it uses tests of the form if a given first-order formula is satisfied in certain infinite classes of graphs, and checks if certain elements are present in the infinite set $\tl-T^{\mathrm{hom}}_{\hvarphi}$. However, we believe that size-restrictions for the structures actually required to look at can be proven, and hence the algorithm hopefully can be implemented to give a deterministic decision procedure. However, the main usage of the algorithm is to show more general classification theorems, as we will see in Corollary~\ref{corollary:horn cunjunction without algorithm}.

\begin{theorem}\label{theorem:horn conjunction classification}
 Let $\psi$ be a conjunction of universal Horn clauses. Then the complexity of $\sat{\K\psi}$ is correctly determined by \algname.
\end{theorem}

A first look at the algorithm in Figure~\ref{figure:classification algorithm} reveals that it is obviously necessary to prove that the choices that the algorithm has to make always can be made: in the relevant situations, at least one of the ``reflexive,'' ``transitive,'' or ``symmetric'' conditions occurs. However, before starting with the proof, we explain the general idea of the algorithm and give an example for a logic which \algname\ proves to have a satisfiability problem in \NP.

\begin{figure}
\begin{algorithmic}[1]
 \STATE{$\tl:=\emptyset$}
 \WHILE{not done}
   \IF{every clause in $\hpsi$ is satisfied on every $\tl$-tree}
      \STATE{$\sat{\K\hpsi}$ is \PSPACE-hard}
   \ENDIF
   \STATE{Let $\hvarphi$ be a clause in $\hpsi$ not satisfied on every $\tl$-tree}
   \IF{$(\hvarphi,\tl)$ satisfies the \NP-case}
          \STATE{$\K\hpsi$ has the polynomial-size model property, and $\sat{\K\hpsi}\in\NP.$}
     \ELSE
       \STATE{$\conc{\hvarphi}=(x,y)$ for $x=y$ or $x,y\in\preq\hvarphi$}
        \IF{$(\hvarphi,\tl)$ satisfies the reflexive case}
          \STATE{$\tl:=\tl\cup\set{\texttt{refl}}$}
        \ELSIF{$(\hvarphi,\tl)$ satisfies the transitive case for $k\ge 2$ }
          \STATE{$\tl:=\tl\cup\set{\texttt{trans}^k}$}
        \ELSIF{$(\hvarphi,\tl)$ satisfies the symmetric case}
          \STATE{$\tl:=\tl\cup\set{\texttt{symm}}$}
        \ENDIF
     \ENDIF
     \IF{for some $k,$ $\set{\texttt{symm},\texttt{trans}^k}\subseteq\tl$}
       \STATE{$\K\hpsi$ has the polynomial-size model property, and $\sat{\K\hpsi}\in\NP.$}
    \ENDIF
  \ENDWHILE
\end{algorithmic}
\caption{The algorithm \algname}
\label{figure:classification algorithm}
\end{figure}

In the variable $\tl,$ the algorithm maintains a list of implications of the formula $\psi$. For example, \algname\ puts $\texttt{refl}$ into $\tl$ if it detects the formula $\psi$ to require a graph satisfying it to be ``near-reflexive'' (meaning, reflexive in all nodes with sufficient height and depth). Similarly, $\texttt{symm}\in\tl$ means that $\hpsi$ requires a graph to be ``near-symmetric,'' and $\texttt{trans}^k\in\tl$ means that $\hpsi$ requires a graph to be ``near-$k$-transitive'' (this will be made precise in Lemma~\ref{lemma:horn conjunction algorithm adds correctly to types-list}).

If at one point \algname\ detects that a clause $\hvarphi$ satisfies one of the \NP-conditions, then this means that the clause $\hvarphi,$ in addition with the requirements kept in $\tl,$ implies a graph-property which leads to the polynomial-size model property. It is well-known that the modal logic over the class of frames which are both transitive and symmetric has a satisfiability problem in \NP. Generalizing this, when \algname\ detects that $\hpsi$ requires ``near-symmetry'' and ``near-$k$-transitivity,'' this also leads to the polynomial-size model property and thus to \NP-membership of the satisfiability problem, applying Theorem~\ref{theorem:transitive and symmetric gives np generalization}.

\begin{wrapfigure}[9]{l}{6.5cm}
\includegraphics[scale=0.75]{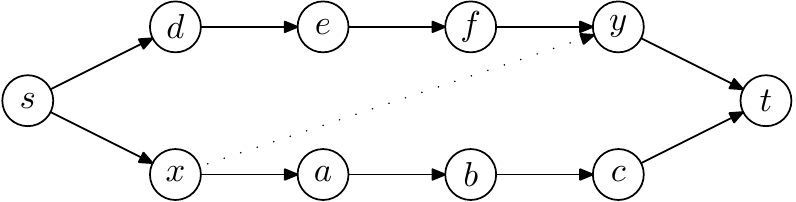}
\caption{Example formula}
\label{fig:applications section example graph 1}
\end{wrapfigure}
As an example, let $\hvarphi$ be the universal Horn clause with prerequisite graph as shown in Figure~\ref{fig:applications section example graph 1}, with $\conc\hvarphi=(x,y),$ and let $\hpsi$ be the Horn formula having $\hvarphi$ as its only conjunct. The algorithm \algname\ starts with $\tl=\emptyset,$ and hence in its first iteration, checks if $\hvarphi$ is satisfied in every strict tree. This is not the case, as Figure~\ref{fig:applications section example graph 2} shows (here, we simply marked each note with the names of the vertices which are preimages of the homomorphism): This is a homomorphic image of $\preq\hvarphi$ as a strict line (which in particular is a strict tree), in which the images of $x$ and $y$ are not connected with an edge. Therefore, $\hvarphi$ is not satisfied in this line, and hence not in every strict tree. However, it is clear that if we want to map $\preq\hvarphi$ into a strict tree, then all the vertices between $s$ and $t$ must be ``pairwise identified,'' just like in Figure~\ref{fig:applications section example graph 2}. 

\begin{wrapfigure}[8]{l}{6.5cm}
\includegraphics[scale=0.75]{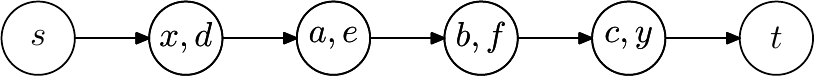}
\caption{Homomorphic image as strict line}
\label{fig:applications section example graph 2}
\end{wrapfigure}
Therefore the transitive case is satisfied, and the figure shows that this is true for $k=3$ (among possibly others). Hence \algname\ adds the element $\trans^3$ to $\tl$. Next it checks if $\hvarphi$ is satisfied in every $\set{\trans^3}$-tree. This again is not the case, and the homomorphic image of $\preq\hvarphi$ as a $\set{\trans^3}$-line in Figure~\ref{fig:applications section example graph 3} shows that $\hvarphi$ satisfies \NP-condition $3$ (we only included those lines added by the $\trans^3$-closure that are required for our function to be a homomorphism). Therefore, the logic $\K\hvarphi$ has the polynomial-size model property, and its satisfiability problem is in \NP.

This example demonstrates that in the run of \algname, a clause $\hvarphi$ can meet different cases depending on the content of the variable $\tl:$ In the situation that $\tl=\emptyset,$ the clause $\hvarphi$ satisfies the transitive case, but when $\tl=\set{\trans^3},$ this is no longer the case.

\begin{figure}
\includegraphics[scale=0.75]{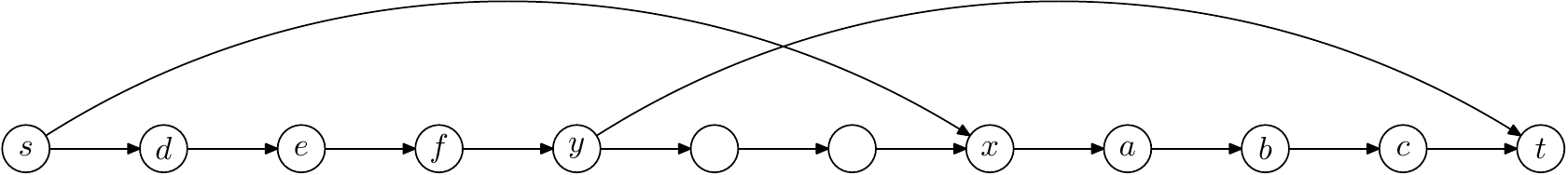}
\caption{Homomorphic image as $\set{\trans^3}$-line}
\label{fig:applications section example graph 3}
\end{figure}

As mentioned above, the first fact that we need to prove about \algname\ is that it is well-defined, that is, the case distinction between the ``reflexive,'' ``transitive'' and ``symmetric'' cases is complete.

\begin{lemma}\label{lemma:big classification algorithm is well-defined}
\algname\ is well-defined: in all relevant situations, at least one of the ``reflexive,'' ``transitive,'' and ``symmetric'' cases occurs.
\end{lemma}

\begin{proof}
 Assume that the algorithm is not well-defined, and let $\hpsi$ be an instance for which the algorithm behavior is unspecified. Let $\hvarphi$ be a clause for which none of the \NP-conditions and none of the reflexive, transitive, and symmetric conditions hold. Let $\tl$ be as determined by \algname\ when encountering $\hvarphi.$

Since the algorithm only chooses $\hvarphi$ if $\hvarphi$ is not satisfied on every $\tl$-tree, we know that due to Proposition~\ref{prop:homomorphism and horn clauses}, there is some $(\alpha,T)\in\tl-T^{\textrm{hom}}_{\hvarphi}$. Now assume that $\conc\hvarphi=\emptyset$ or $\conc\hvarphi=(x,y)$ for $x\neq y,$ and $x,y$ do not both appear in $\preq\hvarphi$. Then \NP-condition $2$ is satisfied, and we have a contradiction. Therefore we know that $\hvarphi$ is a universal Horn clause such that $\conc\hvarphi=(x,y)$ for some variables $x,y$ such that $x=y$ or $x,y\in\preq\hvarphi$ Since the reflexive case does not apply, we know that $x\neq y,$ and hence $x,y\in\preq\hvarphi.$

Since $\NP$-condition $1$ does not apply, we know that for every $(\alpha,T)\in\tl-T^{\mathrm{hom}}_{\hvarphi},$ there is a directed path connecting $\alpha(x)$ and $\alpha(y).$

Since the transitive case does not apply, we know that there is some $(\alpha,T)\in\tl-T^{\mathrm{hom}}_{\hvarphi}$ in which there is no path from $\alpha(x)$ to $\alpha(y)$. We therefore know, by the above, that there is a path from $\alpha(y)$ to $\alpha(x)$ in $T$. Hence, $T$ is not symmetric. Since $T$ is a $\tl$-tree, we know that $\texttt{symm}\notin\tl.$

 Since there is a path from $\alpha(y)$ to $\alpha(x)$ in $T$ but not vice versa, we know that $\alpha(x)\neq\alpha(y),$ and since there is a path from $\alpha(y)$ to $\alpha(x)$ in the (possibly reflexive and/or $S$-transitive for some set $S$) tree $T,$ this implies that $\alpha(x)$ is on a lower level in $T$ than $\alpha(y)$. Let $k$ denote the difference in levels between $\alpha(x)$ and $\alpha(y),$ then $k\ge 1$. By mapping $T$ into a $\tl$-line $L_1=(x_0,\dots,x_{n_1})$ with the homomorphism $\beta$ assigning each vertex $u$ the element $x_i,$ where $i$ is the level of $u$ in $T,$ we can, using the homomorphism $\alpha_1:=\beta\circ\alpha,$ map $\preq\hvarphi$ homomorphically into $L_1$ in such a way that $\alpha_1(x)=x_i,$ and $\alpha_1(y)=x_{i-k},$ where $i$ denotes the level of $\alpha(x)$ in $T.$

If $k\ge 2,$ then the line $L_1$ satisfies \NP-condition $3,$ a contradiction. Therefore, we know that $k=1,$ and thus $\alpha_1$ and $L_1$ satisfy \NP-condition $4$. 

Since $\hvarphi$ also does not satisfy the symmetric case, we know that there is some $(\alpha,T)\in\tlhom$ such that there is no edge $(\alpha(x),\alpha(y))$ and no edge $(\alpha(y),\alpha(x))$ in $T$. By $\NP$-condition $1,$ we know that there is a directed path connecting $\alpha(x)$ and $\alpha(y)$ in $T$. Since $\texttt{symm}\notin\tl,$ and $T$ is a \tl-tree, we know that $T$ is not symmetric. If $\texttt{refl}\in\tl,$ and $T$ is a $\tl$-tree, this implies that $T$ is reflexive, and hence $\alpha(x)\neq\alpha(y)$. If $\texttt{refl}\notin\tl,$ then we know that, since $\texttt{symm}$ also is not an element of $\tl,$ that no node in $T$ is connected to itself with a directed path, and therefore we also know that $\alpha(x)\neq\alpha(y)$. Now assume that there is a path from $\alpha(y)$ to $\alpha(x)$ in $T$. Then the shortest of these paths must have length $k\ge 2$ (since there is no edge $(\alpha(y),\alpha(x))$ and therefore we can map $T$ (and therefore $\preq\hvarphi$) homomorphically into a line $L_1$ as in the case above, where the distance of the image of $\alpha(x)$ and $\alpha(y)$ is $k,$ and since $\alpha(y)$ is mapped to a predecessor of $\alpha(x),$ this line satisfies \NP-condition $3,$ a contradiction. Therefore we know that in $T,$ there is a directed path from $\alpha(x)$ to $\alpha(y)$. Since $(\alpha(x),\alpha(y))$ is no edge in $T,$ we know that the shortest of these paths must have length $\ge 2$. Since $T$ is not symmetric, this implies that $\alpha(y)$ must be on a lower level in $T$ than $\alpha(x),$ and the difference in levels is $\ge 2$. Similarly, to the above, we can construct a \tl-line $L_2=(x_0,\dots,x_{n_2})$ and a homomorphism $\beta\colon T\rightarrow L_2,$ such that for $v\in T,$ if $i$ is the level of $v$ in $T,$ then $\beta(v)=x_i$. Let $\alpha_2$ denote the homomorphism $\beta\circ\alpha$. Due to the choice of $k$ and the definition of $\beta,$ we know that $\alpha_2(x)=x_i,$ and $\alpha_2(y)=x_{i+k}$ for some $i,$ and hence $\alpha_2$ and $L_2$ satisfy \NP-condition $4$. Since above, we already constructed $\alpha_1$ and $L_1$ satisfying \NP-condition $4,$ this implies that $\hvarphi$ satisfies \NP-condition $4,$ a contradiction.

Also note that in the transitive case, a $k\ge 2$ can always be chosen: Since by construction of the algorithm, $\hvarphi$ is not satisfied on every \tl-tree, there is some $(\alpha,T)$ such that $(\alpha(x),\alpha(y))$ is not an edge in $T$. Since  the reflexive case does not apply, we can choose $(\alpha,T)$ in such a way that $\alpha(x)\neq\alpha(y)$. Due to the transitive case, there is a path from $\alpha(x)$ to $\alpha(y)$ in $T,$ and hence the shortest of these paths must have length of at least $2,$ hence the difference of levels is at least $2$ (we know that even if $\texttt{symm}\in\tl,$ $\alpha(y)$ must be below $\alpha(x)$ in $T,$ otherwise \NP-condition $3$ would apply).
\end{proof}

Now that we know that \algname\ is well-defined, we show that it always comes to a halt and hence, generates an answer---this is not immediate: Note that in the example formula from Figure~\ref{fig:applications section example graph 1} discussed above, we already saw that it is possible for the algorithm to revisit a clause. Therefore one might consider it possible for the algorithm to repeatedly add the same element to the variables $\tl$ without coming to a halt, or the list $\tl$ to grow infinitely. In order to show that this does not happen, we prove a bound on the number of the elements that can be added into this list during the run of the algorithm.

A first helpful tool to show this is the following Proposition, which says that there are no ``redundant'' transitivity conditions which get added to \tl.

\begin{proposition}\label{prop:algorithm does not add already implied transitivity}
 Let $\trans^k$ be added to the variable $\tl$ by \algname. Then $\tl$ did not imply $k$-transitivity before adding $\trans^k.$
\end{proposition}

\begin{proof}
 Assume that $\tl$ already implied $\trans^k$. Since $\trans^k$ is added by \algname, there is a clause $\hvarphi$ with conclusion edge $(x,y)$ for $x,y\in\preq\hvarphi$ which satisfies the transitive condition for $k$ at this point of the algorithm's run. Therefore, there is some $(\alpha,T)\in\tlhom$ such that $(\alpha(x),\alpha(y))$ is not an edge in $T,$ and $\alpha(y)$ is exactly $k$ levels below $\alpha(x)$. Since by the assumption the conditions of $\tl$ already imply $k$-transitivity, the $k$-step path from $\alpha(x)$ to $\alpha(y)$ implies that $(\alpha(x),\alpha(y))$ is an edge in $T,$ a contradiction. 
\end{proof}

Now that we know that transitivity conditions which were already implied by the conditions present in \tl\ do not get added to the list, there is a clear strategy how we can prove that \tl\ does not grow infinitely: We first prove just how many transitivity conditions are implied by the conditions in \tl, and then show that from a certain point on, everything that will be added to \tl\ by \algname\ already is implied. In order to prove this, we need some technical results about implications of various forms of $S$-transitivity.

\begin{lemma}\label{lemma:test case 1}
 Let $k_1,k_2\in\mathbb{N}$. Then $\set{k_1,k_2}$-transitivity implies $(k_2+l\cdot (k_1-1))$-transitivity for all $l\ge 0.$
\end{lemma}

\begin{proof}
 We show the claim by induction on $l$. For $l=0,$ this holds trivially, since $\set{k_1,k_2}$-transitivity by definition implies $k_2$-transitivity. Now assume that it holds for $l,$ let $G$ be some graph which is $\set{k_1,k_2}$-transitive, and let $u_0,\dots,u_{k_2+(l+1)\cdot (k_1-1)}$ be nodes in $G$ such that $(u_i,u_{i+1})$ is an edge for all relevant $i$. We need to show that $(u_0,u_{k_2+(l+1)(k_1-1)})$ is an edge in $G$. Since by the induction hypothesis, $\set{k_1,k_2}$-transitivity implies $(k_2+l\cdot (k_1-1))$-transitivity, we know that $G$ is $(k_2+l\cdot (k_1-1))$-transitive. Therefore there is an edge $(u_0,u_{k_2+l\cdot(k_1-1)})$ in $G$. By choice of $u_i,$ there is a $(k_1-1)$-step path from $u_{k_2+l\cdot(k_1-1)}$ to $u_{k_2+(l+1)\cdot(k_1-1)}$. Therefore there is a $k_1$-step path from $u_0$ to $u_{k_2+(l+1)\cdot(k_1-1)},$ and since $G$ is $k_1$-transitive, this implies that $(u_0,u_{k_2+(l+1)\cdot(k_1-1)})$ is an edge in $G,$ as required.
\end{proof}

\begin{lemma}\label{lemma:test case 2}
 Let $m,k_1,e_1,d\in\mathbb N$ such that $(k_1-1)=e_1\cdot d$. Then $\set{k_1,m,m+d,\dots,m+(e_1-1)\cdot d}$-transitivity implies $(m+l\cdot d)$-transitivity for all $l\ge 0.$
\end{lemma}

\begin{proof}
 Let $l=p\cdot e_1+r$ for some $p,r\in\mathbb{N},$ $r<e_1$. Then $m+l\cdot d=m+(p\cdot e_1+r)\cdot d=m+r\cdot d+pe_1\cdot d=m+r\cdot d+p(k_1-1)$. Since $r<e_1,$ we know that $\set{k_1,m,m+d,\dots,m+(e_1-1)\cdot d}$-transitivity implies $(m+r\cdot d)$-transitivity. By Lemma~\ref{lemma:test case 1}, we know that $\set{k_1,m+r\cdot d}$-transitivity implies $(m+r\cdot d+p(k_1-1))$-transitivity, and since $k_1$-transitivity is obviously implied by the prerequisites, this proves the Lemma.
\end{proof}

\begin{lemma}\label{lemma:test case 4b}
 Let $(k_1-1)=e_1\cdot d,$ and let $n_0,\dots,n_{e_1-1}\in\mathbb{N}$ such that  $n_i\equiv n_j\mod d,$ and for $i\neq j,$ $n_i\not\equiv n_j\mod (k_1-1)$. Then there exists some $m\in\mathbb N$ such that $m\equiv n_i\mod d$ such that $\set{k_1,n_0,\dots,n_{e_1-1}}$-transitivity implies $(m+l\cdot d)$-transitivity for all $l\ge 0.$
\end{lemma}

\begin{proof}
 Let $S:=\set{k_1,n_0,\dots,n_{e_1-1}},$ and let $n_i=a_i\cdot(k_1-1)+n_i',$ where $n_i'<k_1-1$. Due to Lemma~\ref{lemma:test case 1}, $\set{k_1,n_i}$-transitivity implies $(n_i+l\cdot (k_1-1))$-transitivity for all $l\ge 0$. Therefore we can choose $a:=\max\set{a_0,\dots,a_{e_1-1}},$ and we know that $S$-transitivity implies $(a\cdot(k_1-1)+n_i')$-transitivity. Since for all $l,$ $n_i+l\cdot (k_1-1)$ is equivalent to $n_i$ modulo $k_1-1$ and modulo $d$ (since $d$ divides $k_1-1$), we can assume without loss of generality that $n_i=a\cdot(k_1-1)+n'_i$ for all $0\leq i\leq e_1-1$. 

Let $n_i'=t_i\cdot d+r,$ where $r<d$. Such numbers $t_i$ exist, since $n_i\equiv n_j\mod d$ for all $i,j$. Since $n_i'<k_1-1,$ we know that $t_i\cdot d+r<k_1-1=e_1\cdot d,$ and therefore $t_i\leq e_1-1$ for all $0\leq i\leq e_1-1.$

Define $m:=a\cdot(k_1-1)+r$. Note that since $d$ divides $k_1-1,$ this implies that $m\equiv r\mod d$. For the $n_i$ it holds that $n_i=a\cdot(k_1-1)+n_i',$ and again, since $d$ divides $k_1-1,$ we have that $n_i\equiv n_i'\mod d$. Now since $n_i'=t_i\cdot d+r,$ it follows that $n_i'$ is equivalent to $r$ modulo $d,$ and hence $m$ and $n_i$ are equivalent modulo $d$ for all $i$. Now note that

$$n_i-m=a\cdot(k_1-1)+n'_i-a\cdot(k_1-1)-r=n'_i-r=t_i\cdot d+r-r=t_i\cdot d.$$

Therefore, since $S$-transitivity implies $n_i$-transitivity for all $i,$ we know that $S$-transitivity implies $(m+t_i\cdot d)$-transitivity for all $i$. Since $t_i\leq e_1-1,$ and for $i\neq j,$ it also holds that $t_i\neq t_j$ (otherwise it would follow that $n_i'=n_j'$ and hence $n_i=n_j,$ a contradiction), it follows that $\set{t_0,\dots,t_{e_1-1}}=\set{0,\dots,e_1-1}$. We therefore know that for each $i\in\set{0,\dots,e_1-1},$ $S$-transitivity implies $(m+i\cdot d)$-transitivity, and since $k_1\in S,$ the claim follows from Lemma~\ref{lemma:test case 2}.
\end{proof}

The following lemma is the final of our results about implied transitivity conditions. Its technical formulation hides the fact that the statement of the lemma is actually rather natural. As an example, consider the case where $k_1=5,$ and $k_2=7$. Then $\gcd(k_1-1,k_2-1)=2,$ and then the lemma says that there is some odd number $m$ such that for all odd numbers $m+2\cdot l,$ $\set{5,7}$-transitivity implies $(m+2\cdot l)$-transitivity. The key idea is that then, once $\trans^5$ and $\trans^7$ are elements of $\tl,$ we know that due to Proposition~\ref{prop:algorithm does not add already implied transitivity}, \algname\ does not add any more $\trans^k$-conditions to \tl\ anymore for odd $k$s which are greater than or equal to $m$. Therefore there are only finitely many odd $k$s such that \algname\ can add $\trans^k$ from this point on. If \algname\ would add an infinite number of $\trans^k$-conditions, then one of them must therefore be one where $k$ is even. Assume that this already happens with $k_3,$ i.e., that $k_3$ is an even number. Then $k_3-1$ is odd, and therefore the greatest common divisor of $k_1-1,k_2-1,$ and $k_3-1$ is $1$. By the lemma, we therefore know that there is some $m$ such that $\set{5,7,k_3}$-transitivity implies $m'$-transitivity for all $m'\ge m,$ and then \algname\ can only add $\trans^k$-conditions for $k\leq m',$ and this is only a finite number of possibilities. Thus we know that \algname\ must stop adding transitivity conditions at some point.

With the above proof-strategy in mind, it is clear that the statement of the following lemma is crucial in restricting the number of $\trans^k$-conditions added by \algname.

\begin{lemma}\label{lemma:test case 5}
 Let $k_1,\dots,k_n\in\mathbb{N},$ and let $d:=\gcd((k_1-1),\dots,(k_n-1))$. Then there exists some $m\in\mathbb N$ such that $m\equiv 1\mod d$ and $\set{k_1,\dots,k_n}$-transitivity implies $(m+l\cdot d)$-transitivity for all $l\ge 0.$
\end{lemma}

\begin{proof}
Let $S=\set{k_1,\dots,k_n}$. We show the claim inductively on $n$. If $n=1,$ then $d=k_1-1,$ and by Lemma~\ref{lemma:test case 1}, $\set{k_1,k_1}$-transitivity implies $(k_1+l\cdot(k_1-1))$-transitivity for all $l\ge 0,$ hence the claim follows with $m=k_1.$

Since we need the case $n=2$ explicitly in the induction, we prove it individually. Hence let $n=2,$ and let $(k_i-1)=e_i\cdot d,$ where $\gcd(e_1,e_2)=1$. 

From Lemma~\ref{lemma:test case 1}, we know that $\set{k_1,k_2}$-transitivity implies $k_1+i\cdot(k_2-1)=i\cdot k_2+k_1-i$-transitivity for all $i\ge 0$. Now note that for all $0\leq i\leq e_1-1:$

$$i\cdot k_2+(k_1-i)=i\cdot(k_2-1)+(k_1-1)+1=i\cdot e_2\cdot d+e_1\cdot d+1=d\cdot (i\cdot e_2+e_1)+1=:n_i.$$

Due to Lemma~\ref{lemma:test case 4b}, it suffices to show that all of these $n_i$ are equivalent to $1$ modulo $d,$ and for $0\leq i<j\leq e_1-1,$ $n_i\not\equiv n_j\mod (k_1-1)$. It is obvious that all $n_i$ are equivalent to $1$ modulo $d$. Now assume that there are $0\leq i<j\leq e_1-1$ such that $n_i\equiv n_j\mod (k_1-1)$. Then the following holds:

$$
\begin{array}{rll}
 d\cdot(j\cdot e_2+e_1) +1 & \equiv d\cdot (i\cdot e_2+e_1) +1 & \mod (k_1-1) \\
 d\cdot(j\cdot e_2+e_1) -d(i\cdot e_2+e_1) & \equiv 0  & \mod (k_1-1) \\
 d\cdot (j\cdot e_2-i\cdot e_2) & \equiv 0  & \mod (k_1-1) \\
 d\cdot e_2\cdot (j-i)& \equiv 0  & \mod (k_1-1) \\
 d\cdot e_2\cdot(j-i) & = t\cdot d\cdot e_1 &\mathtext{ for some }t \\
 e_2\cdot(j-i) & = e_1\cdot t \\
 e_2\cdot(j-i) & \equiv 0 & \mod e_1 \\
 j-i & \equiv 0 & \mod e_1,\mathtext{ since }\gcd(e_1,e_2)=1

\end{array}
$$

This is a contradiction, since due to the above, we know that $0<j-i\leq e_1-1$. This completes the proof for the case $n=2.$

Now inductively assume that the lemma holds for $n\ge 2$. Let $d':=\gcd(k_1-1,\dots,k_n-1)$. Recall that $d=\gcd(k_1-1,\dots,k_{n+1}-1)$. We first show that $d=\gcd(k_{n+1}-1,d'),$ which is a standard fact from number theory. Obviously, $d$ divides both $k_{n+1}$ and $d',$ and hence $d\ \vert\ \gcd(k_{n+1}-1,d')$. On the other hand, let $g$ be a common divisor of $k_{n+1}-1$ and $d'$. Since $g$ divides $d',$ $g$ is also a divisor of $k_1-1,\dots,k_n-1,$ and since $g$ divides $k_{n+1}-1,$ it is a common divisor of all $k_i-1$. Therefore, $g\ \vert\ d$. In particular, this holds for $g=\gcd(k_{n+1}-1,d'),$ and therefore $\gcd(k_{n+1}-1,d')\ \vert\ d$. Since the other direction holds due to the above, we have shown that $d=\gcd(k_{n+1}-1,d').$

Due to the induction hypothesis, we know that there is a natural number $m'$ such that $m'\equiv 1\mod d',$ and for all $l'\ge 0,$ $S$-transitivity implies $(m'+l'\cdot d')$-transitivity. Let $m'=q\cdot d'+1,$ and let $l'$ be a natural number such that $\gcd(q+l',k_{n+1}-1)=1$. Then, by choice of $m',$ we know that $S$-transitivity implies $(m'+l'\cdot d')$-transitivity. Let $e:=\gcd(m'+l'\cdot d'-1,k_{n+1}-1)$. Since $S$-transitivity also implies $k_{n+1}$-transitivity, we know from the case $n=2$ that there exists a natural number $m$ such that $m\equiv 1\mod e,$ and for all $l\ge 0,$ $S$-transitivity implies $(m+l\cdot e)$-transitivity. In order to prove the lemma, it therefore suffices to show that $e=d$. It holds that 
$$e=\gcd(m'+l'\cdot d'-1,k_{n+1}-1)=\gcd(q\cdot d'+1+l'\cdot d'-1,k_{n+1}-1)=\gcd((q+l')\cdot d',k_{n+1}-1).$$

Since $d$ is a divisor of $d'$ and of $k_{n+1}-1,$ it follows that $d\ \vert\ e$. On the other hand, $e$ is a common divisor of $(q+l')\cdot d'$ and $k_{n+1}-1$. Since $\gcd(q+l',k_{n+1}-1)=1,$ we know that $e$ divides $d'$. Since $e$ also divides $k_{n+1}-1,$ this implies that $e$ is a divisor of $\gcd(d',k_{n+1}-1)=d$. Therefore $e\ \vert\ d$ and $d\ \vert\ e,$ and hence we have proven that $e=d,$ as claimed.
\end{proof}

We can now prove that there are only finitely many additions of the form $\trans^k$ during the algorithm's run, using the arguments from the discussion before Lemma~\ref{lemma:test case 5}:

\begin{lemma}\label{lemma:algorithm adds only finitely many transitivity conditions}
 Let $\hpsi$ be a universal Horn formula. Then $\algname,$ on input $\hpsi,$ only adds finitely many conditions of the form $\trans^k$ to $\tl.$
\end{lemma}

\begin{proof}
 Assume that this is not the case. Due to Proposition~\ref{prop:algorithm does not add already implied transitivity}, we know that each added $\trans^k$-element is not implied by the previous elements, in particular, no element is added twice. Hence there is an infinite sequence $(k_n)_{n\in\mathbb N}$ such that all $\trans^{k_n}$ are added to $\tl,$ they are added in this order, and for $n\neq m,$ we have $k_n\neq k_m.$

For $n\in\mathbb N,$ let $d_n:=\gcd(k_1-1,\dots,k_n-1)$. Then $d_n$ is obviously decreasing, and bounded by $1$. Therefore, the sequence converges, i.e., there is some $d,n_0\in\mathbb N$ such that $d_n=d$ for all $n\ge n_0$. Due to Lemma~\ref{lemma:test case 5}, we know that there is some $m$ such that $\set{k_1,\dots,k_{n_0}}$-transitivity implies $(m+l\cdot d)$-transitivity for all $l\ge 0,$ and $m\equiv 1\mod d$. Since the sequence $(k_n)_{n\in\mathbb{N}}$ is infinite and no element is repeated, there is some $n\ge n_0$ such that $k_n\ge m$. Since $\gcd(k_1-1,\dots,k_{n_0}-1,\dots,k_n-1)=d_n=d,$ we know that $d$ divides $k_n-1$. Therefore, let $k_n=d\cdot p+1$. Since $m\equiv 1\mod d,$ let $m=d\cdot q+1$. Since $k_n\ge m,$ we know that $p\ge q$. Hence it follows that

$$k_n=d\cdot p+1=d(p-q+q)+1=d\cdot(p-q)+d\cdot q+1=d\cdot (p-q)+m.$$

Due to the choice of $m$ and since $p-q\ge 0,$ we therefore know that $\set{k_0,\dots,k_{n_0}}$-transitivity implies $k_n$-transitivity. This is a contradiction to Proposition~\ref{prop:algorithm does not add already implied transitivity}, since $k_n$ is added after $k_0,\dots,k_{n_0}$. Therefore, there are only finitely many elements of the form $\trans^k$ added by the algorithm.
\end{proof}

Since we have now restricted the number of elements $\trans^k$ added to $\tl,$ we can prove that the algorithm halts on any input.

\begin{lemma}
 \algname\ always halts.
\end{lemma}

\begin{proof}
Assume that it does not halt for an instance $\hpsi$. 
Due to Lemma~\ref{lemma:algorithm adds only finitely many transitivity conditions}, we know that there are only finitely many elements of the form $\trans^k$ which are added to $\tl$ by \algname. Since the only other elements which can be added to $\tl$ are $\refl$ and $\symm,$ this implies that there are only finitely many elements of any type which are added to $\tl$. Now let $\tl$ be as determined by the algorithm. Since the size of $\tl$ is bounded and the variable never shrinks, this is well-defined. 

Since the algorithm does not halt, we know that there is a clause $\hvarphi$ in $\hpsi$ which is not satisfied in every $\tl$-tree, and $\hvarphi$ does not satisfy any of the $4$ \NP-conditions.

 Thus there is some $(\alpha,T)\in\tl-T^{\textrm{hom}}_{\hvarphi}$ such that $\hvarphi$ is not satisfied in $T$. Since we already showed that \algname\ is well-defined, and none of the \NP-cases applies (otherwise the algorithm would come to a halt), we know that one of the reflexive, transitive, or symmetric cases applies. In particular, $\conc\hvarphi=(x,y),$ and $x=y$ or $x,y\in\preq\hvarphi$ and $(\alpha(x),\alpha(y))$ is not an edge in $T$. If one of the cases corresponding to an element not in $\tl$ applies, then this leads to an enlargement of $\tl,$ a contradiction to the choice of $\tl$. Hence one of the cases corresponding to one of the elements already in $\tl$ applies. We make a case distinction:
\begin{description}
\item[Case 1: the reflexive case applies.]{In the case that $x=y,$ the clause is trivially satisfied in the reflexive graph $T$. Assume that $x\neq y$. In this case, $\alpha(x)=\alpha(y)$. But since $\texttt{refl}\in\tl,$ and $(\alpha,T)\in\tl-T^{\textrm{hom}}_{\hvarphi},$ we know that $T$ is a reflexive tree. Hence, the edge $(\alpha(x),\alpha(y))$ exists in $T,$ a contradiction.}
\item[Case 2: the transitive case applies for some $k.$]{In this case, there is a path of length $k$ from $\alpha(x)$ to $\alpha(y)$ in $T$. But since $\texttt{trans}^k\in\tl,$ and $(\alpha,T)\in\tl-T^{\textrm{hom}}_{\hvarphi},$ we know that $T$ is a $k$-transitive tree. Hence, the edge $(\alpha(x),\alpha(y))$ exists in $T,$ a contradiction.}
\item[Case 3: the symmetric case applies.]{In this case, there is an edge from $\alpha(y)$ to $\alpha(x)$ in $T$. But since $\texttt{symm}\in\tl,$ and $(\alpha,T)\in\tl-T^{\textrm{hom}}_{\hvarphi},$ we know that $T$ is a symmetric tree. Hence, the edge $(\alpha(x),\alpha(y))$ exists in $T,$ a contradiction.}
\end{description}

Since we have a contradiction in each case, this completes the proof.
\end{proof}

Now that we know that the algorithm is both well-defined and comes to a halt, it remains to prove its correctness. The case in which the algorithm states \PSPACE-hardness is easily seen to be correct:

\begin{lemma}\label{lemma:horn conjunction algorithm correct on PSPACE}
 If for a universal Horn formula $\hpsi,$ \algname\ states that $\sat{\K\hpsi}$ is \PSPACE-hard, then this is true.
\end{lemma}

\begin{proof}
 The only possibility for the algorithm to state that the problem is \PSPACE-hard is when the WHILE-loop does not discover any clauses $\hvarphi$ anymore which are not satisfied in every $\tl$-tree, and $\tl$ does not contain both $\texttt{symm}$ and $\texttt{trans}^k$ for any $k\in\mathbb{N}$. Since each clause $\hvarphi$ in $\hpsi$ is satisfied in every $\tl$-tree, this implies that the conjunction $\hpsi$ is also satisfied in each $\tl$-tree. Hence one of the following cases occurs:
\begin{itemize}
 \item $\hpsi$ is satisfied in every strict tree (if $\tl=\emptyset$),
 \item $\hpsi$ is satisfied in every reflexive tree (if $\tl=\set{\texttt{refl}}$),
 \item there is a set $S\subseteq\mathbb{N}$ such that $\hpsi$ is satisfied in every $S$-transitive tree (if $\tl=\set{\texttt{trans}^k\ \vert\ k\in S}$),
 \item $\hpsi$ is satisfied in every symmetric tree (if $\tl=\set{\texttt{symm}}$),
 \item $\hpsi$ is satisfied in every tree which is both reflexive and symmetric (if $\tl=\set{\texttt{refl},\texttt{symm}}$),
 \item there is a set $S\subseteq\mathbb{N}$ such that $\hpsi$ is satisfied in every tree which is both reflexive and $S$-transitive (if $\tl=\set{\texttt{refl},\texttt{trans}^k\ \vert\ k\in S}$).
\end{itemize}

In each of these cases, the hardness result follows directly from Theorem~\ref{theorem:ladner hardness cases for k-transitivity}.
\end{proof}

We are now interested in the \NP-cases. We first show that if \algname\ adds one of the requirements $\texttt{symm},\texttt{refl},$ or $\texttt{trans}^k$ for some $k$ to the list $\tl,$ then the formula $\hpsi$ requires any graph $G$ satisfying $\hpsi$ to have the corresponding property (except for vertices with insufficient depth or height in the graph). More precisely, we show the following Lemma:

 \begin{lemma}\label{lemma:horn conjunction algorithm adds correctly to types-list}
 Let $\hpsi$ be a universal Horn formula, and let $\textit{type}\in\set{\textup{\texttt{refl}},\textup{\texttt{symm}},\textup{\texttt{trans}}^k\ \vert\ k\in\mathbb{N}}$ be added to $\tl$ by \algname\ on input $\hpsi$. Then the following holds:
\begin{itemize}
 \item If $\textit{type}=\textup{\texttt{symm}},$ then $\hpsi$ implies $\genvarphi10ppp$ for some $p\in\mathbb{N},$
 \item If $\textit{type}=\textup{\texttt{refl}},$ then $\hpsi$ implies $\genvarphi00ppp$ for some $p\in\mathbb{N},$
 \item If $\textit{type}=\textup{\texttt{trans}}^k,$ then $\hpsi$ implies $\genvarphi0kppp$ for some $p\in\mathbb{N}.$
\end{itemize}
\end{lemma}

\begin{proof}
 Inductively assume that all values added before $\textit{type}$ (if any) were ``correct'' in the sense that they satisfy the conditions of the lemma. Let the content of the variable $\tl$ directly before adding $\textit{type}$ be denoted with $\textit{prev-types-list}$. Let $\hvarphi$ be the clause in $\hpsi$ for which $\textit{type}$ was added. Then, by the construction of the algorithm, we know that $\hvarphi$ satisfies the case in the algorithm corresponding to $\textit{type},$ it particular, $\conc\hvarphi=(x,y)$ for some variables $x=y$ or $x,y\in\preq\hvarphi$. First assume that $x,y\in\preq\hvarphi.$

 We also know that $\hvarphi$ is not satisfied in every $\textit{prev-types-list}$-tree. Hence, let $(\alpha,T)\in\textit{prev-types-list}-T^{\textrm{hom}}_{\hvarphi}$ such that $(\alpha(x),\alpha(y))$ is not an edge in $T$. This must exist due to Proposition~\ref{prop:homomorphism and horn clauses}. By construction of \algname, in the case that $\textit{type}=\texttt{refl},$ we have that $\alpha(x)=\alpha(y),$ and in the case that $\textit{type}=\texttt{trans}^k,$ we can choose $(\alpha,T)$ in such a way that $\alpha(y)$ is is $k$ levels below $\alpha(x)$ in $T$. Finally, in the case that $\textit{type}=\texttt{symm},$ we know that there is an edge $(\alpha(y),\alpha(x))$ in $T$. Since the reflexive condition does not apply, we can choose $(\alpha,T)$ in such a way that $\alpha(x)\neq\alpha(y),$ and since $\texttt{symm}\notin\textit{prev-types-list},$ this implies that $\alpha(x)$ is in a lower level of $T$ than $\alpha(y)$. 

Since all additions before $\textit{type}$ were correct, we know that all the types in $\textit{prev-types-list}$ have been added correctly, i.e., we can assume that there is some $p$ such that 

\begin{itemize}
 \item If $\texttt{symm}\in\textit{prev-types-list},$ then $\hpsi$ implies $\genvarphi10ppp$ for some $p\in\mathbb{N},$
 \item If $\texttt{refl}\in\textit{prev-types-list},$ then $\hpsi$ implies $\genvarphi00ppp$ for some $p\in\mathbb{N},$
 \item If $\texttt{trans}^k\in\textit{prev-types-list},$ then $\hpsi$ implies $\genvarphi0kppp$ for some $p\in\mathbb{N}.$
\end{itemize}

We can use the same value $p$ for all formulas here by taking the maximum, since clearly, $\genvarphi klpqr$ implies $\genvarphi kl{p'}{q'}{r'}$ if $p'\ge p,$ $q'\ge q,$ and $r'\ge r$. Now let $L=(x_0,\dots,x_n)$ be the homomorphic image of $T$ under the homomorphism $\delta$ as a $\textit{prev-types-list}$-line, where $\delta$ assigns each $v\in T$ the node $x_i,$ with $i$ being the level of $v$ in $T$. Let $\alpha(x)$ be in the $a$th level of $T,$ and let $\alpha(y)$ be in the $n-b$th level of $T$ (where $n$ is the height of $T$). Due to the construction of $\delta,$ we know that $\delta(\alpha(x))=x_a,$ and $\delta(\alpha(y))=x_{n-b}$. First note that:

\begin{itemize}
 \item If $\textit{type}=\texttt{refl},$ then $\alpha(x)=\alpha(y),$ and hence $\delta(\alpha(x))=\delta(\alpha(y))$. It follows that $n=a+b.$
 \item If $\textit{type}=\texttt{trans}^k,$ then by choice of $T,$ $\alpha(y)$ is exactly $k$ levels below $x$ in $T,$ and therefore we know that $a+b+k=n$. 
 \item If $\textit{type}=\texttt{symm},$ then, since $\texttt{symm}\notin\textit{prev-types-list},$ we know that in $L,$ we only have edges between non-decreasing nodes. Since the symmetric case applies, we know that $(\alpha(y),\alpha(x))$ is an edge in $T$. Since $\alpha(x)\neq\alpha(y),$ $\alpha(x)$ is at a lower level than $\alpha(y)$ in $T,$ and it follows that  $a>n-b$. If $a>n-b+2,$ then it follows that the line $L$ satisfies \NP-condition $3,$ a contradiction. Therefore we know that $a=n-b+1.$
 \end{itemize}

We now show that

\begin{center}
$
\begin{array}{lll}
\hpsi\mathtext{ implies } \genvarphi{0}{0}{a+b+p}{a+b+p}{a+b+p} & \mathtext{ if }\textit{type}=\texttt{refl}, \\
\hpsi\mathtext{ implies } \genvarphi{0}{k}{a+b+p}{a+b+p}{a+b+p+k} & \mathtext{ if }\textit{type}=\texttt{trans}^k, \\
\hpsi\mathtext{ implies } \genvarphi{1}{0}{a+b+p}{a+b+p+1}{a+b+p} & \mathtext{ if }\textit{type}=\texttt{symm}. \\
\end{array}
$
\end{center}

For this, let $G$ be a graph such that $G\models\hpsi,$ by induction we know that the nodes of $G$ which have both a $p$-step predecessor and a $p$-step successor satisfy the conditions of $\textit{prev-types-list}$. To prove the claim, let $u,v$ be elements of $G$ which have both an $a+b+p$-step predecessor and an $a+b+p$-step successor, and

\begin{center}
$
\begin{array}{ll}
 u=v&\mathtext{ if }\textit{type}=\texttt{refl}, \\
 \mathtext{There is a }k\mathtext{-step path in }G\mathtext{ from }u\mathtext{ to }v&\mathtext{ if }\textit{type}=\texttt{trans}^k, \\
 \mathtext{There is an edge }(v,u)\mathtext{ in }G&\mathtext{ if }\textit{type}=\texttt{symm}. \\
\end{array}
$
\end{center}

In order to prove the lemma, we need to prove that $(u,v)$ is an edge in $G$. Since $G\models\hpsi,$ and $\hvarphi$ is a clause in $\hpsi,$ it suffices to prove that there is a homomorphism $\gamma\colon\preq\hvarphi\rightarrow G$ such that $\gamma(x)=u$ and $\gamma(y)=v$. Then there is an edge $(u,v)$ in $G$ due to Proposition~\ref{prop:homomorphism and horn clauses}. Since $\delta\circ\alpha\colon\preq\hvarphi\rightarrow L$ is a homomorphism and $\delta\circ\alpha(x)=x_a$ and $\delta\circ\alpha(y)=x_{n-b},$ it suffices to show that there is a homomorphism $\beta\colon L\rightarrow G$ such that $\beta(x_a)=u$ and $\beta(x_{n-b})=v$. Then the homomorphism $\gamma:=\beta\circ\delta\circ\alpha$ satisfies the necessary conditions. Let $w$ be an $a$-step predecessor of $u$ in $G$ such that $w$ has a $p$-step predecessor in $G,$ and let $t$ be a $b$-step successor of $v$ in $G$ which has a $p$-step successor in $G$. Both must exist due to the choice of $u,v$. Additionally, if $\textit{types}=\texttt{symm}$ and therefore $(v,u)$ is an edge in $G,$ choose $w$ to be an $a-1$-step predecessor of $v$ (in which case it is also a $a$-step predecessor of $u$), and in this case also let $t$ to be a $b-1$-step successor of $u$ (in which case it is also a $b$-step successor of $v$).

Now, define $y_0:=w,$ and let $y_1,\dots,y_n$ from $G$ be chosen in such a way that $y_a=u,y_{n-b}=v,$ and $y_n=t,$ and for all relevant $i,$ $(y_i,y_{i+1})$ is an edge in $G$. This is possible since $u$ and $v$ satisfy the conditions corresponding to $\textit{type}:$

\begin{itemize}
 \item If $\textit{type}=\texttt{refl},$ we know that $n-b=a,$ and $u=v$. Hence the nodes $y_i$ can be chosen satisfying the demanded conditions by choosing $y_0,\dots,y_a$ to be the nodes on the $a$-step-path from $w$ to $u=v,$ and $y_{n-b},\dots,y_n$ to be the nodes on the $b$-step path from $u=v$ to $t.$
\item If $\textit{type}=\texttt{trans}^k,$ we know that $n=a+b+k,$ and we can choose $y_0,\dots,y_a$ to be the nodes on the path from $w$ to $u,$ $y_{a},\dots,y_{n-b}$ to denote the $k$-step path from $u$ to $v$ (which exists due to the choice of $u$ and $v$), and $y_{n-b},\dots,y_n$ be the nodes on the $b$-step path from $v$ to $t.$
\item If $\textit{type}=\texttt{symm},$ then by the above we know that $a=n-b+1,$ and we know that there is an edge $(v,u)$ in $G$. We also know that in this case, $w$ is an $a-1$-step predecessor of $v$. Hence we can choose the nodes in the following way:
 Let $y_0,\dots,y_{a-1}=y_{n-b}$ be chosen as the nodes on the $a-1$-step path from $w$ to $v,$ and let $y_{n-b+1}=y_{a},\dots,y_n$ be the nodes on the $b-1$-step path from $u$ to $t$. Since $(v,u)$ is an edge in $G,$ this gives the edge $(y_{n-b},y_a)$ which is required since $n-b+1=a.$
\end{itemize}

We now construct the homomorphism $\beta:$ for each relevant $i,$ let $\beta(x_i):=y_i$. Then by construction, $\beta(x_a)=u$ and $\beta(x_{n-b})=v$. Hence it remains to prove that $\beta$ is a homomorphism. Since $L$ is a \textit{prev-types-list}-line, let $L_{\mathrm{strict}}$ be a strict line such that $L$ is the $\textit{prev-types-list}$-closure of $L_{\mathrm{strict}}$. Since $(y_i,y_{i+1})$ is an edge in $G$ for all relevant $i,$ it follows that $\beta\colon L\rightarrow G$ is a homomorphism. Since every $y_i$ has a $p$-step successor and a $p$-step predecessor in $G,$ we know that the subgraph $\set{y_0,\dots,y_n}$ satisfies the conditions from $\textit{prev-types-list},$ and since $L$ is the $\textit{prev-list-types}$-closure of $L_{\mathrm{strict}},$ this implies that for every edge present in $L,$ the images of the corresponding vertices are also connected with an edge in $G,$ and hence $\beta$ is indeed a homomorphism, finishing the proof of the lemma for the case that $x,y\in\preq\hvarphi.$

Now assume that $x\notin\preq\hvarphi$ or $y\notin\preq\hvarphi,$ and therefore $x=y$. Since there obviously is a homomorphism $\alpha\colon\preq\hvarphi\rightarrow T$ for some $\textit{prev-types-list}$-tree $T,$ the clause $\hvarphi$ forces every node in a graph containing a $\tl$-line of sufficient length to be reflexive. Since due to the induction hypothesis, we know that every model of $\psi$ of sufficient depth contains arbitrary long $\textit{prev-types-list}$-lines, this concludes the proof for the remaining case $x=y.$
\end{proof}

Due to Lemma~\ref{lemma:horn conjunction algorithm adds correctly to types-list}, we know that the list $\tl$ maintained by \algname\ is sensible, and we are now in a position to prove that the \NP-cases claimed by the algorithm are correct as well.

\begin{lemma}
 Let $\hpsi$ be a universal Horn formula. If \algname\ states that $\K\psi$ has the polynomial-size model property and $\sat{\K\hpsi}$ is in \NP, then this is true.
\end{lemma}

\textit{Proof.}
 There are two possibilities for the algorithm to claim the polynomial-size model property, and hence \NP-membership. First let us assume that for some $k\ge 2,$ both $\texttt{symm}$ and $\texttt{trans}^k$ were added to $\tl$. In this case, due to Lemma~\ref{lemma:horn conjunction algorithm adds correctly to types-list}, we know that there is some $p\in\mathbb{N}$ such that $\hpsi$ implies both $\genvarphi{0}{k}ppp$ and $\genvarphi10ppp$. Hence Theorem~\ref{theorem:transitive and symmetric gives np generalization} implies both the polynomial-size model property of the logic $\K\hpsi$ and the \NP-membership of its satisfiability problem.

The second case in which \algname\ claims the \NP-result is if it detects a clause $\hvarphi$ which satisfies one of the conditions $1-4$. We know by Lemma~\ref{lemma:horn conjunction algorithm adds correctly to types-list}, that for each element from $\tl,$ the formula $\hpsi$ implies a formula of the corresponding type. Hence we can assume that there is a natural number $p,$ such that the set of vertices in $G$ which have both a $p$-step predecessor and a $p$-step successor satisfy the conditions from $\tl$. We make a case distinction.

\textbf{$\hvarphi$ satisfies \NP-condition $1.$\ }In this case, we know that $\conc\hvarphi=(x,y)$ for some $x\neq y\in\preq\hvarphi,$ and that there is a pair $(\alpha,T)\in\tl-T^{\mathrm{hom}}_{\hvarphi}$ such that there is no directed path connecting $\alpha(x)$ and $\alpha(y)$ in $T$. Note that we can assume $\texttt{symm}\notin\tl,$ since otherwise every pair of vertices in $T$ would be connected with a directed path. First assume that it is possible to choose $T$ and $\alpha$ in such a way that $\alpha(x)\neq\alpha(y)$. In this case, since $T$ is a tree, we know that $\alpha(x)$ and $\alpha(y)$ have a common predecessor $w$ in $T$. Let $w$ be a common predecessor which is ``minimal,'' i.e., no node in a lower level than $w$ is a common predecessor. 

We construct a $\tl-$tree $T'$ as a homomorphic image of $T$ via the homomorphism $\beta$ as follows: Let $n$ be the height of the tree $T,$ and let $x_0,\dots,x_n$ be a $\tl$-line such that every element in $T$ is mapped on its corresponding level in $L,$ except the nodes on the path from $w$ to $\alpha(x)$ (excluding $w$) and the successors of $\alpha(x)$. Let $s$ be the level of $w$ in $T,$ i.e., let $\alpha(w)=x_s$. Let $w+k$ and $w+l$ be the levels of $\alpha(x)$ and $\alpha(y)$ in $T,$ since there is no directed path connecting $\alpha(x)$ and $\alpha(y),$ it follows that $k,l>0$. Now introduce nodes $y_{s+1},\dots$ such that $(x_s,y_{s+1})$ is an edge, and $(y_i,y_{i+1})$ is an edge for every relevant $i,$ and map the path from $w$ to $\alpha(x)$ and the successors to $\alpha(x)$ to the ``branch'' $y_{s+1},\dots$ (add as many of these nodes as the ``branch'' of $T$ requires). Now close the construction under the $\tl$-condition, and call the tree obtained in this way $T'$. Since $T'$ is the ``canonical homomorphic image'' of $T,$ it is again a \tl-tree, and in $T',$ there is a path of length $k$ from $\beta(w)$ to $\beta(\alpha(x))$ and a path of length $l$ from $\beta(w)$ to $\beta(\alpha(y))$. Let $\hvarphi'$ be the clause with prerequisite graph $T'$ and conclusion edge $(\beta(\alpha(x)),\beta(\alpha(y)))$. Since $\beta\circ\alpha$ is a homomorphism, we know from Proposition~\ref{prop:homomorphism and horn clause implication} that $\hvarphi$ implies $\hvarphi'.$

Intuitively, $\hvarphi'$ is the $\tl$-closure of a clause of the form $\genvarphi klsqn,$ with $q$ chosen according to the height of the tree $T$ and the length of the branch containing $\alpha(x)$ (without loss of generality, we assume that the branch containing $\alpha(x)$ does not contain the deepest node in the tree). We now show that $\hpsi\implies\genvarphi kl{s+p}{q+p}{n+p},$ the complexity result then follows from Corollary~\ref{corollary:main np result corollary}.

Hence, let $G$ be a graph such that $G\models\hpsi,$ and let $u$ and $v$ be vertices in $G$ such that $u$ and $v$ have a predecessor $w',$ and $w'$ has a $s+p$-step predecessor, $u$ has a $q+p-k$-step successor, and $v$ has an $n+p-l$-step successor, and there is a $k$-step path from $w'$ to $u,$ and an $l$-step path from $w'$ to $v$. Then, since these vertices satisfy the conditions of $\tl$ by Lemma~\ref{lemma:horn conjunction algorithm adds correctly to types-list}, we can homomorphically map the prerequisite graph of $T'$ into $G$ via the homomorphism $\gamma$ such that $\gamma(\beta(w))=w',$ $\gamma(\beta(\alpha(x)))=u,$ and $\gamma(\beta(\alpha(x)))=v$. The edges required in order for $\gamma$ to be a homomorphism exist because of the paths of the corresponding lengths connecting $w,$ $u,$ and $v,$ and because all of the relevant nodes in $G$ satisfy the $\tl$-conditions. Hence, by Proposition~\ref{prop:homomorphism and horn clauses}, we know that $(u,v)$ is an edge in $G,$ as required to show.

Therefore, assume that it is not possible to choose $(\alpha,T)$ in such a way that $\alpha(x)\neq\alpha(y),$ i.e., assume that in every $(\alpha,T)\in\tl-T^{\mathrm{hom}}_{\hvarphi},$ $\alpha(x)$ and $\alpha(y)$ are connected with a directed path or are identical. Since $\hvarphi$ satisfies the first \NP-condition, we know that there is a pair $(\alpha,T)\in\tl-T^{\mathrm{hom}}_{\hvarphi}$ such that there is no directed path connecting $\alpha(x)$ and $\alpha(y)$ in $T,$ and hence $\alpha(x)=\alpha(y),$ and this node is irreflexive in $T$. In particular, we know that $\texttt{refl}$ is not an element of $\tl$. We also know that $\symm\notin\tl,$ since otherwise, every node in the connected graph $T$ would be connected to any other with a directed path ($T$ obviously is not the irreflexive singleton, since $\alpha(x)\neq\alpha(y)$). Now define $\beta$ to be $T$'s canonical homomorphic mapping to a $\tl$-line $L=(x_0,\dots,x_n),$ it then follows that $\beta(\alpha(x))=\beta(\alpha(y))=x_i$ for some $i$. Since $x$ and $y$ are different nodes in $\preq\hvarphi,$ we can modify this line as follows: We introduce a new node $y_i$ which is a ``neighbor'' to $x_i,$ i.e., a node which is connected to all predecessors and all successors of $x_i.$

Call this line $L'$. Note that since $\tl$ only contains variations of transitivity, there is no condition requiring that $x_i$ is a reflexive node. Since we connected $y_i$ to all successors and predecessors of $x_i,$ the line $L'$ still satisfies all conditions from $\tl,$ and $x_i$ and $y_i$ are not connected with an edge in $L'$. We now construct a homomorphism $\gamma\colon\preq\hvarphi\rightarrow L',$ by defining $\gamma(z)=\beta(\alpha(z))$ for all $z\neq y,$ and $\gamma(y)=y_i$. Since the involved nodes are irreflexive and $\beta\circ\alpha$ is a homomorphism, and $y_i$ has all edges that $x_i$ has, we know that $\gamma$ is a homomorphism. Let $\hvarphi'$ be the universal Horn clause with prerequisite graph $L',$ and conclusion edge $(x_i,y_i)$. The homomorphism $\gamma$ and Proposition~\ref{prop:homomorphism and horn clause implication} show that $\hvarphi$ implies $\hvarphi'$. 

Note that the clause $\hvarphi'$ requires the following: For every pair of nodes $(x',y')$ in a graph $G$ satisfying the conditions of $\tl$ which have a common predecessor $w'$ and a common successor $z'$ with sufficient height and depth, there is an edge connecting $x'$ and $y'$. Since by Lemma~\ref{lemma:horn conjunction algorithm adds correctly to types-list}, every graph satisfying $\hpsi$ also satisfies the requirements from $\tl$ for nodes with sufficient depth and height, we can apply $\hvarphi'$ to nodes of sufficient depth and height in every such graph.

\begin{wrapfigure}[19]{l}{8cm}
\includegraphics[scale=0.75]{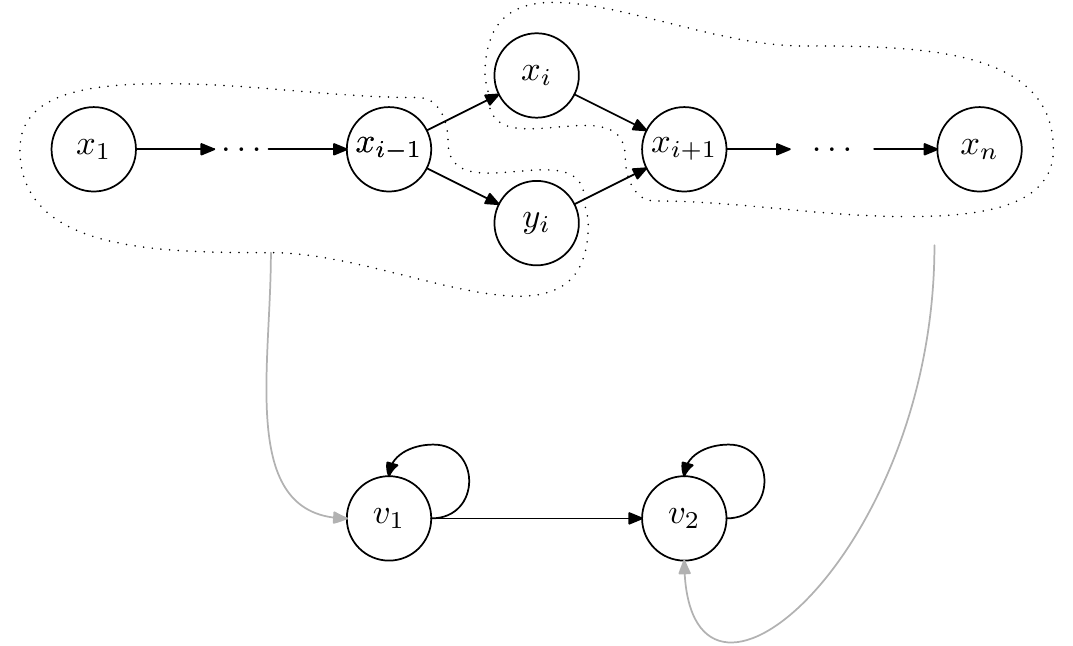}
\caption{The homomorphism for the proof that $C$ is symmetric}
\label{fig:symmetric proof for C}
\end{wrapfigure}
 Let $G$ be a graph satisfying $\psi,$ and let $C$ be the set of nodes in $G$ which have sufficient height and depth to be able to apply the conditions of $\hvarphi'$. We show that $C$ is reflexive, transitive, and symmetric. Note that since $G$ satisfies $\psi,$ and $\hvarphi$ is a clause in $\hpsi$ implying $\hvarphi',$ we know that $G$ satisfies $\hvarphi'$ as well.

We first show that $C$ is reflexive. Let $v$ be a node in $C$. Since $v$ satisfies the conditions of both $x_i$ and $y_i$ in the line $L',$ and since $L'=\preq\hvarphi'$ and $G$ satisfies $\hvarphi',$ we know that there is an edge $(v,v)$ in $G,$ and hence $v$ is reflexive.

For symmetry, let there be nodes $v_1,v_2\in C$ such that $(v_1,v_2)$ is an edge in $G$. By the above, we know that both of these nodes are reflexive, and in particular have unbounded height and depth in $G$. Therefore, we can map the predecessor graph of $\hvarphi',$ i.e., the line $L'$ to these nodes in such a way that for $j<i,$ $x_i$ is mapped to $v_1,$ $y_i$ is mapped to $v_1,$ and all $x_j$ for $j\ge i$ are mapped to $v_2$. Since in $L',$ all edges go from variables with lower indexes to variables with higher indexes, in this way we have constructed a homomorphism from $\preq{\hvarphi'}$ to $C,$ such that the conclusion edge $(x_i,y_i)$ of $\hvarphi'$ is mapped to the pair $(v_2,v_1)$. Hence, by Proposition~\ref{prop:homomorphism and horn clauses}, and since $G$ satisfies $\hvarphi',$ we know that there must be an edge $(v_2,v_1)$ in $C,$ as required.

\begin{wrapfigure}[18]{r}{8cm}
\includegraphics[scale=0.75]{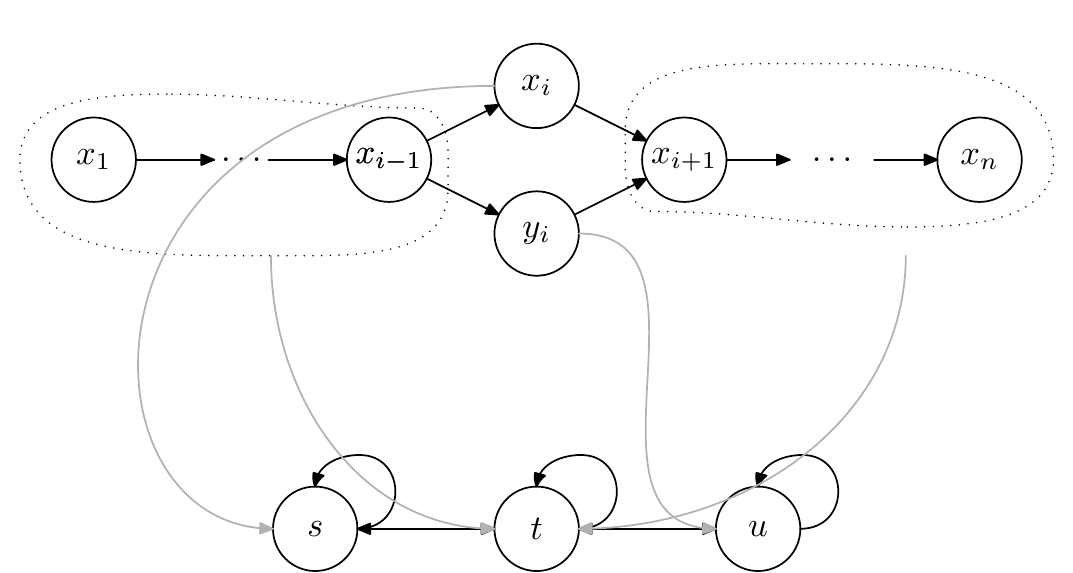}
\caption{The homomorphism for the proof that $C$ is transitive}
\label{fig:transitive proof for C}
\end{wrapfigure}
Finally, we show that $C$ is transitive. For this, assume that $(s,t)$ and $(t,u)$ are edges in $C$. Since we already proved symmetry for $C,$ we know that in this case, $(t,s)$ and $(u,t)$ are also edges. We can again construct a homomorphism mapping all elements of $L'$ to $t,$ except mapping $x_i$ to $s$ and $y_i$ to $u$. Since the nodes $s,t,$ and $u$ are reflexive in $C$ by the above, and there is no edge $(x_i,y_i)$ or $(y_i,x_i)$ in $L,$ this is a homomorphism. Since $G$ satisfies $\varphi',$ Proposition~\ref{prop:homomorphism and horn clauses} shows that $(s,u)$ is an edge in $C,$ as claimed.

In particular, since $C$ is reflexive, symmetric, and transitive, this means that $C\models\hvarphi^{1\rightarrow 1}$. Since $C$ is defined as the nodes which have some sufficient depth and height $p$ in the graph, and $C\models\hvarphi^{1\rightarrow 1},$ it follows that $G\models\genvarphi 11ppp$ for some $p,$ and hence $\psi\implies\genvarphi 11ppp$. Therefore, $\K\psi$ is a logic extending $\K{\genvarphi 11ppp},$ and therefore the complexity result as well as the polynomial-size model property follow from Corollary~\ref{corollary:main np result corollary}.

\textbf{$\hvarphi$ satisfies \NP-condition $2.$\ }{By the prerequisites, since $\hvarphi$ is not satisfied on every $\tl$-tree, we can homomorphically map $\preq\hvarphi$ onto a $\tl$-line. First assume that $\conc\hvarphi$ is empty. Then any graph $G$ satisfying the conditions of $\tl$ which has a line of more than this length does not satisfy the clause $\hvarphi$. Since every graph satisfying $\hpsi$ also satisfies $\hvarphi$ and the conditions of $\tl$ for nodes of sufficient depth and height, every graph satisfying $\hpsi$ of sufficient depth and height does not satisfy $\hvarphi,$ and hence not $\hpsi$. This implies that graphs satisfying $\hpsi$ can only be of depth limited by a constant, and hence Lemma~\ref{lemma:few nodes with few predecessors} immediately implies the polynomial-size model property.

Therefore assume that $\conc\hvarphi=(x,y)$ for variables $x,y,$ and assume that $x,y$ are not connected with an undirected path in $\preq\hvarphi$. Since $x$ and $y$ lie in different connected components of the underlying directed graph consisting of $\preq\hvarphi\cup\conc\hvarphi,$ and $\preq\hvarphi$ can be homomorphically mapped into a $\tl$-line, the ``left side'' of the implication $\hvarphi$ is satisfied by any pair of vertices $(u,v)$ in a graph satisfying the $\tl$-conditions of sufficient depth and height. Therefore, the subgraph $C$ containing all nodes with sufficient height and depth forms a universal subgraph, and the \NP-result follows by the same reasoning as in Case $1.$}

\textbf{$\hvarphi$ satisfies \NP-condition $3.$\ }{In this case, it follows by the same reasoning as in Lemma~\ref{lemma:horn conjunction algorithm adds correctly to types-list} that any graph satisfying $\hpsi$ also needs to satisfy the formula $\genvarphi k0ppp$ for some $p$ and some $k\ge 2$.  Hence the complexity result follows from Corollary~\ref{corollary:main np result corollary}.}

\textbf{$\hvarphi$ satisfies \NP-condition $4.$\ }{Again, with the proof of Lemma~\ref{lemma:horn conjunction algorithm adds correctly to types-list} it can easily be seen that $\hpsi$ implies formulas $\genvarphi 10ppp$ and $\genvarphi 0kppp$ for some $k\ge 2$ and some natural number $p$. Therefore the \NP-result follows from Theorem~\ref{theorem:transitive and symmetric gives np generalization}.}
\hfill$\Box$\bigskip

The preceding Lemmas have established that \algname\ is well-defined, always comes to a halt, and in each case produces the correct result. Hence, we have proven Theorem~\ref{theorem:horn conjunction classification}. This theorem and its proof now yield an interesting Corollary:

\begin{corollary}\label{corollary:horn cunjunction without algorithm}
 Let $\hpsi$ be a universal Horn formula. If one of the following cases applies:
\begin{itemize}
 \item $\hpsi$ is satisfied in every strict tree,
 \item $\hpsi$ is satisfied in every reflexive tree,
 \item $\hpsi$ is satisfied in every $S$-transitive tree for some $S\subseteq\mathbb{N},$
 \item $\hpsi$ is satisfied in every symmetric tree,
 \item $\hpsi$ is satisfied in every symmetric and reflexive tree,
 \item $\hpsi$ is satisfied in every $S$-transitive and reflexive tree for some $S\subseteq\mathbb{N},$
\end{itemize}
then $\sat{\K\hpsi}$ is \PSPACE-hard and $\K\hpsi$ does not have the polynomial-size model property. In all other cases, $\K\hpsi$ has the polynomial-size model property and $\sat{\K\hpsi}\in\NP.$
\end{corollary}

\begin{proof}
 We know by Theorem~\ref{theorem:horn conjunction classification} that \algname\ correctly determines the complexity of the problem $\sat{\K\hpsi}$. By the algorithm, it is obvious that the only arising cases are \PSPACE-hard and \NP. From the proof of Lemma~\ref{lemma:horn conjunction algorithm correct on PSPACE}, we know that the \PSPACE-cases all satisfy the statement of the corollary, and Theorem~\ref{theorem:ladner hardness cases for k-transitivity} shows that in these cases, we always have \PSPACE-hardness.

Finally note that all of our \NP-proofs also give the polynomial-size model property, and all \PSPACE-hardness proofs also show that this property does not apply.
\end{proof}

\subsection{Tree-like models for Horn logics}\label{sect:trees for horn}

In Section~\ref{sect:main dichotomy result}, we have shown that logics defined by universal Horn formulas have a satisfiability problem which is solvable in \NP, or is \PSPACE-hard. We now show a tree-like model property for these logics, which we will put to use in the next section by concluding \PSPACE-membership for a broad class of these logics. Recall that Corollary~\ref{corollary:every k-satisfiable formula satisfiable in strict tree} stated the tree-like model property for the modal logic $\logicname K:$ Every $\logicname K$-satisfiable formula has a tree-like model which also is a $\logicname K$-model (which is easy, since every model is a $\logicname K$-model). What we want to show is that a logic $\K\hpsi,$ where $\hpsi$ is a universal Horn formula satisfied on every reflexive/symmetric/$S$-transitive tree, has the following property: Every modal formula which has a $\K\hpsi$-model also has a $\K\hpsi$-model which is ``nearly'' the reflexive/symmetric/$S$-transitive closure of a strict tree. In order to prove this result, we first recall from the literature the concept of a bounded morphism, which allows us to prove modal equivalence of models.

\begin{definition}[\cite{blrive01}]
 Let $T$ and $M$ be modal models, and let $f\colon T\rightarrow M$ be a function. Then $f$ is a \emph{bounded morphism} if the following holds:
\begin{enumerate}
 \item[\textit{(i)}] For all $w\in T,$ $w$ and $f(w)$ satisfy the same propositional variables,
 \item[\textit{(ii)}] $f$ is a homomorphism,
 \item[\textit{(iii)}] if $(f(u),v')$ is an edge in $M,$ then there is some $v\in T$ such that $(u,v)$ is an edge in $T$ and $f(v)=v'.$
\end{enumerate}
\end{definition}

The most important feature of bounded morphisms is that they leave the modal properties of the involved models invariant:

\begin{proposition}[Proposition 2.14 from \cite{blrive01}]\label{prop:bounded morphism modal invariance}
Let $T$ and $M$ be modal models, and let $f\colon T\rightarrow M$ be a bounded morphism. Then for each modal formula $\phi$ and for each $w\in T,$ it holds that $T,w\models\phi$ if and only if $M,f(w)\models\phi.$
\end{proposition}

With bounded morphisms, we can now prove that our logics have ``tree-like'' models. We first recall a standard result from the literature about the logic \logicname{K}:

\begin{proposition}[Proposition 2.15 from \cite{blrive01}]\label{prop:tree like model exists with bounded morphism}
 Let $M,w\models\phi$ such that $M$ is rooted at $w$. Then there exists a tree-like model $T$ and a surjective bounded morphism $f\colon T\rightarrow M.$
\end{proposition}

We now generalize this result to universal Horn logics, and ``tree-like'' models. We already know from Corollary~\ref{corollary:horn cunjunction without algorithm} that for each universal Horn formula $\hpsi$ for which the logic $\K\hpsi$ has a satisfiability problem which cannot be solved in \NP, the formula $\hpsi$ is satisfied on every tree which additionally is closed under reflexivity, $S$-transitivity and/or symmetry. We now show that a version of the ``converse'' is also true: Not only are all of these trees models for the corresponding logics, but for every modal formula satisfiable in such a logic, we can find a model which is ``almost'' such a tree. The reason for the ``almost'' is that our Horn formulas usually will not imply a property like $S$-transitivity, but only $S$-transitivity for nodes at a certain depth in the graph, as shown in Lemma~\ref{lemma:horn conjunction algorithm adds correctly to types-list}. We will see in Sections~\ref{sect:pspace for non-transitive} and~\ref{sect:applications}, that the characterization of the involved models can be used to obtain \PSPACE\ upper bounds for a wide class of logics.

\begin{theorem}\label{theorem:tree-like models for horn logics}
 \begin{enumerate}
  \item Let $\hpsi$ be a universal Horn formula such that \algname\ returns \PSPACE-hard on input $\hpsi$. Let $\tl$ be as determined by \algname\ on input $\hpsi$. Since the algorithm determines the logic to be \PSPACE-hard, this is well-defined. Then 
for every modal formula $\phi$ which is $\K\hpsi$-satisfiable, there exists a $\K\hpsi$-model $T$ and a world $w\in T$ such that $T,w\models\phi$ and there is a strict tree $T_{\mathrm{strict}}$ such that $T$ and $T_{\mathrm{strict}}$ have the same set of vertices, and 

$$\edges{T_{\mathrm{strict}}}\subseteq\edges{T}\subseteq\edges{\tl(T_{\mathrm{strict}})},$$

where $\tl(T_{\mathrm{strict}})$ denotes the $\tl$-closure of $T_{\mathrm{strict}}.$
\item Let $\hpsi$ be a universal Horn formula such that $\hpsi$ is satisfied on every strict line. Then for every modal formula $\phi$ which is $\K\hpsi$-satisfiable, there exists a $\K\hpsi$-model $T$ and a world $w\in T$ such that $T,w\models\phi,$ and $T$ is $w$-canonical.
\end{enumerate}
\end{theorem}

\begin{proof}
We prove both claims with nearly the same construction, indicating when differences are required. Since $\phi$ is $\K\hpsi$-satisfiable, there is a $\K\hpsi$-model $M$ and a world $w\in M$ such that $M,w\models\phi$. 

For the first claim, let $\tl$ be as determined by \algname\ when started on input $\hpsi$. Since the algorithm does not return $\NP,$ we know that $\tl\subseteq\set{\refl,\symm}$ or $\tl\subseteq\set{\refl,\trans^k\ \vert\ k\in\mathbb N}$. Since the second \NP-condition does not apply, we know that for any clause $\hvarphi$ in $\hpsi$ such that there is a homomorphism $\alpha\colon\preq\hvarphi\rightarrow T$ for some $\tl$-tree $T,$ that $\conc\hvarphi=(x,y)$ for some variables $x,y$ with $x=y$ or $x,y\in\preq\hvarphi.$

For the second claim, we show that for every clause $\hvarphi$ in $\psi,$ if $\hvarphi$ is not satisfied in every $w$-canonical graph, then $\conc\hvarphi\neq\emptyset$. Obviously, the prerequisite graph of such a clause $\hvarphi$ can be mapped into some $w$-canonical graph, and therefore also into its image $L$ as a strict line. Since $\hpsi$ is satisfied on every strict line, so is $\hvarphi,$ and hence $\conc\hvarphi\neq\emptyset,$ since otherwise, the clause would be unsatisfied on $L,$ a contradiction. 

For both claims, due to Proposition~\ref{prop:rooted models exist}, we can assume that $M$ is rooted at $w$. From Proposition~\ref{prop:tree like model exists with bounded morphism}, we know that there is a model $T_0$ which is a strict tree, and a surjective bounded morphism $f\colon T_0\rightarrow M$. Let $T_{\mathrm{strict}}:=T_0$. Let $w_T$ denote the root of $T_0$. The strict tree $T_0$ is trivially $w_T$-canonical.

We define a sequence of models $\left(T_n\right)_{n\in\mathbb N}$ such that for each $n\in\mathbb N,$ it holds that when we are proving the first claim:

\begin{enumerate}
 \item $\vertices{T_n}=\vertices{T_0},$
 \item $\edges{T_{\mathrm{strict}}}\subseteq\edges{T_n}\subseteq\edges{\tl(T_{\mathrm{strict}})},$
 \item $f\colon T_n\rightarrow M$ is a homomorphism.
\end{enumerate}

In the case of the second claim, we exchange the second point with ``$T_n$ is $w_T$-canonical.''

Note that due to Corollary~\ref{corrollary:universal elementary logics have countable model property}, we can assume that the model $M$ is countable. The proof of Proposition~\ref{prop:tree like model exists with bounded morphism} from \cite{blrive01} then constructs a model $T_0$ which is also countable. Hence assume that $T_0$ is a countable model. Now let $(e_n)_{n\in\mathbb N}$ be a surjective enumeration of $\vertices{T_0}\times\vertices{T_0},$ i.e., of all possible edges in the involved trees. The construction of our model is as follows: For $n=0,$ the model $T_0$ from above obviously satisfies the conditions, since every bounded morphism is also a homomorphism. For $n\ge 0,$ we make a case distinction:

\begin{itemize}
 \item If every clause in $\hpsi$ is satisfied in $T_n,$ then let $T_{n+1}:=T_n.$
 \item Otherwise, let $\hvarphi_{n+1}$ be a clause from $\hpsi$ which is not satisfied in $T_n$ such that $\conc{\hvarphi_{n+1}}=(x_{n+1},y_{n+1})$ for variables $x_{n+1},y_{n+1},$ and let $\alpha_{n+1}\colon\preq{\hvarphi_{n+1}}\cup\set{x_{n+1},y_{n+1}}\rightarrow T_n$ be a homomorphism such that $(\alpha_{n+1}(x_{n+1}),\alpha_{n+1}(y_{n+1}))$ is not an edge in $T_n$. This must exist by Proposition~\ref{prop:homomorphism and horn clauses}, since by the above, the case $\conc{\hvarphi_{n+1}}=\emptyset$ cannot occur. Choose $\hvarphi_{n+1},\alpha_{n+1}$ in such a way that the pair $(\alpha_{n+1}(x_{n+1}),\alpha_{n+1}(y_{n+1}))$ has a minimal index in the sequence $(e_n)$. Now let $T_{n+1}$ be defined with vertex set $\vertices{T_n},$ edge set $\edges{T_n}\cup\set{(\alpha_{n+1}(x_{n+1}),\alpha_{n+1}(y_{n+1}))},$ and the same propositional assignments as $T_n.$
\end{itemize}

We show that the construction satisfies the requirements 1-3. The first point, which is the same for both claims, holds by definition. 

For the first claim, we prove that for each $n,$ it holds that $\edges{T_{\mathrm{strict}}}\subseteq\edges{T_n}\subseteq\edges{\tl(T_{\mathrm{strict}})}$. Since by definition, $\edges{T_n}\subseteq\edges{T_{n+1}},$ we know that for all $n,$ it holds that $\edges{T_{\mathrm{strict}}}=\edges{T_0}\subseteq\edges{T_n},$ and we also know that $\edges{T_{\mathrm{strict}}}=\edges{T_0}\subseteq\edges{\tl(T_{\mathrm{strict}})}.$

Hence assume that there is some minimal $n$ such that $\edges{T_n}\nsubseteq\edges{\tl(T_{\mathrm{strict}})}$. Due to the minimality of $n,$ and since the claim holds for $n=0,$ we know that the edge which is not present in $\edges{\tl(T_{\mathrm{strict}})}$ is the edge $(\alpha_n(x_n),\alpha_n(y_n))$. By definition, $\alpha_n\colon\preq{\hvarphi_n}\cup\set{x_{n},y_{n}}\rightarrow T_{n-1}$ is a homomorphism. Since due to minimality of $n,$ we know that $\edges{T_{n-1}}\subseteq\edges{\tl(T_{\mathrm{strict}})},$ this implies that $\alpha_n\colon\preq{\hvarphi_n}\cup\set{x_{n},y_{n}}\rightarrow \tl(T_{\mathrm{strict}})$ is a homomorphism as well. Since $\hvarphi$ is a clause in $\hpsi,$ and $\tl(T_{\mathrm{strict}})\models\hpsi,$ we know that $\tl(T_{\mathrm{strict}})\models\hvarphi,$ and with Proposition~\ref{prop:homomorphism and horn clauses}, we conclude that $(\alpha_n(x_n),\alpha_n(y_n))$ is an edge in $\tl(T_{\mathrm{strict}}),$ a contradiction. Therefore we know that $\edges{T_{\mathrm{strict}}}\subseteq\edges{T_n}\subseteq\edges{\tl(T_{\mathrm{strict}})}$ for all $n,$ and hence we have proven the second point in the case of the first claim.

For the second claim, we need to show that $T_n$ is a $w_T$-canonical graph. For $i\in\mathbb{N},$ let $L_i$ denote the nodes in the $i$-th level of $T_{n-1},$ i.e., the set $\set{v\in T_{n-1}\ \vert \ T_{n-1}\models\path {w_T}iv}$. In order to prove that $T_n$ is $w_T$-canonical, since $T_{n-1}$ is, it suffices to prove that the edge $(\alpha_n(x_n),\alpha_n(y_n))$ in the step from $T_{n-1}$ to $T_{n}$ does not destroy the property of being canonical, i.e., we need to show that $\alpha_n(x_n)\in L_i$ and $\alpha_n(y_n)\in L_{i+1}$ for some $i$. Let $L$ be the homomorphical image of $T_{n-1}$ as a strict line via the homomorphism $\beta$ (since $T_{n-1}$ is $w$-canonical, this exists and is unique). Since $\alpha_n\colon\preq{\hvarphi_n}\cup\set{x_{n},y_{n}}\rightarrow T_{n-1}$ is a homomorphism, we know that $\beta\circ\alpha_n\colon\preq{\hvarphi_n}\cup\set{x_{n},y_{n}}\rightarrow L$ is a homomorphism as well. Since $\hpsi$ is satisfied on every strict line, this is also true for $\hvarphi,$ and hence due to Proposition~\ref{prop:homomorphism and horn clauses}, we know that $(\beta\circ\alpha_n(x_n),\beta\circ\alpha_n(y_n))$ is an edge in $L$. Therefore, $\alpha_n(x_n)$ is exactly one level above $\alpha_n(y_n)$ in $T_{n-1},$ as required.

For both claims, we now show that $f\colon T_n\rightarrow M$ is a homomorphism for all $n$. Again, we prove the fact by induction, and for $n=0,$ this holds due to the choice of $T_0,$ since every bounded morphism is a homomorphism. Therefore, let the claim hold for $n,$ and let $(u,v)$ be an edge in $T_{n+1}$. If $(u,v)$ is not the edge $(\alpha_{n+1}(x_{n+1}),\alpha_{n+1}(y_{n+1})),$ then we know that $(u,v)$ is an edge in $T_n$ as well, and since due to induction hypothesis, we know that $f\colon T_n\rightarrow M$ is a homomorphism, it follows that $(f(u),f(v))$ is an edge in $M$. Therefore assume that $u=\alpha_{n+1}(x_{n+1})$ and $v=\alpha_{n+1}(y_{n+1})$. We need to show that $(f\circ\alpha_{n+1}(x_{n+1}),f\circ\alpha_{n+1}(y_{n+1}))$ is an edge in $M$. By construction, we know that $\alpha_{n+1}\colon\preq{\hvarphi_{n+1}}\cup\set{x_{n+1},y_{n+1}}\rightarrow T_n$ is a homomorphism. Since by induction hypothesis, we know that $f\colon T_n\rightarrow M$ is a homomorphism, it follows that $f\circ\alpha_{n+1}\colon\preq{\hvarphi_{n+1}}\cup\set{x_{n+1},y_{n+1}}\rightarrow M$ is a homomorphism as well. Since $M$ is a $\K\hpsi$-model and $\hvarphi_{n+1}$ is a clause in $\hpsi,$ we know that $M$ satisfies $\hvarphi_{n+1},$ and thus by Proposition~\ref{prop:homomorphism and horn clauses}, we know that $(f\circ\alpha_{n+1}(x_{n+1}),f\circ\alpha_{n+1}(y_{n+1}))$ is an edge in $M,$ as required.

We now construct the desired $\K\hpsi$-model as follows: Define $T_{\infty}$ as having $\vertices{T_{\infty}}=\vertices{T_0}$ and $\edges{T_{\infty}}=\cup_{n\in\mathbb{N}}\edges{T_n}$. We show that $T_{\infty}$ is a $\K\hpsi$-model, and that $f\colon T_{\infty}\rightarrow M$ is a bounded morphism. 

Assume that $T_{\infty}$ is not a $\K\hpsi$-model. By construction, we know that in this case, there is no $n$ such that $T_n=T_{n+1},$ and therefore each $T_{n+1}$ has exactly one additional edge in comparison to $T_n$. Hence we can define a sequence $(f_n)_{n\in\mathbb N}$ such that $f_n=i\in\mathbb N$ iff $\edges{T_{n+1}}=\edges{T_n}\cup\set{e_i}$. Then $f_n$ is a sequence of pairwise different natural numbers. Since $\hpsi$ is not satisfied in $T_{\infty},$ there is a clause $\hvarphi$ from $\hpsi$ which is not satisfied in $T_{\infty}$. For both claims, we therefore know by the above that $\conc\hvarphi=(x,y)$ for variables $x,y$ with $x=y$ or $x,y\in\preq\hvarphi$. By Proposition~\ref{prop:homomorphism and horn clauses}, we know that there is a homomorphism $\alpha\colon\preq\hvarphi\cup\set{x,y}\rightarrow M$ such that $(\alpha(x),\alpha(y))$ is not an edge in $T_{\infty}$. Let $j\in\mathbb N$ such that $(\alpha(x),\alpha(y))=e_j$. Since in the sequence $f_n,$ no number is repeated, there is some natural number $n_0$ such that $f_n>j$ for all $n\ge n_0$. Since $\preq\hvarphi$ is a finite graph, and $\edges{T_n}\subseteq\edges{T_{n+1}},$ and $\edges{T_{\infty}}=\cup_{n\in\mathbb{N}}\edges{T_n},$ we know that there is some $n_1\in\mathbb{N}$ such that $\alpha\colon\preq\hvarphi\cup\set{x,y}\rightarrow T_{n_1}$ is a homomorphism. Then $\alpha\colon\preq\hvarphi\cup\set{x,y}\rightarrow T_n$ is also a homomorphism for all $n\ge n_1,$ since every edge present in $T_{n_1}$ is also present in every $T_n$ for $n\ge n_1$. Let $n:=\max(n_0,n_1)$. Then in the step from $T_n$ to $T_{n+1},$ the edge $e_{f_n}$ was added, and by choice of $n$ we know that $f_n>j$. Since $e_j$ is not an edge in $T_{\infty},$ we know that $e_j$ is also not an edge in $T_n$. Recall that $\alpha_{n+1},\hvarphi_{n+1}$ are chosen in such a way that the edge $(\alpha_{n+1}(x),\alpha_{n+1}(y))$ has minimal index in the sequence $(e_n)_{n\in}\mathbb N,$ this is a contradiction, since the edge $e_j$ is an edge with a smaller index than $e_{f_j},$ and since $\alpha\colon\preq\hvarphi\cup\set{x,y}\rightarrow T_n$ is a homomorphism, $e_j$ satisfies the conditions of the edge $(\alpha_{n+1}(x),\alpha_{n+1}(y))$ in the construction of $T_{n+1}.$

It remains to show that $f\colon T_{\infty}\rightarrow M$ is a bounded morphism. Property $(i)$ holds by construction, since we do not change propositional assignments, and $f\colon T_0\rightarrow M$ is a bounded morphism. We now show that $f$ is a homomorphism. Hence let $(u,v)$ be an edge in $T_{\infty}$. Since $T_{\infty}$ is the union over all $T_n,$ there is some $n\in\mathbb{N}$ such that $(u,v)$ is an edge in $T_n$. Since by the above, $f\colon T_n\rightarrow M$ is a homomorphism, it follows that $(f(u),f(v))$ is an edge in $M,$ as claimed. For property $(iii),$ let $u\in T_{\infty}$ and $v'$ in $M$ such that $(f(u),v')$ is an edge in $M$. Since $f\colon T_0\rightarrow M$ is a bounded morphism, we know that there is some $v\in T_0$ such that $(u,v)$ is an edge in $T_0,$ and $f(v)=v'$. By construction, $(u,v)$ is also an edge in $T_{\infty},$ and therefore $f\colon T_{\infty}\rightarrow M$ is a bounded morphism. By choice of $f,$ and since $\vertices{T_0}=\vertices{T_{\infty}},$ $f$ is also surjective. Hence there is some $w'\in T_{\infty}$ such that $f(w')=w$. Since $M,w\models\phi,$ it follows from Proposition~\ref{prop:bounded morphism modal invariance} that $T_{\infty},w'\models\phi$. By construction, for the first claim it holds that $\edges{T_{\mathrm{strict}}}\subseteq\edges{T_{\infty}}\subseteq\edges{\tl(T_{\mathrm{strict}}}),$ since this holds for the individual $T_n$ and $T_{\infty}$ is the union over the edges of all $T_n,$ and for the second claim, $T_{\infty}$ is obviously $w_T$-canonical: Assume that it is not, then there is a node $x\in T_{\infty}$ and natural numbers $i\neq j$ such that $T_{\infty}\models\path {w_T}ix$ and $T_{\infty}\models\path {w_T}jx$. Since only a finite number of edges is relevant for this path, there exists some $n$ such that $T_n\models\path {w_T}ix$ and $T_n\models\path {w_T}jx$. This is a contradiction, since $T_n$ is $w_T$-canonical.

Note that we can assume that the node $w'$ is the root of $T_{\infty},$ since due to Proposition~\ref{prop:rooted models exist}, we can assume that every node in $T_{\infty}$ can be reached from $w'.$
\end{proof}

\subsection{\PSPACE\ upper complexity bounds}\label{sect:pspace for non-transitive}

In the previous section, we showed that for logics defined by universal Horn formulas, satisfiable formulas are always satisfiable in a tree-like model. Tree-like models are the main argument in many proofs showing \PSPACE-membership for satisfiability problems in modal logic. We now show that these models indeed allow us to construct \PSPACE\ algorithms for a wide class of modal logics.

Note that while the construction in the proof of the following theorem has similarities to the constructions by Ladner in~\cite{lad77} or by Halpern and Moses in~\cite{hamo92}, the focus of our result is different. Many proofs of previous \PSPACE-algorithms also gave proofs of a variant of some tree-like model property. Our proof relies on this property (which we already proved for our logics in Theorem~\ref{theorem:tree-like models for horn logics}), and as a consequence, the verification that the algorithm works correctly with respect to the modal aspect of its task is very easy to verify. The main work of the proof is to prove that the algorithm handles the first-order part of the satisfiability problem correctly, i.e., that the model it constructs is in fact a model satisfying the first-order formula $\hpsi$ defining the logic. Whereas this is easy for standard classes of frames (checking reflexivity, symmetry, transitivity etc is straightforward), in the general case that we cover here this requires most of the work.

The main feature of the logics that we use here is that of \emph{locality}: The proof makes extensive use of the fact that in order to verify that the first-order clauses are satisfied it is sufficient to consider local parts of the model constructed by the algorithm. This is the main reason why we believe that this proof does not easily generalize to cases where a variant of transitivity is among the conditions implied by the first-order formula $\hpsi.$

\begin{theorem}\label{theorem:pspace membership for non-trans logics}
 Let $\hpsi$ be a universal Horn formula such that \algname\ does not add any element of the form $\trans^k$ to \tl\ on input $\hpsi$. Then $\sat{\K\hpsi}\in\PSPACE.$
\end{theorem}

\textit{Proof.}
 Since $\NP\subseteq\PSPACE$ we can assume that $\sat{\K\hpsi}\notin\NP$. Let $\tl$ be as determined by \algname\ on input $\hpsi$. From the prerequisites, we know that $\tl\subseteq\set{\refl,\symm}$. From Theorem~\ref{theorem:tree-like models for horn logics}, we know that for every $\K\hpsi$-satisfiable modal formula $\phi,$ there is a model $T$ of $\phi$ such that $T$ is an edge-extension of a strict tree $T_{\mathrm{strict}}$, and every edge present in $T$ which is not an edge of $T_{\mathrm{strict}}$ is a reflexive or a symmetric edge.

By the proof of Lemma~\ref{lemma:few nodes with few predecessors}, we can assume that in the tree $T_{\mathrm{strict}},$ every node has at most $\card{\subf\phi}$ successors.

The strategy of the \PSPACE-algorithm is as follows: We nondeterministically guess the model $T$ and verify that it is a model of both $\phi$ and of $\hpsi$ by performing a depth-first-search. It is straightforward to see that in this way, we can verify that the modal formula $\phi$ holds in the model. To also check if $\hpsi$ is satisfied requires a bit more effort: Even though we know that $\hpsi$ only adds reflexive or symmetric edges in addition to those present in $T_{\mathrm{strict}},$ we need to be careful which edges are required and which are not (there may very well be formulas which are not satisfiable on a symmetric and reflexive tree, but are $\K\hpsi$-satisfiable). The main reason why we can perform these tests is that the properties that we work with have a ``local character:'' To check if an edge is required between some nodes $u$ and $v$ in $T,$ we need to know whether there is a clause $\hvarphi$ in $\hpsi$ such that $\hvarphi$ can be homomorphically mapped into $T$ in such a way that the conclusion edge is mapped to the pair $(u,v)$. Since we are working with a tree only extended with reflexive and symmetric edges, we know that a homomorphic image of a connected component of some $\preq\hvarphi$ including $u$ and $v$ contains only vertices which are ``near'' to both $u$ and $v$. Therefore we can verify that these clauses are satisfied by procedures looking only locally at the model $T$. 

There are two main obstacles to this approach: For once, the clause $\hvarphi$ that requires $(u,v)$ to be an edge might very well contain more than one connected component except that one containing the conclusion edge (we can ignore the case where there is no conclusion edge, or the vertices from the conclusion edge are unconnected in $\preq\hvarphi,$ since if such a clause can be applied, i.e., homomorphically mapped to $T,$ then it can also be homomorphically mapped to the $\tl$-closure of $T,$ and hence to a $\tl$-tree, in which case \algname\ reports \NP-membership, which is a contradiction to our assumption $\sat{\K\hpsi}\notin\NP$). The other obstacle is that although we only need to look at vertices in the ``neighborhood'' of the current vertex to check that it has all the right edges coming in and out, we need to ensure that all the vertices that we looked at ``locally'' are consistent, when we revisit a part of the model which is close to a node that we already considered.

The ways to deal with these obstacles is the following: For the first problem, we simply keep a list of connected components of $\preq\hvarphi$-graphs, and at the beginning of the algorithm, guess for each one if it will appear as a homomorphic image in the tree (which of course later we need to verify). For the second problem, we keep more nodes in storage than just the ones in the neighborhood of the one we are currently visiting, but only a polynomial number.

Strictly speaking, the algorithm does not operate on a model, but on an ``annotated model.'' The annotation of a world is the set of subformulas and negated subformulas of the input formula $\phi$ which are true at this world, and are required to be true to ensure that the formula $\phi$ is true at the root-world.

By Proposition~\ref{prop:rooted models exist}, we can assume that the tree $T$ has height of at most $\md\phi$. For a node $v$ in the $i$-th level of $T,$ let $\annot v$ denote the set of subformulas and negated subformulas of $\phi$ which have a modal depth of at most $\md\phi-i$. These are exactly those formulas for which we need to know that they hold at $v$ in order to verify that the input formula $\phi$ holds at the root of $T.$

We now describe the decision procedure, which is a nondeterministic \PSPACE-algorithm. Let $S$ be the cardinality of the largest connected component in any of the graphs $\preq\hvarphi$ for clauses $\hvarphi$ of $\hpsi$. Note that this number only depends on $\hpsi,$ and therefore can be regarded as constant. The algorithm as stated in Figure~\ref{fig:pspace algorithm} does not work in polynomial space, since it guesses and stores the possibly exponentially-sized model $T$. We will first show that the algorithm as stated is correct and then prove how it can be implemented using only polynomial space, by only storing a currently relevant subset of the model $T.$

For the description of the algorithm, we will call a node $v\in T$ \emph{back-symmetric} if there is an edge $(v,u),$ where $u$ is the predecessor of $v$ in $T_{\mathrm{strict}}$. 

\begin{figure}
\begin{algorithmic}
 \STATE{For connected components $C_i$ of all $\preq\hvarphi,$ guess if it appears as homomorphic image in $T$}
 \STATE{Guess the model $T$}
 \STATE{Verify that $\phi\in\annot{0}$}
 \STATE{$\mathsf{current}:=w$}
 \WHILE{$w$ not marked $\mathsf{done}$}
    \STATE{Let $\mathsf{preq}$ be the predecessor of $\mathsf{current}$ (if $\mathsf{current}\neq w$)}
    \STATE{$\verifyhorn{\mathsf{current}}$}
    \IF{There is $\Diamond\chi\in\annot{\mathsf{current}}$ not marked $\mathsf{done}$}
       \IF{$\mathsf{current}$ is reflexive and $\chi\in\annot{\mathsf{current}}$}
         \STATE{Mark $\Diamond\chi$ $\mathsf{done}$ in $\annot{\mathsf{current}}$}
       \ENDIF
       \IF{$\mathsf{current}$ is back-symmetric and $\chi\in\annot{\mathsf{preq}}$}
         \STATE{Mark $\chi$ $\mathsf{done}$}
       \ENDIF
       \STATE{Let $\mathsf{next}$ be next unvisited successor of $\mathsf{current}$}
       \STATE{Verify that $\chi\in\annot{\mathsf{next}}$}
       \STATE{$\mathsf{current}:=\mathsf{next}$}
    \ELSE
       \STATE{$\verifycons{\mathsf{current}}$}
       \IF{$\mathsf{current}$ is reflexive}
          \STATE{Verify that $\annot{\mathsf{current}}$ does not contain $\chi$ and $\neg\Diamond\chi$ for any $\chi$}
       \ENDIF
       \IF{$\mathsf{current}$ is back-symmetric}
          \STATE{Verify that $\annot{\mathsf{current}}$ does not contain $\neg\Diamond\chi$ for $\chi\in\annot{\mathsf{preq}}$}
       \ENDIF
       \STATE{Verify that $\annot{\mathsf{preq}}$ does not contain $\neg\Diamond\chi$ for some $\chi\in\annot{\mathsf{current}}$}
       \STATE{Mark $\mathsf{current}$ as $\mathsf{done}$}
       \STATE{In $\mathsf{preq},$ mark $\Diamond\chi$ done for all $\chi\in\annot{\mathsf{current}}$}
       \STATE{$\mathsf{current}:=\mathsf{preq}$}
    \ENDIF
 \ENDWHILE
 \STATE{Accept}
 \end{algorithmic}
\caption{Algorithm $\algnametwo$}
\label{fig:pspace algorithm}
\end{figure}

When the algorithm guesses the model $T,$ it additionally guesses the set $\annot v$ for every node $v$ in $T,$ and for each node it guesses if it is back-symmetric and if it is reflexive.

The procedure $\verifycons v$ performs the following check: For a node $v$ on the $i$-th level of $T,$ $\annot{i}$ is required to contain all subformulas and negated subformulas of $\phi$ which have modal depth of at most $\md\phi-i,$ and are true at $v$. Hence, $\annot{\mathsf{current}}$ must contain exactly one of $\neg\chi$ or $\chi$ for each relevant $\chi,$ and additionally, if $\chi_1\wedge\chi_2\in\annot{\mathsf{current}},$ then both $\chi_1$ and $\chi_2$ need to be members as well. Similarly, if $\chi_1\vee\chi_2$ is a member, then at least one of them must be an element of $\annot{\mathsf{current}}.$

The procedure $\verifyhorn v$ works as follows: If there is a clause $\hvarphi$ in $\hpsi$ with $\conc\hvarphi=(x,x)$ for $x\notin\preq\hvarphi$ such that all connected components of $\preq\hvarphi$ can be mapped homomorphically into $T,$ then $\verifyhorn v$ ensures that $v$ is reflexive. Note that all other Horn clauses in $\hpsi$ satisfy that $\conc\hvarphi=(x,y)$ for some $x,y\in\preq\hvarphi$. For these clauses, the procedure considers the subgraph $G_v$ consisting of all nodes of $T$ which can be reached from $v$ in at most $S$ undirected steps (note that a node can be reached in at most $S$ steps in $T$ if and only if it can be reached in at most $S$ steps in $T_{\mathrm{strict}}$). For every connected component $C$ of $\preq\hvarphi,$ $\verifyhorn{v}$ tests all functions $\alpha\colon C\rightarrow G_v$. If one of these $\alpha$ is a homomorphism, then $\verifyhorn v$ rejects, if the algorithm guessed in the beginning that $C$ cannot be mapped homomorphically into $T$. If there is one clause $\hvarphi$ in $\hpsi$ such that all connected components of $\preq\hvarphi$ can be mapped into $T$ (according to the list of these possibilities maintained by the algorithm) and $\verifyhorn v$ detected a homomorphism $\alpha\colon C_{\preq\hvarphi}\rightarrow G_v$ (where $C_{\preq\hvarphi}$ is the connected component of $\preq\hvarphi$ containing the nodes from the conclusion edge of $\hvarphi$) for some $\hvarphi$ with $\conc{\hvarphi}=(x,y)$ such that $v\in\set{\alpha(x),\alpha(y)},$ then $\verifyhorn v$ rejects if $(\alpha(x),\alpha(y))$ is not an edge in $G_v.$

We prove that the algorithm is correct. First note that for each connected component $C$ of some $\preq\hvarphi$ for a clause $\hvarphi$ in $\hpsi,$ if there is a homomorphism $\alpha\colon C\rightarrow T,$ then there is a node $v\in T$ such that $\alpha\colon C\rightarrow G_v$ is a homomorphism. This holds because all edges in $T$ are already present in the strict tree $T_{\mathrm{strict}}$ or are symmetric or reflexive edges, and the homomorphic image of $C$ under the homomorphism $\alpha$ is a connected component of $T,$ and the maximal distance of the nodes in this image is $S$ (recall that this is the maximal cardinality of a connected component in any $\preq\hvarphi$). Therefore, we can assume that for each connected component $C$ of some $\preq\hvarphi,$ if it can be homomorphically mapped into $T,$ then \algnametwo\ guessed this correctly in the beginning in every accepting run of the algorithm (an incorrect guess would, due to the observation just made, be detected by $\verifyhorn v$ for some node $v$).

 Now assume that the algorithm accepts. We claim that the model obtained from the annotated model guessed by the algorithm where a variable $x$ is true at a world $v$ if and only if $x\in\annot{v}$ is a model of both the modal formula $\phi$ (at the root-world $w$) and of the Horn formula $\hpsi$. By the checks the algorithm performs, it can easily be verified by induction on the level of the nodes (corresponding to the modal depth of the involved formulas) that for every world $v$ in $T,$ every formula in $\annot{v}$ is satisfied at $v$. Since $\phi\in\annot{w},$ this implies that $T,w\models\phi$. The base case for the induction is clear, since for worlds $v$ in the level $\md\phi,$ $\annot v$ only contains literals, and the algorithm ensures that $\annot v$ is propositionally consistent. Since for each subformula $\chi$ of $\phi$ of relevant modal depth, $\annot v$ contains exactly one of $\chi$ and $\neg\chi,$ the induction hypothesis can be applied in the relevant cases. Note that the algorithm also checks consistency for the cases in which we have reflexive and/or symmetric edges.

It remains to show that $T$ is also a model of the first-order formula $\hpsi$. Assume that this is not the case. Then there exists some clause $\hvarphi$ in $\hpsi$ which is not satisfied in $T$. In particular, this implies that $\preq\hvarphi$ can be homomorphically mapped into $T,$ and since $\edges{T}\subseteq\edges{\tl(T)},$ this implies that $\preq\hvarphi$ can be homomorphically mapped into a \tl-tree. Since we assumed that $\algname$ does not return $\NP$ on input $\hpsi,$ we know that none of the \NP-conditions from \algname\ are satisfied. Since none of the \NP-conditions from \algname\ occur, we know that $\hvarphi$ has conclusion edge $\conc\hvarphi=(x,y)$ for variables $x,y$ where $x=y$ or $x,y\in\preq\hvarphi$. First assume that $x=y$. Since $\preq\hvarphi$ can be mapped homomorphically into $T,$ the procedure $\verifyhorn .$ required every node in $T$ to be reflexive, hence $\hvarphi$ is satisfied in $T$. Not assume that $x,y\in\preq\hvarphi$. Since none of the \NP-conditions apply, we know that $x$ and $y$ are connected with an undirected path in $\preq\hvarphi$. In particular, they lie in the same connected component $C_{\preq\hvarphi}$ of $\preq\hvarphi$. Since $\hvarphi$ is not satisfied in $T,$ this implies by Proposition~\ref{prop:homomorphism and horn clauses} that there are nodes $u,v\in T$ such that $(u,v)$ is not an edge in $T,$ there is no edge $(u,v)$ in $T,$ and there is a homomorphism $\alpha\colon\preq\hvarphi\rightarrow T$ such that $\alpha(x)=u,$ and $\alpha(y)=v$. Due to the above, we know that $\alpha\colon C_{\preq\hvarphi}\rightarrow G_v$ is a homomorphism. Therefore, the homomorphism $\alpha$ was found by $\verifyhorn{u}$ and $\verifyhorn{v}$. Since $\preq\hvarphi$ can be homomorphically mapped into $T,$ we know that every connected component $C$ of $\preq\hvarphi$ can be homomorphically mapped into $T,$ and due to the above, we know that the \algnametwo\ guessed this correctly in an accepting run of the algorithm. Therefore, the procedure $\verifyhorn .$ ensured that $(u,v)$ is an edge in $T,$ a contradiction.

Now assume that $\phi$ is $\K\hpsi$-satisfiable. Due to the remarks at the beginning of the proof, we know that in this case, there exists a $\K\hpsi$-model $T$ such that $T$ is an edge-extension of a strict tree, and $\edges T\subseteq\edges{\tl(T)}$. Therefore the algorithm can guess this model and verify that it satisfies both $\phi$ and $\hpsi.$

It remains to prove that the algorithm can be implemented in nondeterministic polynomial space. The result then follows, since due to a classic result by Savitch~\cite{sav73}, $\NPSPACE=\PSPACE$. In order to implement the algorithm using only polynomial space, the main change needed compared to the version stated in Figure~\ref{fig:pspace algorithm} is how much of the guessed model $T$ is stored in memory at a given time.

The \NPSPACE-implementation does not guess the entire model $T$ at the start of the algorithm, but guesses each node the moment it is first accessed (either by being created explicitly, or by being explored as an $S$-step neighbor of another node by the procedure $\verifyhorn .$). It removes the node from memory at a time when it will not be accessed anymore in the remaining execution of the algorithm.

To be precise, the algorithm at all times keeps in its memory the node $\mathsf{current}$ and all of its predecessors, and all nodes which can be reached from these in at most $S$ steps in the tree $T_{\mathrm{strict}}$. Since in $T,$ every node has at most $\card{\subf{\phi}}$ successors in the next level, and $S$ is a constant, this is a polynomial number of nodes. 

We now need to prove that no necessary information is removed from memory, i.e., that no node is first created, then deleted and then accessed again. Note that from the construction of the algorithm, it is obvious that new nodes are visited in a depth-first order.

Therefore assume that this happens for some node $\mathsf{node}$. Note that any node which gets deleted from memory is not reachable from the root world $w$ in at most $S$ steps, and therefore $\mathsf{node}$ is at some level $i>S$ in the tree $T_{\mathrm{stricŧ}}$. Let $v_1$ be the node for which $\mathsf{node}$ was visited for the first time, i.e., the first node visited such that $\mathsf{node}\in G_{v_1}$ (recall that $G_{v_1}$ is the set of nodes which can be reached from $v_1$ in at most $S$ undirected steps in $T$). Since $\mathsf{node}$ is deleted from memory and required again later, there is some node $v_2$ such that $\mathsf{node}$ cannot be reached from any predecessor of $v_2$ in at most $S$ steps, and a node $v_3$ such that $\mathsf{node}$ can be reached from $v_3$ in at most $S$ steps. Let $a$ be the (uniquely determined) common predecessor of $\mathsf{node}$ and $v_2$ with a maximal level in the tree. Then, since $a$ is a predecessor of $v_2,$ we know that $\mathsf{node}\notin G_a$. Hence, $\mathsf{node}$ is at least $S$ levels below $a$. Since $T_{\mathrm{strict}}$ is a tree, any node $t$ such that $\mathsf{node}\in G_t$ must therefore be a successor of $a$. In particular, $v_3$ is a successor of $a$. This is a contradiction, because \algnametwo\ traverses the tree in depth-first-search, and hence does not leave the sub-tree with root $a$ and re-enters it later.

Therefore we have shown that it is sufficient to keep a polynomial number of nodes in storage, and which nodes to keep can be decided by an easy pattern. Hence it follows that the algorithm can indeed by implemented in nondeterministic polynomial space as required, concluding the proof.
\hfill$\Box$\bigskip

\subsection{Applications}\label{sect:applications}

Theorem~\ref{theorem:horn conjunction classification} and \ref{theorem:pspace membership for non-trans logics} can be used to classify the complexity of a lot of concrete logics, but they also imply more general results, for which we will give two examples. For once, recall that Ladner proved that all normal modal logics \KL\ such that $\logicname{S4}$ (the logic over all transitive and reflexive frames) is an extension of \KL\ give rise to a \PSPACE-hard satisfiability problem. The following corollary shows that this result is optimal in the sense that every universal Horn logic which is a ``proper extension'' of \logicname{S4}\ in the way that they imply the conditions of \logicname{S4}, already gives an \NP-solvable satisfiability problem.

\begin{corollary}\label{corollary:ladner hardness result optimal}
 Let $\hpsi$ be a universal Horn formula such that $\hpsi$ implies $\hvarphi_{\mathrm{refl}}\wedge\hvarphi_{\mathrm{trans}}$. Then either $\K\hpsi=\logicname{S4},$ or $\K\hpsi$ has the polynomial-size model property and $\sat{\K\hpsi}\in\NP.$
\end{corollary}

\begin{proof}
 By the prerequisites, we know that $\hpsi$ is equivalent to $\hpsi\wedge\hvarphi_{\mathrm{refl}}\wedge\hvarphi_{\mathrm{trans}}$. Hence we can, without loss of generality, assume that $\hvarphi_{\mathrm{refl}}$ and $\hvarphi_{\mathrm{trans}}$ appear as clauses in $\hpsi.$

If every clause in $\hpsi$ is satisfied in every transitive and reflexive tree, then every modal formula $\phi$ which is satisfiable in a transitive and reflexive tree is $\K\hpsi$-satisfiable. Note that a special case of Theorem~\ref{theorem:tree-like models for horn logics} gives the result that every $\logicname{S4}$-satisfiable formula also is satisfiable in a reflexive and transitive tree. Therefore, every $\logicname{S4}$-satisfiable formula is also $\K\hpsi$-satisfiable, and hence, every $\K\hpsi$-validity is also $\logicname{S4}$-valid. Therefore, $\logicname{S4}$ is an extension of $\K\hpsi$. Since by Proposition~\ref{proposition:extensions of elementary logics and first-order implication}, $\K\hpsi$ is an extension of $\logicname{S4}=\K{\hvarphi_{\mathrm{refl}}\wedge\hvarphi_{\mathrm{trans}}},$ this implies that $\K\hpsi=\logicname{S4}.$

Therefore, we can assume that $\hpsi$ is not satisfied in every reflexive and transitive tree. Now let $\tl$ be as determined by \algname\ on input $\hpsi$. Since $\hvarphi_{\mathrm{refl}}$ and $\hvarphi_{\mathrm{trans}}$ are clauses in $\hpsi,$ we know that $\refl$ and $\trans^2$ are elements of $\tl$. If all elements in $\tl$ are of the form $\refl$ or $\trans^k,$ then we know (since $k$-transitivity is implied by $2$-transitivity), since $\hpsi$ is satisfied on every \tl-tree, that $\hpsi$ is satisfied in every reflexive and transitive tree, a contradiction. Therefore, we now that $\symm\in\tl,$ and hence by construction, \algname\ reports \NP-membership. Since by Theorem~\ref{theorem:horn conjunction classification} the output of the algorithm is correct, we know that $\sat{\K\hpsi}\in\NP,$ and $\K\hpsi$ has the polynomial-size model property, as claimed. 
\end{proof}

We further can show a \PSPACE\ upper bound for all universal Horn logics which are extensions of the logic \logicname{T}, and hence, from Theorem~\ref{theorem:horn conjunction classification}, conclude that these are all either solvable in \NP\ (and thus \NP-complete if they are consistent), or \PSPACE-complete.

\begin{corollary}\label{corollary:above reflexive gives pspace}
 Let $\hpsi$ be a universal Horn formula such that $\hpsi$ implies $\hvarphi_{\mathrm{refl}}$. Then $\sat{\K\hpsi}\in\PSPACE.$
\end{corollary}

\begin{proof}
 Assume without loss of generality that \algname\ determines the logic $\K\hpsi$ to have a \PSPACE-hard satisfiability problem, otherwise the theorem holds trivially, since $\NP\subseteq\PSPACE$. If $\hpsi$ implies $\hvarphi_{\mathrm{trans}},$ then the result follows from Corollary~\ref{corollary:ladner hardness result optimal}. Hence assume that this is not the case. Note that in reflexive graphs, $k$-transitivity is equivalent to transitivity. Also note that the conditions requiring a node to have a certain depth or height in a graph are always satisfied in a reflexive graph, because nodes here have infinite depth and height. Therefore, if $\hpsi$ implies a formula of the form $\genvarphi 0kpqr$ for some $2\leq k$ and some $p,q,r\in\mathbb{N},$ then $\hpsi$ also implies $\varphi_{\mathrm{trans}},$ and due to the above, we can assume that this is not the case. Thus, $\tl$ as determined by \algname\ contains no condition of the form $\trans^k$ for any $k\in\mathbb{N}$. The complexity result now follows from Theorem~\ref{theorem:pspace membership for non-trans logics}. 
\end{proof}

In a similar way, we can prove that all universal Horn logics which imply a variant of symmetry give rise to a satisfiability problem in \PSPACE. A noteworthy difference in the prerequisites of Corollary~\ref{corollary:above reflexive gives pspace} and Corollary~\ref{corollary:pspace for extensions of symm} is that the former requires the reflexivity condition to be implied by the formula $\psi,$ while the latter only needs a ``near-symmetry''-condition as detected by \algname.

\begin{corollary}\label{corollary:pspace for extensions of symm}
 Let $\hpsi$ be a universal Horn formula such that $\algname$ adds $\symm$ to $\tl$ on input $\hpsi$. Then $\sat{\K\hpsi}\in\PSPACE$. In particular, any universal Horn logic which is an extension of $\logicname{B}$ has a satisfiability problem solvable in \PSPACE. 
\end{corollary}

\begin{proof}
 If $\sat{\K\hpsi}\in\NP,$ the claim trivially holds. Hence, since $\algname$ is correct due to Theorem~\ref{theorem:horn conjunction classification}, we can assume that $\algname$ returns \PSPACE-hard, and $\symm\in\tl,$ where $\tl$ is as determined by \algname. Since \algname\ does not report \NP, we know from its construction that $\trans^k\notin\tl$ for all $k\in\mathbb N$. Hence the complexity result follows from Theorem~\ref{theorem:pspace membership for non-trans logics}.
\end{proof}

\section{Conclusion and Future Research}\label{sect:conclusion}

We analyzed the complexity of modal logics defined by universal Horn formulas, covering many well-known logics. We showed that the non-trivial satisfiability problems for these logics are either \NP-complete or \PSPACE-hard, and gave an easy criterion to recognize these cases. Our results directly imply that (unless $\NP=\PSPACE$) such a logic has a satisfiability problem in \NP\ if and only if it has the polynomial-size model property. We also demonstrated that a wide class of the considered logics has a satisfiability problem solvable in \PSPACE.

Open questions include determining complexity upper bounds for the satisfiability problems for all modal logics defined by universal Horn formulas. We strongly conjecture that all of these are decidable, and consider it possible that all of these problems are in \PSPACE. A successful way to establish upper complexity bounds is the guarded fragment~\cite{anbene98,gr99}. This does not seem to be applicable to our logics, since it cannot be used for transitive logics, and we obtain \PSPACE-upper bounds for all of our logics except those involving a variant of transitivity.

The next major open challenges are generalizing our results to formulas not in the Horn class, and allowing arbitrary quantification. Initial results show that even when considering only universal formulas over the frame language, undecidable logics appear. An interesting enrichment of Horn clauses is to allow the equality relation. Preliminary results indicate that Corollary~\ref{corollary:horn cunjunction without algorithm} holds for this more general case as well. 

\ \\
\textbf{Acknowledgments:} We thank the anonymous referees for many hints and suggestions. The second author thanks Thomas Schneider for helpful discussionshint


\begin{thebibliography}{BdRV01}

\bibitem[AvBN98]{anbene98}
H.~Andr\'eka, J.~van Benthem, and I.~N\'emeti.
\newblock Modal languages and bounded fragments of predicate logic.
\newblock {\em Journal of Philosophical Logic}, 27:217--274, 1998.

\bibitem[BdRV01]{blrive01}
P.~Blackburn, M.~de~Rijke, and Y.~Venema.
\newblock {\em Modal logic}.
\newblock Cambridge University Press, New York, NY, USA, 2001.

\bibitem[BG04]{bega04}
B.~Bennett and A.~Galton.
\newblock A unifying semantics for time and events.
\newblock {\em Artificial Intelligence}, 153(1-2):13--48, 2004.

\bibitem[BHSS06]{bhss05b}
M.~Bauland, E.~Hemaspaandra, H.~Schnoor, and I.~Schnoor.
\newblock Generalized modal satisfiability.
\newblock In {\em Proceedings of STACS 2006}, pages 500--511,
  2006.

\bibitem[BZ05]{bazh05}
C.~Baral and Y.~Zhang.
\newblock Knowledge updates: Semantics and complexity issues.
\newblock {\em Artificial Intelligence}, 164(1-2):209--243, 2005.

\bibitem[CDF03]{codofl03}
T.~Coffey, R.~Dojen, and T.~Flanagan.
\newblock On the automated implementation of modal logics used to verify
  security protocols.
\newblock In {\em Proceedings of ISICT 2003}, pages 329--334. Trinity College
  Dublin, 2003.

\bibitem[CL94]{cheli94}
C.~Chen and I.~Lin.
\newblock The computational complexity of the satisfiability of modal horn
  clauses for modal propositional logics.
\newblock {\em Theoretical Computer Science}, 129(1):95--121, 1994.

\bibitem[FHJ02]{frhuje02}
U.~Frendrup, H{\"u}ttel, and J.~Jensen.
\newblock Modal logics for cryptographic processes.
\newblock In {\em Proceedings of EXPRESS 02}, 2002.

\bibitem[Gr{\"a}99]{gr99}
E.~Gr{\"a}del.
\newblock Why are modal logics so robustly decidable?
\newblock {\em Bulletin of the European Association for Theoretical Computer
  Science}, 68:90–103, 1999.

\bibitem[Hal95]{hal95}
J.~Halpern.
\newblock The effect of bounding the number of primitive propositions and the
  depth of nesting on the complexity of modal logic.
\newblock {\em Artificial Intelligence}, 75(2):361--372, 1995.

\bibitem[HM92]{hamo92}
J.~Halpern and Y.~Moses.
\newblock A guide to completeness and complexity for modal logics of knowledge
  and belief.
\newblock {\em Artificial Intelligence}, 54(2):319--379, 1992.

\bibitem[HMT88]{hamotu88}
J.~Halpern, Y.~Moses, and M.~Tuttle.
\newblock A knowledge-based analysis of zero knowledge.
\newblock In {\em Proceedings of STOC 1988}, pages 132--147, New York, NY, USA, 1988. ACM Press.

\bibitem[HR07]{hare07}
J.~Halpern and L.~R{\^e}go.
\newblock Characterizing the {NP-PSPACE} gap in the satisfiability problem for
  modal logic.
\newblock In {\em Proceedings of IJCAI 2007}, pages 2306--2311, 2007.

\bibitem[Lad77]{lad77}
R.~Ladner.
\newblock The computational complexity of provability in systems of modal
  propositional logic.
\newblock {\em SIAM Journal on Computing}, 6(3):467--480, 1977.

\bibitem[LR86]{lare86}
R.~Ladner and J.~Reif.
\newblock The logic of distributed protocols: Preliminary report.
\newblock In {\em Proceedings of TARK 1986}, pages 207--222, San Francisco, CA,
  USA, 1986. Morgan Kaufmann Publishers Inc.

\bibitem[Ngu05]{ngu05}
L.~Nguyen.
\newblock On the complexity of fragments of modal logics.
\newblock {\em Advances in Modal Logic - Volume 5}, pages 249--268.
  King's College Publications, 2005.

\bibitem[Sah73]{sahl73}
H.~Sahlqvist.
\newblock Completeness and correspondence in the first and second order
  semantics for modal logic.
\newblock In {\em Proceedings of the Third Scandinavian Logic Symposium}, 1973.

\bibitem[Sav73]{sav73}
W.~J. Savitch.
\newblock Maze recognizing automata and nondeterministic tape complexity.
\newblock {\em Journal of Computer and Systems Sciences}, 7:389--403, 1973.

\bibitem[SC85]{sicl85}
A.~Sistla and E.~Clarke.
\newblock The complexity of propositional linear temporal logics.
\newblock {\em Journal of the ACM}, 32(3):733--749, 1985.

\bibitem[SP06]{schpa06}
L.~Schr{\"o}der and D.~Pattinson.
\newblock {PSPACE} bounds for rank-1 modal logics.
\newblock In {\em Proceedings of LICS 2006}, pages 231--242, 2006.

\bibitem[Spa93]{spaan93}
E.~Spaan.
\newblock {\em Complexity of Modal Logics}.
\newblock PhD thesis, Department of Mathematics and Computer Science,
  University of Amsterdam, 1993.

\end{thebibliography}
\end{document}